\documentclass[11pt]{article}

\usepackage[utf8]{inputenc}
\usepackage[utf8]{inputenc}
\usepackage[T1]{fontenc}
\usepackage[sc]{mathpazo}
\usepackage{fullpage}
\usepackage{lineno}
\usepackage{amsthm}
\usepackage{mathtools}
\usepackage{amsmath}
\usepackage{mathrsfs}
\usepackage{amssymb}
\usepackage{algorithm}
\usepackage[noend]{algpseudocode}
\usepackage{blkarray}
\usepackage{dsfont}
\usepackage{enumitem}
\usepackage{marginnote}
\usepackage{multirow}
\usepackage{url}
\usepackage[pdftex, colorlinks,
  linkcolor=blue,
  citecolor=purple,
  filecolor=blue,urlcolor=blue]%
{hyperref}
  
\usepackage{cleveref}

\newtheorem{theorem}{Theorem}[section]
\newtheorem{lemma}[theorem]{Lemma}
\newtheorem{definition}[theorem]{Definition}
\newtheorem{proposition}[theorem]{Proposition}
\newtheorem{claim}[theorem]{Claim}
\newtheorem{corollary}[theorem]{Corollary}
\newtheorem{observation}[theorem]{Observation}
\newtheorem{example}[theorem]{Example}
\newtheorem{problem}[theorem]{Problem}
\newtheorem{question}[theorem]{Question}

\newtheorem{remark}[theorem]{Remark}

\def\cont{\mathsf{cont}}
\def\Mon{\mathsf{Mon}}

\let\gm=\gamma
\let\ld=\lambda
\let\al=\alpha
\let\vf=\varphi

\let\dl=\delta

\let\Dl=\Delta
\let\om=\omega
\def\cE{\mathcal E}
\def\cF{\mathcal F}

\def\cD{\mathcal D}
\def\cB{\mathcal B}
\def\cP{\mathcal P}

\def\cV{\mathcal V}
\def\zC{\mathbb C}
\def\zE{\mathbb E}
\def\zZ{\mathbb Z}
\def\zR{\mathbb R}
\def\zN{\mathbb N}
\def\GF{\mathsf{GF}}

\def\zzero{\text{\Large\bf0}}
\def\oa{{\overline a}}
\let\ov=\overline
\let\sse=\subseteq
\let\tm=\times
\def\vc#1#2{#1_1,\dots,#1_{#2}}
\def\zd{,\dots,}
\def\rel{R}
\def\ang#1{\langle #1\rangle}
\def\NEQ{\mathsf{NEQ}}
\def\bs{\mathbf{s}}
\def\bv{\mathbf{v}}
\def\bx{\mathbf{x}}
\def\Id{\mathsf{Id}}
\def\Alg{\mathsf{Alg}}
\def\pr{\mathrm{pr}}

\newcommand{\Field}{\mathbb{F}}
\newcommand{\Real}{\mathbb{R}}

\newcommand{\Complex}{\mathbb{C}}

\newcommand{\CSP}{\textsc{CSP}}
\newcommand{\IMP}{\textsc{IMP}}
\newcommand{\xIMP}{\mbox{$\chi$\textsc{IMP}}}
\newcommand{\Pol}{\textsf{Pol}}
\newcommand{\Inv}{\textsf{Inv}}
\newcommand{\Term}{\textsf{Term}}
\newcommand{\Sos}{\textsf{SOS}}

\newcommand{\Sol}{\textsf{Sol}}
\newcommand{\I}{\emph{\texttt{I}}}
\def\zQ{\mathbb Q}
\let\Gm=\Gamma
\newcommand{\mc}[1]{\mathcal{#1}}
\newcommand{\mb}[1]{\mathbf{#1}}
\newcommand{\Variety}[1]{{\textbf{V}}\left( #1 \right)}
\newcommand{\Ideal}[1]{\left\langle #1 \right\rangle}

\newcommand{\Min}{\textsf{Min}}

\newcommand{\lex}{\textsf{lex }}

\newcommand{\grlex}{\textsf{grlex }}

\newcommand{\rank}{\textsf{rank}}

\newcommand{\coNPc}{\text{\textbf{coNP}-complete}}

\newcommand{\multideg}{\textnormal{multideg}}
\newcommand{\LM}{\textnormal{LM}}
\newcommand{\LT}{\textnormal{LT}}
\newcommand{\LC}{\textnormal{LC}}
\newcommand{\LCM}{\textnormal{lcm}}
\newcommand{\GB}{\text{Gr\"{o}bner} Basis}
\newcommand{\GBs}{\text{Gr\"{o}bner} Bases}
\newcommand{\reduce}[2]{{#1}\vert_{#2}}

\usepackage{xcolor}

\title{On the Complexity of CSP-based Ideal Membership Problems}
\author{Andrei A. Bulatov\thanks{{abulatov@sfu.ca}. Department of Computing Science, Simon Fraser University, Burnaby, BC, Canada. Research supported by an NSERC Discovery Grant.} \and Akbar Rafiey\thanks{{arafiey@sfu.ca}. Department of Computing Science, Simon Fraser University, Burnaby, BC, Canada. Research supported by NSERC.}}

\begin{document}
\date{}

\maketitle
\begin{abstract}
In this paper we consider the Ideal Membership Problem (IMP for short), in which we are given real polynomials $f_0,\vc fk$ and the question is to decide whether $f_0$ belongs to the ideal generated by $\vc fk$. In the more stringent version the task is also to find a proof of this fact. The IMP underlies many proof systems based on polynomials such as Nullstellensatz, Polynomial Calculus, and Sum-of-Squares. In the majority of such applications the IMP involves so called combinatorial ideals that arise from a variety of discrete combinatorial problems. This restriction makes the IMP significantly easier and in some cases allows for an efficient algorithm to solve it.

The first part of this paper follows the work of Mastrolilli [SODA'19] who initiated a systematic study of IMPs arising from Constraint Satisfaction Problems (CSP) of the form $\CSP(\Gm)$, that is, CSPs in which the type of constraints is limited to relations from a set $\Gm$. 
%
We show that many CSP techniques can be translated to IMPs thus allowing us to significantly improve the methods of studying the complexity of the IMP. We also develop universal algebraic techniques for the IMP that have been so useful in the study of the CSP. This allows us to prove a general necessary condition for the tractability of the IMP, and three sufficient ones. The sufficient conditions include IMPs arising from systems of linear equations over $\GF(p)$, $p$ prime, and also some conditions defined through special kinds of polymorphisms.


Our work has several consequences and applications. First, we introduce a variation of the IMP and based on this propose a unified framework, different from the celebrated Buchberger's algorithm, to construct a bounded degree \GB. Our algorithm, combined with the universal algebraic techniques, leads to polynomial-time construction of \GB\ for many combinatorial problems. Second, we also study a recently posed question by O'Donnell [ITCS'17] on the bit complexity of sum-of-squares (SOS) proofs and their automatizability. We prove that for almost all the CSP-based ideals SOS proofs are automatizable which improves upon the result of Raghavendra and Weitz [ICALP'17]. Furthermore, in the case where the IMP part is well behaved, we propose automatizable SOS proofs of nonnegativity with much relaxed degree restrictions. 
Finally, we provide a unifying framework that studies (construction of) theta bodies of combinatorial problems through the universal algebra of the constraint languages and present several positive results for problems such as \textsc{Stable Set}, \textsc{Binary Matroids}, \textsc{H-Coloring}, \textsc{Min/Max Ones}, and \textsc{Strict CSPs}.

\end{abstract}


\clearpage

{
  \hypersetup{linkcolor=blue, linktocpage=true}
  \tableofcontents
}

\newpage

%
%

\section{Introduction}

\paragraph{The Ideal Membership Problem.}
The study of polynomial ideals and algorithmic problems related to them goes back to David Hilbert~\cite{hilbert1893vollen}. In spite of such a heritage, methods developed in this area  till these days keep finding a wide range of applications in mathematics and computer science. In this paper we consider the Ideal Membership Problem (IMP for short), in which the goal is to decide whether a given polynomial belongs to a given ideal. It underlies such proof systems as Nullstellensatz and Polynomial Calculus. 
 
To introduce the problem more formally, let $\Field$ be a field and $\Field[x_1,x_2,\dots,x_n]$ denote the ring of polynomials over $\Field$ with indeterminates $\vc xn$. In this paper $\Field$ is always the field of real or complex numbers. A set of polynomials $\I\sse\Field[x_1,x_2,\dots,x_n]$ is said to be an \emph{ideal} if it is closed under addition and multiplications by elements from $\Field[x_1,x_2,\dots,x_n]$. By the Hilbert Basis Theorem every ideal $\I$ has a finite generating set \cite{hilbertBasisThm}, that is, there exists $P=\{f_1,f_2,\dots,f_r\}\sse\Field[x_1,x_2,\dots,x_n]$ such that for every $f_0\in\Field[x_1,x_2,\dots,x_n]$ the polynomial $f_0$ belongs to $\I$ if and only if there exists a \emph{proof}, that is, polynomials $\vc hr$ such that the identity $f_0=h_1f_1+\dots+h_rf_r$ holds. Such proofs will also be referred to as \emph{ideal membership} proofs. We then write $\I=\langle P\rangle$. The Hilbert Basis Theorem allows one to state the IMP as follows: given polynomials $f_0,\vc fr$ decide whether there exists a proof that $f_0\in\langle\vc fr\rangle$.

In many cases combinatorial or optimization problems can be encoded as collections of polynomials, and the problem is then reduced to proving or refuting that some polynomial vanishes at specified points or is nonnegative at those points. Polynomial proof systems can then be applied to find a proof or a refutation of these facts. Polynomial Calculus and Nullstellensatz proof systems are some of the standard techniques to check for zeroes of a polynomial, and Sum-of-Squares (\Sos) allows to prove or refute the nonnegativity of a polynomial. We may be interested in the length or degree of a proof in one of those systems. Sometimes such proofs can also be efficiently found --- such proof systems are referred to as \emph{automatizable} --- and in those cases we are also concerned with the complexity of finding a proof. 

The general IMP is a difficult problem and it is not even obvious whether or not it is decidable. The decidability was established in \cite{hermann1926frage,richman1974constructive,seidenberg1974constructions}. Then Mayr and Meyer~\cite{mayr1982complexity} were the first to study the complexity of the IMP. They proved an exponential space lower bound for the membership problem for ideals generated by polynomials with integer and rational coefficients. Mayer~\cite{Mayr89} went on establishing an exponential space upper bound for the $\IMP$ for ideals over $\zQ$, thus proving that such IMPs are \textbf{EXPSPACE}-complete. The source of hardness here is that a proof that $f_0\in\ang P$ may require polynomials of exponential degree. In the cases when the degree of a proof has a linear bound in the degree of $f_0$, the IMP can be solved efficiently. (There is also the issue of exponentially long coefficients that we will mention later.)

\paragraph{Combinatorial ideals.}
By Hilbert's Nullstellensatz  polynomial ideals can often be characterized by the set of common zeroes of all the polynomials in the ideal. Such sets are known as \emph{affine varieties} and provide a link between ideals and combinatorial problems, where an affine variety corresponds to the set of feasible solutions of a problem. Combinatorial problems give rise to a fairly narrow class of ideals known as \emph{combinatorial} ideals. The corresponding varieties are finite, and therefore the ideals itself are zero-dimensional and radical. The former implies that the IMP can be decided in single-exponential time~\cite{DickensteinFGS91}, while the latter will be important later for IMP algorithms. Indeed, if the IMP is restricted to radical ideals, it is equivalent to (negation of) the question: given $f_0,\vc fr$ does there exists a zero of $\vc fr$ that is not a zero of $f_0$.

The special case of the $\IMP$ with $f_0=1$ has been studied for combinatorial problems in the context of lower bounds on Polynomial Calculus and Nullstellensatz proofs, see e.g. \cite{BeameIKPP94,BussP96,Grigoriev98}. A broader approach of using polynomials to represent finite-domain constraints has been explored in  \cite{CleggEI96,JeffersonJGD13}. Clegg et al., \cite{CleggEI96}, discuss a propositional proof system based on a bounded degree version of Buchberger's algorithm~\cite{BUCHBERGER2006475} for finding proofs of unsatisfiability. Jefferson et al., \cite{JeffersonJGD13} use a modified form of Buchberger's algorithm that can be used to achieve the same benefits as the local-consistency algorithms which are widely used in constraint processing.

\paragraph{Complexity of the IMP and its applications in other proof systems.}
Whenever the degree of a proof $\vc hr$ is bounded, that is, the degree of each $h_i$ is bounded by a constant, there is an LP or SDP program of polynomial size whose solutions are the coefficients of the proof. If in addition the solution of the LP or SDP program can be represented by a polynomial number of bits (thus having low \emph{bit complexity}), a proof can be efficiently found. This property also applies to \Sos\ and provides one of the most powerful algorithmic methods for optimization problems. 

It was recently observed by O'Donnell~\cite{ODonnell17} that low degree of proofs does not necessarily implies its low bit complexity. More precisely, he presented a collection of polynomials and a polynomial such that there are low degree proofs of nonnegativity for these polynomials, that is, there also exists a polynomial size SDP whose solutions represent a \Sos\ proof of that. However, the size of those solutions are always exponential (or the proof has high bit complexity), and therefore the Ellipsoid method will take exponential time to find them. It therefore is possible that every low degree \Sos\ proof has high bit complexity. Raghavendra and Weitz \cite{RaghavendraW17} also demonstrated an example showing that this is the case even if all the constraints in the instance are Boolean, that is, on a 2-element set. 

The examples of O'Donnell and Raghavendra-Weitz indicate that it is important to identify conditions under which low degree proofs (Nullstellensatz, Polynomial Calculus, or \Sos) exist and also have low bit complexity. Raghavendra and Weitz \cite{RaghavendraW17} suggested some sufficient conditions of this kind for \Sos\ proofs that are satisfied for a number of well studied problems such as Matching, TSP, and others. More precisely, they formulate three conditions that a polynomial system ought to satisfy to yield for a low bit complexity \Sos\ proof. Two of these conditions hold for the majority of combinatorial ideals, while the third, the low degree of Nullstellensatz, is the only nontrivial one. Noting that Nullstellensatz proofs are basically witnesses of ideal membership, the IMP is at the core of all the three proof systems. 

\paragraph{IMP and CSP.}
The Constraint Satisfaction Problem (CSP) provides a general framework for a wide range of combinatorial problems. In a CSP we are given a set of variables and a collection of constraints, and we have to decide whether the variables can be assigned values so that all the constraints are satisfied. Following \cite{JeffersonJGD13,vandongenPhd,Mastrolilli19} every CSP can be associated with a polynomial ideal. Let CSP $\cP$ be given on variables $\vc xn$ that can take values from a set $D=\{0,\dots,t-1\}$. The ideal $\I(\cP)$ of $\Field[\vc xn]$ whose corresponding variety equals the set of solutions of $\cP$ is constructed as follows. First, for every $x_i$ the ideal $\I(\cP)$ contains a \emph{domain} polynomial $f_D(x_i)$ whose zeroes are precisely the elements of $D$. Then for every constraint $R(x_{i_1},\dots, x_{i_k})$, where $R$ is a predicate on $D$, the ideal $\I(\cP)$ contains a polynomial $f_R(x_{i_1},\dots, x_{i_k})$ that interpolates $R$, that is, for $(x_{i_1},\dots, x_{i_k})\in D^k$ it holds $f_R(x_{i_1},\dots, x_{i_k})=0$ if and only if $R(x_{i_1},\dots, x_{i_k})$ is true. This model generalizes a number of constructions used in the literature to apply Nullstellensatz or \Sos\  proof systems to combinatorial problems, see, e.g., \cite{BeameIKPP94,BussP96,Grigoriev98,RaghavendraW17}. 

The construction above also provides useful connections between the IMP and the CSP. For instance, the CSP $\cP$ is unsatisfiable if and only if the variety associated with $\I(\cP)$ is empty, or equivalently, if and only if $1\in\I(\cP)$ (1 here denotes the polynomial of degree 0). In this sense it is related to the standard decision version of the CSP. However, since $\I(\cP)$ is radical for any instance $\cP$, the \IMP\ reduces to verifying whether every point in the variety of $\I(\cP)$ is a zero of $f_0$. Thus, it is probably closer to the CSP Containment problem (given two CSP instances over the same variables, decide if every solution of the first one is also a solution of the second one), which has mainly been studied in the context of Database theory and Conjunctive Query Containment, see, e.g., \cite{Kolaitis98:conjunctive}. 

Studying CSPs in which the type of constraints is restricted has been a very fruitful research direction. For a fixed set of relations $\Gm$ over a finite set $D$, also called a \emph{constraint language}, $\CSP(\Gm)$ denotes the CSP restricted to instances with constraints from $\Gm$, \cite{Schaefer78,FV98}. Feder and Vardi in \cite{FV98} conjectured that every CSP of the form $\CSP(\Gm)$ is either solvable in polynomial time or is \textbf{NP}-complete. This conjecture was recently confirmed by Bulatov and Zhuk in \cite{Bulatov17,Zhuk17}. Mastrolilli in \cite{Mastrolilli19} initiated a similar study of the IMP. Let $\IMP(\Gm)$ be the IMP restricted to ideals produced by instances from $\CSP(\Gm)$. As was observed above, the complement of $\CSP(\Gm)$ reduces to $\IMP(\Gm)$, and so $\IMP(\Gm)$ is \coNPc\ whenever $\CSP(\Gm)$ is \textbf{NP}-hard. Therefore the question posed in \cite{Mastrolilli19} is: 

\begin{problem}\label{prob:Mastrolilli}
For which constraint languages $\Gm$ is it possible to efficiently find a generating set for the ideal $\I(\cP)$, $\cP\in\CSP(\Gm)$, that allows for a low bit complexity proof of ideal membership?
\end{problem}

Mastrolilli \cite{Mastrolilli19} (along with \cite{Bharathi-Minority}) resolved this question in the case when $\Gm$ is Boolean language, that is, over the set $\{0,1\}$. He proved that in this case $\IMP(\Gm)$ is polynomial time solvable, and moreover ideal membership proofs can be efficiently found, too, for any Boolean $\Gm$ for which $\CSP(\Gm)$ is polynomial time solvable. However, $\IMP(\Gm)$ in \cite{Mastrolilli19} satisfies two restrictions. First, the result is obtained under the assumption that $\Gm$ contains the constant relations that allows one to fix a value of a variable. Mastrolilli called such languages \emph{idempotent}. We will show that this suffices to obtain a more general result. Second, in the majority of cases a bound on the degree of the input polynomial $f_0$ has to be introduced. The IMP where the input polynomial has degree at most $d$ will be denoted by $\IMP_d,\IMP_d(\Gm)$. The exact result is that for any idempotent Boolean language $\Gm$ the problem $\IMP_d(\Gm)$ is polynomial time solvable for any $d$ when $\CSP(\Gm)$ is polynomial time solvable, and an ideal membership proof can be efficiently found, and $\IMP_0(\Gm)$ is \coNPc\ otherwise. We will reflect on the distinction between $\IMP$ and $\IMP_d$ later in the paper, but do not go deeply into that. The case when $\CSP(\Gm)$ is equivalent to solving systems of linear equations modulo~2 was also considered in \cite{Bharathi-Minority} fixing a gap in \cite{Mastrolilli19}. There has been very little work done on $\IMP(\Gm)$ beyond 2-element domains. The only result we are aware of is \cite{Bharathi-Dual-Disc} that proves that $\IMP_d(\Gm)$ is polynomial time when $\Gm$ is on a 3-element domain and is invariant under the so-called dual-discriminator operation, which imposes very strong restrictions on the relations from $\Gm$; we will discuss this case in greater details in Section~\ref{sec:dual-discriminator}.

The main tool for proving the tractability of $\IMP(\Gm)$ is constructing a \GB\ of the corresponding ideal. It is not hard to see that the degree of polynomials in a \GB\ of an ideal of $\Field[\vc xn]$ that can occur in $\IMP(\Gm)$ is only bounded by $n|D|$, where $\Gm$ is over a set $D$. Therefore the basis and polynomials themselves can be exponentially large in general. In fact, this is the main reason why considering $\IMP_d(\Gm)$ instead makes the problem easier. For solving this problem it suffices to find a \emph{$d$-truncated} \GB, in which the degree of polynomials is bounded by $d$, and so such a basis always has polynomial size. Thus, the (possible) hardness of $\IMP_d(\Gm)$ is due to the hardness of constructing a \GB.

\subsubsection*{Our contribution}

In this paper we expand on \cite{Mastrolilli19} and \cite{Bharathi-Dual-Disc,Bharathi-Minority} in several ways. We consider $\IMP(\Gm)$ for languages $\Gm$ over arbitrary finite set and attempt to obtain general results about such problems. However, we mainly focus on a slightly different problem than Problem~\ref{prob:Mastrolilli}.

\begin{problem}\label{prob:us}
For which constraint languages $\Gm$ the problem $\IMP(\Gm)$ [or $\IMP_d(\Gm)$] can be solved in polynomial time?
\end{problem}

Note that answering whether $f_0$ belongs to a certain ideal does not necessarily mean finding an ideal membership proof of that. However, we will argue that, firstly, in many applications this is the problem we need to solve and therefore our results apply. Secondly, in Section~\ref{sec:ximp-truncated} we will show that in the majority of cases if the existence of an ideal membership proof can be efficiently decided, such a proof can also be efficiently found.

\paragraph{Expanding the constraint language.}
Firstly, in Section~\ref{sec:IMPs} we study reductions between IMP's when the language $\Gm$ is enlarged in certain ways. Let $\Gm$ be a constraint language over a set $D$. By $\Gm^*$ we denote $\Gm$ with added \emph{constant} relations, that is, relations of the form $\{(a)\}$, $a\in D$. Imposing such a constraint on a variable $x$ essentially fixes the allowed values of $x$ to be $a$. First, we prove that adding constant relations does not change the complexity of the IMP.

\begin{theorem}\label{the:adding-constants-intro}
For any $\Gm$ over $D$ the problem $\IMP(\Gm^*)$ is polynomial time reducible to $\IMP(\Gm)$, and for any $d$ the problem $\IMP_d(\Gm^*)$ is polynomial time reducible to $\IMP_{d+|D|(|D|-1)}(\Gm)$.
\end{theorem}

Theorem~\ref{the:adding-constants-intro} has two immediate consequences. Since $\IMP(\Gm)$ is in \textbf{co-NP}, for any $\CSP(\Gm)$ instance $\cP$ there is always a proof that the input polynomial $f_0$ does not belong to the ideal $\I(\cP)$. Any solution of $\cP$ that is not a zero of $f_0$ will do. Finding such a proof may be treated as a search version of $\IMP(\Gm)$. Through self-reducibility, Theorem~\ref{the:adding-constants-intro} allows us to solve the search problem.

\begin{theorem}\label{the:search-intro}
Let $\Gm$ be such that $\IMP(\Gm)$ [$\IMP_{d+|D|(|D|-1)}(\Gm)$] is solvable in polynomial time. Then for any instance $(f_0,\cP)$ of $\IMP(\Gm)$ [$\IMP_d(\Gm)$] such that $f_0\not\in\I(\cP)$, a solution $\mb a$ of $\cP$ such that $f_0(\mb a)\ne0$ can also be found in polynomial time. 
\end{theorem}

Theorem~\ref{the:adding-constants-intro} also provides a hint at a more plausible conjecture for which languages $\Gm$ the problem $\IMP(\Gm)$ or $\IMP_d(\Gm)$ is polynomial time. In particular, it allows to find an example of $\Gm$ such that $\CSP(\Gm)$ is tractable while $\IMP_d(\Gm)$ is not, even for a $\Gm$ on a 2-element set and $d=1$. Later we state some results that might indicate that $\IMP_d(\Gm)$ is polynomial time for every $\Gm$ such that $\CSP(\Gm^*)$ is polynomial time. Note that the structure of such CSPs is now very well understood.

Another way of expanding a constraint language is by means of \emph{primitive-positive (pp-) definitions} and \emph{pp-interpretations}, and it is at the core of the so-called algebraic approach to the CSP. A relation $\rel$ is said to be pp-definable in $\Gm$ if there is a first order formula $\Phi$ using only conjunctions, existential quantifiers, equality relation, and relations from $\Gm$ that is equivalent to $\rel$. Pp-interpretations are more complicated (see Section~\ref{sec:pp-interpretations}) and allow for certain encodings of $\rel$.

\begin{theorem}\label{the:pp-definitions-intro}
\begin{itemize}
    \item [(1)] Let $\Gm,\Dl$ be constraint languages over the same set $D$, $\Dl$ is finite, and every relation from $\Dl$ is pp-definable in $\Gm$. Then $\IMP(\Dl)$ is polynomial time reducible to $\IMP(\Gm)$ and $\IMP_d(\Dl)$ is polynomial time reducible to $\IMP_d(\Gm)$ for any $d$.
    \item[(2)] Let $\Gm,\Dl$ be constraint languages, $\Dl$ is finite, and $\Dl$ is pp-interpretable in $\Gm$. 
    Then there is a constant $k$ such that $\IMP_d(\Dl)$ is polynomial time reducible to $\IMP_{kd}(\Gm)$ for any $d$.
\end{itemize}
\end{theorem}

The approach of Theorem~\ref{the:pp-definitions-intro} was first applied to various proof systems in \cite{Atserias19:proof}, although that work is mostly concerned with proof complexity rather than computational complexity. Mastrolilli \cite{Mastrolilli19} ventured into pp-definability without proving any reductions. In particular, the first part of Theorem~\ref{the:pp-definitions-intro} uses techniques from \cite{Mastrolilli19} for projections of ideals. It will later allow us to develop further universal algebra techniques for the IMP. The second part of that theorem will also work towards more powerful universal algebra methods.

Recall that according to \cite{RaghavendraW17} in order to find an \Sos\ proof one needs to be able to find Nullstellensatz proofs efficiently. The reductions from Theorems~\ref{the:adding-constants-intro}, \ref{the:pp-definitions-intro} do not always allow to recover such a proof efficiently only confirming such a proof exists, and therefore cannot be directly used in the conditions from \cite{RaghavendraW17}. However, in Section~\ref{sec:SOS} we address this issue. 


\paragraph{Polymorphisms and sufficient conditions for tractability.}
Recall that a \emph{polymorphism} of a constraint language $\Gm$ over a set $D$ is a multi-ary operation on $D$ that can be viewed as a multi-dimensional symmetry of relations from $\Gm$. By $\Pol(\Gm)$ we denote the set of all polymorphisms of $\Gm$. As in the case of the CSP, Theorem~\ref{the:pp-definitions-intro} implies that polymorphisms of $\Gm$ is what determines the complexity of $\IMP(\Gm)$. In Section~\ref{sec:polymorphisms} we show the following.

\begin{corollary}\label{cor:polymorphisms-intro}
Let $\Gm,\Dl$ be constraint languages over the same set $D$, $\Dl$ is finite. If $\Pol(\Gm)\sse\Pol(\Dl)$, then $\IMP(\Dl)$ is polynomial time reducible to $\IMP(\Gm)$ and $\IMP_d(\Dl)$ is polynomial time reducible to $\IMP_d(\Gm)$ for any $d$.
\end{corollary}

Corollary~\ref{cor:polymorphisms-intro} allows us to represent \IMP s through polymorphisms and classify the complexity of \IMP s according to the corresponding polymorphisms. The method has been initiated by Mastrolilli and Bharathi \cite{Bharathi-Dual-Disc,Bharathi-Minority,Mastrolilli19}, although mainly for 2- and one case of 3-element sets. We apply this approach to obtain three sufficient conditions for tractability of the \IMP. 

\begin{theorem}\label{the:sufficient-intro}
Let $\Gm$ be a constraint language over a set $D$. Then if one of the following conditions holds, $\IMP_d(\Gm)$ is polynomial time solvable for any $d$.
\begin{enumerate}
\item
$\Gm$ has the \emph{dual-discriminator} polymorphism (i.e.\ a ternary operation $g$ such that $g(x,y,z)=x$ unless $y=z$, in which case $g(x,y,z)=y$);
\item
$\Gm$ has a \emph{semilattice} polymorphism (i.e.\ a binary operation $f$ such that $f(x,x)=x$, $f(x,y)=f(y,x)$, and $f(f(x,y),z)=f(x,f(y,z))$);
\item
$|D|=p$, $p$ prime, and $\Gm$ has an \emph{affine} polymorphism modulo $p$ (i.e.\ a ternary operation $h(x,y,z)=x\ominus y\oplus z$, where $\oplus,\ominus$ are addition and subtraction modulo $p$, or, equivalently, of the field $\GF(p)$). In this case every CSP can be represented as a system of linear equations over $\GF(p)$.
\end{enumerate}
\end{theorem}

The three polymorphisms occurring in Theorem~\ref{the:sufficient-intro} have played an important role in the CSP research. For one reason they completely cover the tractable cases when $|D|=2$ and therefore the results of \cite{Mastrolilli19,Bharathi-Minority}, although we used some results (on semilattice polymorphisms) from \cite{Mastrolilli19}. 

For the first part of Theorem~\ref{the:sufficient-intro} (also known as \emph{0/1/all constraints}) we come up with a technique of preprocessing the input polynomial $f_0$ that allows us to get rid of permutation constraints and greatly simplify the proof for the dual-discriminator polymorphism, including the special case $|D|=3$ considered in \cite{Bharathi-Dual-Disc}. 

In the second part of Theorem~\ref{the:sufficient-intro} we use the fact, see \cite{Papert64:congruence}, that any language with a semilattice polymorphism is pp-interpretable in a language on a 2-element set also having a semilattice polymorphism. Then Theorem~\ref{the:sufficient-intro}(2) follows from Theorem~\ref{the:pp-definitions-intro}(2) and the results of \cite{Mastrolilli19}. 

As is mentioned, the third part of Theorem~\ref{the:sufficient-intro} is in fact about systems of linear equations over $\GF(p)$, since every instance of $\CSP(\Gm)$, where $\Gm$ has an affine polymorphism, is equivalent to a system of linear equations over $\GF(p)$. Bharathi and Mastrolilli solved this case for $p=2$ showing that $\IMP_d(\Gm)$ is polynomial time for any $d$. Their approach is based on FGML algorithm \cite{FGLM} to construct a bounded degree \GB. Instead, we map an instance of $\IMP(\Gm)$ on a different domain consisting of $p$-th roots of unity rather than $\{0\zd p-1\}$. This transforms the generators of the ideal into very simple polynomials that form a \GB\ without any further modifications. Such a transformation makes it difficult to find an ideal membership proof, but this is resolved in Section~\ref{sec:ximp-truncated}.

\paragraph{Algebras.}
In Section~\ref{sec:algebraic} we prove that the standard features of the universal algebraic approach to the CSP work for \IMP\ as well. These include reductions for standard algebraic constructions such as \emph{subalgebras}, \emph{direct powers}, and \emph{homomorphic images}. They easily follow from Theorem~\ref{the:pp-definitions-intro}(2). A more general construction of direct product requires a more general version of $\CSP(\Gm)$, and therefore of $\IMP(\Gm)$, the multi-sorted one, in which every variable can have its own domain of values. We do not venture into that direction in the paper. One implication of these results is a necessary condition for tractability of $\IMP(\Gm)$ that follows from a similar one for the CSP.

\paragraph{The IMP with indeterminate coefficients.}
In the second part of the paper we consider a number of applications of the techniques developed in the first part. The key to those applications is an extension of the IMP defined as follows. Given an ideal $\I\sse\Field[\vc x n]$ and a vector of $\ell$ polynomials $M=(g_1,\dots,g_\ell)$, the \xIMP~asks if there exist coefficients $\mathbf{c}=(\vc c \ell)\in \Field^\ell$ such that $\mathbf{c}M=\sum_{i=1}^\ell c_i g_i$ belongs to the ideal $\I$. 

As with the regular IMP, \xIMP\ can be parametrized by specifying a constraint language $\Gm$, in which case the resulting problem \xIMP$(\Gm)$ (or \xIMP${}_d(\Gm)$ if the degree of input polynomials is bounded) only allows ideal produced by instances of $\CSP(\Gm)$. 

We prove that \xIMP\ can be solved in polynomial time when a ($d$-truncated) \GB\ can be efficiently generated, and also admits the same reductions as the IMP.

\begin{theorem}
Let $\Gm,\Dl$ be constraint languages and $\Dl$ is finite. Then 
\begin{itemize}
    \item [(1)] If every relation from $\Dl$ is pp-definable in $\Gm$, then $\chi\IMP(\Dl)$ is polynomial time reducible to $\chi\IMP(\Gm)$ and $\chi\IMP_d(\Dl)$ is polynomial time reducible to $\chi\IMP_d(\Gm)$ for any $d$.
    \item[(2)] If $\Dl$ is pp-interpretable in $\Gm$, 
    then there is a constant $k$ such that $\chi\IMP_d(\Dl)$ is polynomial time reducible to $\chi\IMP_{kd}(\Gm)$ for any $d$.
\end{itemize}
\end{theorem}

The theorem above allows us to show that for every $\Gm$ for which $\IMP_d(\Gm)$ is polynomial time solvable, so is \xIMP${}_d(\Gm)$. This includes constraint languages invariant under dual-discriminator, semilattice, and affine polymorphisms. 

In Section~\ref{sec:ximp-truncated} we use \xIMP\ along with the factor ring $\Field[\vc xn]/\I$ modulo an ideal $\I$ to generate a basis for the factor ring consisting of monomials of degree at most $d$, and then use it to construct a $d$-truncated \GB\ for $\I$. More precisely, we prove the following.

\begin{theorem}\label{thm:GB+xIMP-intro}
    Let $\mc{H}$ be a class of ideals for which $\chi\IMP_d$ is polynomial time solvable. Then there exists a polynomial time algorithm that constructs a degree $d$ \GB~of an ideal $\I\in \mc{H}$, $\I\subseteq\Field[\vc xn]$, in time $O(n^d)$.  
\end{theorem}

Theorem~\ref{thm:GB+xIMP-intro} makes it possible to construct $d$-truncated \GBs\ in all cases $\IMP(\Gm)$ is known to be polynomial time. Thus, it basically eliminates the gap between deciding the existence of an ideal membership proof and finding such a proof.

\paragraph{IMP and SOS proofs.}
We apply \xIMP\ and algebraic techniques to finding Sum-of-Squares (SoS) proofs. For variables $\vc xn$ a semialgebraic set is given by a collection of polynomial equalities and inequalities such as 
\begin{align*}
    S=\{\bx\in \zR^n\mid p_1(\bx)=0,\dots,p_{m}(\bx)=0, q_{1}(\bx)\geq 0,\dots,q_{\ell}(\bx)\geq 0\}.
\end{align*}
The goal is to proof that some polynomial $r(\bx)$ is nonnegative on $S$. An \Sos\ proof of $r(\bx) \geq 0$ is given by a polynomial identity of the form
\begin{align}
\label{eq:main-body-sos-proof}
    r(\bx) = \sum_{i=1}^{t_0} h_i^2(\bx) + \sum_{k=1}^\ell (\sum_{j=1}^{t_k} s_j^2(\bx)) q_k(\bx) + \sum_{i=1}^m \lambda_i(\bx) p_i(\bx).
\end{align}
The degree of an \Sos\ proof is often defined to be the maximum degree of the polynomials involved in the proofs i.e., $\max\{deg(h_i^2),deg(s_j^2q_k),deg(\lambda_ip_i)\}$. If the degree is bounded such a proof can often be found using an SDP program. 

We obtain two results related to \Sos\ proofs. The first one deals with the recently discovered issue of high bit complexity of such proofs. O'Donnell \cite{ODonnell17} discovered that in some cases although an \Sos\ proof of low degree exists, it may involve exponentially long coefficients, and therefore it may be impossible to find such a proof efficiently. Raghavendra and Weitz~\cite{RaghavendraW17} suggested three conditions such that if the set $S$ satisfies them, the existence of a low degree \Sos\ proof implies the existence of a low bit complexity one. In Section~\ref{sec:automatizability} we show that for sets $S$ in which the polynomials $\vc pm$ are produced from some CSP instance of the form $\CSP(\Gm)$ one of these conditions ($\dl$-richness) is guaranteed to be true, and another one ($k$-effectiveness) can be dropped altogether. 

\begin{theorem}
     Let $\mc{P}$ be an instance of $\CSP(\Gamma)$ and $\I(\mc{P})=\Ideal{p_1\dots,p_m}$ be the corresponding ideal to $\mc{P}$. Assume that $\forall q_i,\forall\mathbf{a}\in S$ we have $q_i(\mathbf{a})> \varepsilon$.
    
    Then if $r$ has a degree $d$ \Sos\ proof of nonnegativity on $S$, it also has a degree $d$ \Sos\ proof of nonnegativity with coefficients bounded by $2^{poly(n^d, \log\frac{1}{\varepsilon})}$. In particular, if there are no polynomial inequalities then every coefficient can be written down with only $poly(n^d)$ bits.
\end{theorem}

The second result concerns the degree bounds that are required for the \Sos\ proof system to be automatizable. An \Sos\ proof is clearly divided into the ``ideal part'', a combination of $\vc pm$ that belongs to the corresponding ideal and the ``\Sos\ part'' (the rest). The standard requirement is that $\max\{deg(h_i^2),deg(s_j^2q_k),deg(\lambda_ip_i)\}$ has to be bounded in order for \Sos\ to be automatizable. We show that this requirement can be relaxed so that only $\max\{deg(h_i^2),deg(s_j^2q_k)\}$ is bounded, provided the ideal generated by $\vc pm$ belongs to a tractable $\chi\IMP(\Gm)$. 

\begin{theorem}
\label{thm:relaxed-SOS:intro}
    Let $\Gm$ be a constraint language such that $\chi\IMP(\Gamma)$ is polynomial time decidable. Let $\mc{P}$ be an instance of $\CSP(\Gamma)$ and $\I(\mc{P})=\Ideal{p_1\dots,p_m}$ be the corresponding ideal to $\mc{P}$. Assume that $\forall q_i,\forall\mathbf{a}\in S$ we have $q_i(\mathbf{a})> \varepsilon$.
    
    Then for a polynomial $r$, the existence of an \Sos\ proof \eqref{eq:main-body-sos-proof} with $\max\{deg(h_i^2),deg(s_j^2)\} \leq d$ is polynomial time decidable. 
\end{theorem}

The above theorem could potentially improve upon the existing \Sos\ SDP relaxations for several problems. For instance, problems such as \textsc{Vertex Cover}, \textsc{Clique}, and \textsc{Stable Set}, are all instances of Boolean \CSP s for which the corresponding \xIMP\ is polynomial time decidable. More generally, \xIMP\ is polynomial time decidable for the \textsc{2-SAT} problem, hence  polynomial optimization problems over \textsc{2-SAT} can enjoy the more relaxed version of \Sos\ proofs introduced in Theorem \ref{thm:relaxed-SOS:intro}. An interesting research direction would be studying the power of this relaxed setting in terms of approximation algorithms. 

\paragraph{IMP and theta bodies.}
Finally, in Section~\ref{sec:theta-body}, we apply our techniques to show that the \emph{theta bodies} \cite{GouveiaPT10-ThetaBody} arising from certain combinatorial problems can be constructed in polynomial time. One of the core problems in optimization is to understand the $\texttt{conv}(S)$ or a relaxation of $\texttt{conv}(S)$, where $S$ the set of feasible solutions to a given problem and $\texttt{conv}(S)$ denotes the convex hull of $S$. Here, we consider combinatorial ideals arising from \CSP s where $S$ is a finite subset of $\zR^n$. Theta body relaxations, introduced by Gouveia, Parrilo and Thomas \cite{GouveiaPT10-ThetaBody}, obtain a hierarchy of relaxations to $\texttt{conv}(S)$. They are strong relaxations, for instance, they achieve the best approximation among all symmetric SDPs of a comparable size \cite{weitz-PhD} and are known to have nice properties \cite{GouveiaPT10-ThetaBody}.
Let $TH_k(\I(\mc{P}))$, $\mc{P}\in\CSP(\Gm)$, denote the $k$-th theta body of $\I(\mc{P})$. Theta bodies create a nested sequence of closed convex relaxations i.e., $TH_1(\I(\mc{P}))\supseteq TH_2(\I(\mc{P}))\supseteq \dots \supseteq \texttt{conv}(S)$.  We say a constraint language $\Gm$ is $TH_k$-exact if for any instance $\mc{P}$ of $\CSP(\Gm)$ the ideal $\I(\mc{P})$ is \emph{$TH_k$-exact}. The ideal $\I(\mc{P})$ is $TH_k$-exact if $TH_k(\I(\mc{P}))=\texttt{conv}(S)$. Zero-dimensional ideals are $TH_k$-exact for some finite $k$ \cite{Laurent07}.
An intriguing question is characterizing $TH_k$-exact constraint languages, for constant $k$: \emph{Which constraint languages are $TH_k$-exact, for some constant $k$?}. This is analogous to Lov{\'a}sz's question \cite{lovasz2003semidefinite} where he asked: \emph{Which ideals in $\zR[\vc xn]$ are $TH_k$-exact \footnote{In fact, Lov{\'a}sz's asked which ideals in $\zR[\vc xn]$ are $(1,k)$-\Sos. An ideal is $(1,k)$-\Sos\ if every nonnegative linear polynomial on $S$ is $k$-\Sos\ mod $\I$. However, it turns out that a radical ideal in $\zR[\vc xn]$ is  $(1,k)$-\Sos\ if and only if it is $TH_k$-exact \cite{GouveiaPT10-ThetaBody}.}.} This question is partially answered in \cite{GouveiaPT10-ThetaBody} for the case $k=1$ from a completely different perspective. From algorithmic point of view, a natural question to ask is that for which constraint languages constructing theta bodies is a polynomial time task.

\begin{problem}
\label{problem:TH-comp-intro}
    For which constraint languages $\Gm$ the $k$-th theta body $TH_k(\I(\mc{P}))$ is computable in polynomial time where $\mc{P}$ is an instance of $\CSP(\Gm)$?
\end{problem}

Here, we provide strong evidence that the polymorphisms of constraint languages might be the right notion that one should consider to address Problem~\ref{problem:TH-comp-intro}. To construct the $k$-th theta bodies it is sufficient to obtain a basis for the factor ring $\zR[\vc xn]_k/\I$ modulo ideal $\I$ \cite{Laurent07,parrilo2005exploiting}. Our result in Theorem \ref{thm:GB+xIMP-intro} gives us such a luxury. While there are only very few problems for which efficient construction of theta bodies are known, we provide a unifying framework to study the computational aspects of theta bodies and, ultimately,  making progress towards answering Problem~\ref{problem:TH-comp-intro}. In particular, we present (rediscover) several positive results for problems such as \textsc{Stable Set}, \textsc{Binary Matroids}, \textsc{H-Coloring}, \textsc{Min/Max Ones}, and \textsc{Strict CSPs}.
\section{Preliminaries}

\subsection{Ideals, varieties and the Ideal Membership Problem}

Let $\Field$ denote an arbitrary field. Let $\Field[x_1,\dots, x_n]$ be the ring of polynomials over a field $\Field$ and indeterminates $x_1, \dots,x_n$. Sometimes it will be convenient not to assume any specific ordering or names of the indeterminates. In such cases we use $\Field[X]$ instead, where $X$ is a set of indeterminates, and treat points in $\Field^X$ as mappings $\vf:X\to\Field$. The value of a polynomial $f\in\Field[X]$ is then written as $f(\vf)$. Let $\Field[x_1,\dots, x_n]_d$ denote the subset of polynomials of degree at most $d$. An \emph{ideal} of $\Field[x_1,\dots, x_n]$ is a set of polynomials from $\Field[x_1,\dots, x_n]$ closed under addition and multiplication by a polynomial from $\Field[x_1,\dots, x_n]$. We will need ideals represented by a generating set.

\begin{definition}
    The ideal (of $\Field[x_1,\dots, x_n]$) generated by a finite set of polynomials $\{f_1, \dots,f_m\}$ in $\Field[x_1,\dots, x_n]$ is defined as
    \[
        \mb{I}(f_1, \dots,f_m)\overset{\mathrm{def}}{=}\Big\{ \sum\limits_{i=1}^m t_if_i \mid t_i \in \Field[x_1,\dots, x_n]\Big\}.
    \]
\end{definition}

\begin{definition}
    The set of polynomials that vanish in a given set $S \subset \Field^n$ is called the \emph{vanishing ideal} of $S$ and denoted
    \[
        \mb{I}(S) \overset{\mathrm{def}}{=} \{f \in \Field[x_1,\dots, x_n] : f(a_1,\dots,a_n) = 0 \ \ \forall (a_1,\dots,a_n) \in S\}.
    \]
\end{definition}

\begin{definition}
An ideal $\I$ is \emph{radical} if $f^m \in\I$ for some integer $m\geq 1$ implies that $f\in \I$. For an arbitrary ideal $\I$ the smallest radical ideal containing $\I$ is denoted $\sqrt\I$. In other words $\sqrt\I=\{f\in \Field[x_1,\dots, x_n]\mid f^m\in\I\text{ for some $m$}\}$.
\end{definition}

Another common way to denote $\mb{I}(f_1,\ldots,f_m)$ is by $\langle f_1,\ldots,f_m \rangle$ and we will use both notations interchangeably.

\begin{definition}\label{def:V(I)}
    Let $\{f_1,\ldots, f_m\}$ be a finite set of polynomials in $\Field[x_1,\ldots,x_n]$. We call
    \[\Variety{ f_1,\ldots,f_m}\overset{\mathrm{def}}{=} \{(a_1,\ldots,a_n)\in \Field^n \mid  f_i(a_1,\ldots,a_n)=0 \quad 1\leq i\leq m\}\]
    the \emph{affine variety} defined by $f_1,\ldots, f_m$.

Similarly, for an ideal $\I\subseteq \Field[x_1,\ldots,x_n]$ we denote by $\Variety{\I}$ the set $\Variety{\I}=\{(a_1,\ldots,a_n)\in \Field^n\mid f(a_1,\ldots,a_n)=0 \quad \forall f\in \I\}$.
\end{definition}

The Weak Nullstellensatz states that in any polynomial ring, algebraic closure is enough to guarantee that the only ideal which represents the empty variety is the entire polynomial ring itself. This is the basis of one of the most celebrated mathematical results, Hilbert's Nullstellensatz.

\begin{theorem}[The Weak Nullstellensatz]
    \label{W-Null}
    Let $\Field$ be an algebraically closed field and let $\I \subseteq \Field[\vc x n]$ be an ideal satisfying $\Variety{\I} = \emptyset$. Then $\I = \Field[\vc x n]$.
\end{theorem}

One might hope that the correspondence between ideals and varieties is one-to-one provided only that one restricts to algebraically closed fields. Unfortunately, this is not the case\footnote{For example $\Variety{x} = \Variety{x^2} = \{0\}$ works over any field.}. Indeed, the reason that the map $\mathbf{V}$ fails to be one-to-one is that a power of a polynomial vanishes on the same set as the original polynomial.

\begin{theorem}[Hilbert's Nullstellensatz]
    \label{W-Hilbert-Null}
    Let $\Field$ be an algebraically closed field. If $f_0, \vc fs \in \Field[\vc x n]$, then $f_0 \in \mathbf{I}(\Variety{\vc f s})$ if and only if $f^m_0 \in \Ideal{\vc f s}$
    for some integer $m \geq 1$.
\end{theorem}

By definition, radical ideals consist of all polynomials which vanish on some variety $V$. This together with Theorem~\ref{W-Hilbert-Null} suggests that there is a one-to-one correspondence between affine varieties and radical ideals.

\begin{theorem}[The Strong Nullstellensatz]
    \label{S-Hilbert-Null}
    Let $\Field$ be an algebraically closed field. If $\I$ is an ideal in $\Field[\vc x n]$, then $\mathbf{I}(\Variety{\I})=\sqrt{\I}$.
\end{theorem}

The following theorem is a useful tool for finding unions and intersections of varieties. We will use it in the following subsections where we construct ideals corresponding to \CSP~instances.

\begin{theorem}[\cite{Cox}]\label{th:ideal_intersection}
  If $I$ and $J$ are ideals in $\Field[x_1,\ldots,$ $x_n]$, then
  \begin{itemize}
      \item [i.] $\Variety{I\cap J}= \Variety{I}\cup \Variety{J}$,
      \item [ii.] $\Variety{I+J}=\Variety{I}\cap\Variety{J}$.
  \end{itemize}
\end{theorem}

\subsection{The Constraint Satisfaction Problem}

We use $[k]$ to denote $\{1\zd k\}$. Let $D$ be a finite set, it will often be referred to as a domain. An \emph{$n$-ary relation} on $D$ is a set of $n$-tuples of elements from $D$; we use $\mb{R}_D$ to denote the set of all finitary relations on $D$. 
A \emph{constraint language} is a subset of $\mb{R}_D$, and may be finite or infinite.

A \emph{constraint} over a constraint language $\Gamma \subseteq \mb{R}_D$ is a pair $\langle \bs, R\rangle$ with $\bs=(x_1,\dots,x_k)$ a list of variables of length $k$ (not necessarily distinct), called the \emph{constraint scope}, and $R$ a $k$-ary relation on $D$, belonging to $\Gamma$, called the \emph{constraint relation}. Another common way to denote a constraint $\langle \bs, R\rangle$ is by $R(\bs)$, that is, to treat $R$ as a predicate, and we will use both notations interchangeably. A constraint is satisfied by a mapping $\vf:\{\vc xk\}\to D$  if $(\vf(x_1),\dots,\vf(x_k)) \in R$.

\begin{definition}[Constraint Satisfaction Problem]
    The constraint satisfaction problem over a constraint language $\Gamma \subseteq \mb{R}_D$, denoted $\CSP(\Gamma)$, is defined to be the decision problem with instance $\cP = (X, D, C)$, where $X$ is a finite set of variables, $D$ is the domain, and $C$ is a set of constraints over $\Gamma$ with variables from $X$. The goal is to decide whether or not there exists a solution, i.e. a mapping $\vf: X \to D$ satisfying all of the constraints. We will use $\Sol(\cP)$ to denote the (possibly empty) set of solutions of $\cP$.
\end{definition} 

It is known \cite{Bulatov17,Zhuk17} that for any constraint language $\Gm$ (finite or infinite) on a finite set the problem $\CSP(\Gm)$ is either solvable in polynomial time or is \textbf{NP}-complete. We will use this fact to determine the complexity of the Ideal Membership Problem.

\subsection{The ideal-CSP correspondence}\label{sect:idealCSP}

Here, we explain how to construct an ideal corresponding to a given instance of \CSP. Constraints are in essence varieties, see e.g.~\cite{JeffersonJGD13,vandongenPhd}. Following~\cite{Mastrolilli19,vandongenPhd}, we shall translate CSPs to polynomial ideals and back. Let $\cP=(X,D,C)$ be an instance of $\CSP(\Gamma)$. Without loss of generality, we assume that $D\subset\Field$\footnote{In fact, we will mainly assume $\Field=\Real$ and $D=\{0,1,\dots,|D|-1\}$}. Let $\Sol(\cP)$ be the (possibly empty) set of all solutions of $\cP$. We wish to map $\Sol(\cP)$ to an ideal $\I(\cP)\subseteq \Field[X]$ such that $\Sol(\cP)=\Variety{\I(\cP)}$.

First we show how to map a constraint to a generating system of an ideal. Consider a constraint $\langle \bs,R\rangle$ from $C$ where $\bs=(x_{i_1},\ldots,x_{i_k})$ is a $k$-tuple of variables from $X$. Every $\bv=(v_1,\ldots,v_k)\in R$ corresponds to some point $\bv\in \Field^k$. The ideal $\Ideal{\{\bv\}}\overset{\mathrm{def}}{=} \Ideal{x_{i_1}-v_{1},\ldots,x_{i_k}-v_{k}}$ is an ideal in $\Field[\bs]$. Moreover, $\Ideal{\{\bv\}}$ is radical~\cite{Cox} and $\Variety{\Ideal{\{\bv\}}}=\bv$. Hence, we can rewrite the relation $R$ as $R=\bigcup_{\bv\in R} \Variety{\Ideal{\{\bv\}}}$. By Theorem~\ref{th:ideal_intersection}, we have the following
\begin{align}\label{eq:constr=var}
 R=\bigcup_{\bv\in R} \Variety{\Ideal{\{\bv\}}}=\Variety{\I(R(\bs))},
 \qquad \text{where } \I(R(\bs)) = \bigcap_{\bv\in R} \Ideal{\{\bv\}}.
\end{align}
Note that $\I(R(\bs))\subseteq \Field[\bs]$ is zero-dimensional as its variety is finite. Moreover, $\I(R(\bs))$ is a radical ideal since it is the intersection of radical ideals (see~\cite{Cox}, Proposition~16, p.197). Equation~\eqref{eq:constr=var} means that relation $R(\bs)$ is a variety of $\Field^k$. Similar to \cite{Mastrolilli19}, the following is a generating system for $\I(R(\bs))$:
 \begin{align}\label{eq:genIConstr}
\I(R(\bs))=\langle\prod_{\bv\in R}(1-\prod_{j=1}^{k}\delta_{v_j}(x_{i_j})),\prod_{a\in D}(x_{i_1}-a),
   \ldots,\prod_{a\in D}(x_{i_k}-a)\rangle,
 \end{align}
where $\delta_{v_j}(x_{i_j})$ are indicator polynomials, i.e.\ equal to 1 when $x_{i_j}=v_j$ and 0 when $x_{i_j}\in D\setminus\{v_j\}$; polynomials $\prod_{a\in D}(x_{i_k}-a)$ force variables to take values in $D$ and will be called the {\emph{domain polynomials}}. Including a domain polynomial for each variable has the advantage that it ensures that the ideals generated by our systems of polynomials are radical (see Lemma 8.19 of \cite{becker93grobner}).

Recall that $X$ is a set of variables and $\bs\sse X$. Let $\I^X(R(\bs))$ be the smallest ideal (with respect to inclusion) of $\Field[X]$ containing $\I(R(\bs))\subseteq \Field[\bs]$. This is called the \emph{$\Field[X]$-module} of $\I(R(\bs))$. The set $\Sol(\cP)\subset \Field^X$ of solutions of $\cP=(X,D,C)$ is the intersection of the varieties of the constraints:
\begin{align}
  \Sol(\cP) & = \bigcap_{\ang{\bs,R}\in C}\Variety{\I(R(\bs))^{\Field[X]}}=\Variety{\I(\cP)},\qquad\text{where} \label{eq:solVar}\\
  \I(\cP)&=\sum_{\ang{\bs,R}\in C}\I^X(R(\bs)).\label{eq:IC}
\end{align}

The equality in~\eqref{eq:IC} follows by Theorem~\ref{th:ideal_intersection}. Finally, by Proposition 3.22 of \cite{vandongenPhd}, the ideal $\I(\cP)$ is radical. Note that in the general version of Hilbert's Nullstellensatz, Theorems~\ref{W-Hilbert-Null} and \ref{S-Hilbert-Null}, it is necessery to work in an algebraically closed field. However, in our case it is not needed due to the presence of domain polynomials. Hence, by \Cref{W-Null} and \Cref{S-Hilbert-Null}, we have the following properties. 

\begin{theorem}\label{th:nullstz}
Let $\cP$ be an instance of the $\CSP(\Gamma)$ and $\I(\cP)$ defined as in~\eqref{eq:IC}. Then
  \begin{align}
    &\Variety{\I(\cP)}=\emptyset \Leftrightarrow 1\in \I(\cP) \Leftrightarrow \I(\cP)=\Field[X],  \tag{Weak Nullstellensatz}\label{eq:weak_nstz}\\
    &\I(\Variety{\I(\cP)})=\sqrt{\I(\cP)},\tag{Strong Nullstellensatz}\label{eq:strong_nstz}\\
    &\sqrt{\I(\cP)}=\I(\cP).\tag{Radical Ideal}\label{eq:ICradical}
  \end{align}
\end{theorem}

\subsection{The Ideal Membership Problem}

In the general Ideal Membership Problem we are given an ideal $\I\sse\Field[\vc xn]$, usually by some finite generating set, and a polynomial $f_0$. The question then is to decide whether or not $f_0\in\I$. If $\I$ is given through a CSP instance, we can be more specific.

\begin{definition}\label{def:imp}
The {\sc Ideal Membership Problem} associated with a constraint language $\Gamma$  over a set $D$ is the problem $\IMP(\Gamma)$ in
which the input is a pair $(f_0,\cP)$ where $\cP = (X, D, C)$ is a
$\CSP(\Gm)$ instance and $f_0$ is a polynomial from $\Field[X]$. The goal is to decide
whether $f_0$ lies in the ideal $\I(\cP)$. We use
$\IMP_d(\Gamma)$ to denote $\IMP(\Gamma)$ when the input polynomial
$f_0$ has degree at most $d$.
\end{definition}

As $\I(\cP)$ is radical, by the Strong Nullstellensatz an equivalent way to solve the membership problem $f_0 \in \I(\cP)$ is to answer the following question:

\begin{quote}
Does there exist an $\mb{a}\in \Variety{\I(\cP)}$ such that $f_0(\mb{a})\neq 0$?
\end{quote}
In the \textbf{yes} case, such an $\mb{a}$ exists if and only if $f_0\not\in \mb{I}(\Variety{\I(\cP)})$ and therefore $f_0$ is \textbf{not} in the ideal $\I(\cP)$. This observation also implies

\begin{lemma}\label{lem:imp-0}
For any constraint language $\Gm$ the problem $\IMP_0(\Gm)$ is equivalent to $\mathsf{not\text{-}}\CSP(\Gm)$.
\end{lemma}

As was observed in the Introduction, $\IMP(\Gm)$ belongs to \textbf{coNP} for any $\Gm$ over a finite domain. We say that $\IMP(\Gm)$ is \emph{tractable} if it can be solved in polynomial time. We say that $\IMP(\Gm)$ is \emph{$d$-tractable} if $\IMP_d(\Gm)$ can be solved in polynomial time for every $d$. 

\section{Expanding the constraint language}\label{sec:IMPs}

In this section we discuss constructions on relations that allow us to reduce one \IMP~with a fixed constraint language to another. First we show that adding so-called \emph{constant} relations does not change the complexity of the problem. Second, we will consider languages on the same domain, and prove that \emph{primitive positive} (\emph{pp-}, for short) \emph{definitions} between constraint languages provides a reduction between the corresponding \IMP s. Third, we will turn our attention to the case where two languages are defined on different domains. In this case, we study a stronger notion called \emph{primitive positive interpretability}. We prove that if a language $\Gm$ pp-interprets a language $\Dl$, then $\IMP(\Dl)$ is reducible to $\IMP(\Gm)$. Finally, we discuss how ideal membership proofs can (or cannot) be recovered under these reductions.

\subsection{Constant relations and the search problem}\label{sec:constants}

We start with expansion of a constraint language $\Gm$ on a set $D$ by  constant relations. A \emph{constant} relation $\rel_a$, $a\in D$, is the unary relation, that is, a subset of $D$, that contains just one element $a$. Using it in a CSP is equivalent to \emph{pinning} a variable to a fixed value $a$. Expansion by constant relations is very important for CSPs. It preserves the complexity of the decision version of the problem when $\Gm$ is a \emph{core}, see, \cite{BulatovJK05}), and it preserves the complexity of the counting version of the problem for any $\Gm$, see \cite{Bulatov07:towards}. For $A\sse D$ let $\Gm^A$ denote $\Gm\cup\{\rel_a\mid a\in A\}$. We also call constraints of the form $\ang{x,\rel_a}$ \emph{pinning} constraints. 

\begin{proposition}\label{pro:adding-constant}
Let $\Gm$ be a constraint language on a set $D$ and $A\sse D$. Then $\IMP_d(\Gm^A)$ is polynomial time reducible to $\IMP_{d+|A|(|D|-1)}(\Gm)$.
\end{proposition}

\begin{proof}
Let $\cP=(X,D,C)$ be an instance of $\CSP(\Gm^A)$ and let $\I(\cP)$ be the ideal corresponding to $\cP$. Suppose $(f_0,\cP)$ is an instance of $\IMP_d(\Gm^A)$ where we want to decide if $f_0\in \I(\cP)$. First, we perform some preprocessing of $(f_0,\cP)$. Note that if $\cP$ contains constraints $\langle x,\rel_a\rangle$, $\langle x,\rel_b\rangle$, $a\ne b$, then $\cP$ has no solution and so $1\in\I(\cP)$ implying $f_0\in\I(\cP)$. Let $X_a$ denote the set of variables $x$, for which there is a constraint $\langle x,\rel_a\rangle\in C$. Introduce new variable $x_a$ for each $a\in A$ and replace every $x\in X_a$, with $x_a$ in both $f_0$ and $\cP$. In particular, let 
\[
X'= \left(X\setminus\bigcup_{a\in A}X_a\right)\cup\{x_a\mid a\in A\}.
\] 
The resulting instance $(f'_0,\cP')$ has the following properties:
\begin{enumerate}
    \item [--]
    The solutions of $\cP$ and $\cP'$ are in one-to-one correspondence, since for every solution $\vf$ of $\cP$ we have $\vf(x)=a$ for each $x\in X_a$, and so the mapping $\vf':X'\to D$ such that $\vf(x_a)=a$ for $a\in A$ and $\vf'(x)=\vf(x)$ otherwise is a solution of $\cP'$ and vice versa.
    \item[--] 
    $f_0(\vf)=0$ if and only if $f'_0(\vf')=0$.
\end{enumerate}
 
Now let $\cP^*=(X',D,C^*)$ be an instance of $\CSP(\Gm)$ where $C^*$ consists of all constraint from $C'$ except the ones of the form $\langle x,\rel_a\rangle$, $a\in A$. We define a new polynomial $f^*_0$ as follows. 
\[
f^*_0 = \left(\prod\limits_{a\in A}\prod\limits_{b\in D\setminus \{a\}}(x_a-b)\right)\cdot f'_0.
\]
Observe that, for any $a\in A$ and $\vf^*:X'\to D$, if $\vf^*(x_a)\ne a$ then $f^*_0(\vf^*)=0$. Suppose $\vf'\in\mb{V}(\I(\cP'))$. As $\vf'$ satisfies all the pinning constraints in $C'$, we have $f'_0(\vf')\neq 0$ if and only if $f^*_0(\vf')\neq 0$. Moreover, suppose $\vf^*\in\mb{V}(\I(\cP^*)$  and $f^*_0(\vf^*)\neq 0$. This implies that 
\begin{enumerate}
        \item 
        $\prod\limits_{a\in A}\prod\limits_{b\in D\setminus \{a\}}(\vf^*(x_a)-b) \neq 0$, which means $\vf^*$ satisfies all the pinning constraints in $C$, and hence $\vf^*\in \mb{V}(\I(\cP'))$, and
        \item 
        $f'_0(\vf')\neq 0$.
    \end{enumerate} 
    Combining the preprocessing step with the second one there exists $\vf\in \mb{V}(\I(\cP))$ such that $f_0(\vf)\neq 0$ if and only if there exists $\vf^*\in \mb{V}(\I(\cP^*)$ such that $f^*_0(\vf^*)\neq 0$. This completes the proof of the proposition.
\end{proof}

Proposition~\ref{pro:adding-constant} together with the fact that $\IMP(\Gm)$ is a subproblem of $\IMP(\Gm^A)$ implies a close connection between the complexity of $\IMP(\Gm)$ and $\IMP(\Gm^A)$.

\begin{corollary}\label{cor:constants-comlexity}
For any constraint language $\Gm$ on $D$ and any $A\sse D$, the problem $\IMP(\Gm^A)$ is tractable ($d$-tractable) if and only if $\IMP(\Gm)$ is tractable ($d$-tractable), and $\IMP(\Gm^A)$ is \coNPc\ if and only if $\IMP(\Gm)$ is \coNPc.
\end{corollary}

Recall that $\Gm^*=\Gm^D$. Then 

\begin{theorem}\label{the:adding-constants}
For any $\Gm$ over $D$ the problem $\IMP(\Gm^*)$ is polynomial time reducible to $\IMP(\Gm)$, and for any $d$ the problem $\IMP_d(\Gm^*)$ is polynomial time reducible to $\IMP_{d+|D|(|D|-1)}(\Gm)$.
\end{theorem}

However, Proposition~\ref{pro:adding-constant} leaves some room for possible complexity of $\IMP_d(\Gm)$ for small $d$, less than $|D|(|D|-1)$.

\begin{example}\label{exa:small-d}
Fix $D$, $\ell\le|D|$. Let $\NEQ_s$, $s\le|D|$ denote the $s$-ary disequality relation on $D$ given by 
\[
\NEQ_s=\{(\vc as)\mid |\{\vc as\}|=s\}.
\]
In particular, $\CSP(\NEQ_2)$ is equivalent to $|D|$-Coloring. Now let a $(\ell+2)$-ary relation $\rel$ be defined as follows
\[
\rel=(\NEQ_2\times\NEQ_\ell)\cup\{(\vc a{2+\ell})\mid |\{\vc a{2+\ell}\}|<\ell\},
\]
and let $\Gm=\{\rel\}$. It is easy to see that $\CSP(\Gm)$ is polynomial time, as assigning the same value to all variables always provides a solution. As we observed in Lemma~\ref{lem:imp-0} this implies that $\IMP_0$ is also easy. Actually, $f_0$ of degree 0 never belongs to the ideal except $f_0=0$. 

It can also be shown that for any $A\sse D$ with $|A|<\ell$ assigning a constant $a\in A$ to all variables except those bound by the pinning constraints is also a solution of $\CSP(\Gm^A)$. Therefore, $\IMP_0(\Gm^A)$ is easy for any such set. On the other hand, if $|A|=\ell$, say, $A=\{\vc a\ell\}$, then $\CSP(\Gm^A)$ can simulate $|D|$-Coloring by using $\rel(x,y,\vc a\ell)$. (This will be made more precise in Section~\ref{sec:pp-definitions}.) Therefore, $\IMP_0(\Gm^A)$ is \coNPc\ in this case. Clearly that playing with the exact definition of $\rel$ one can construct a language $\Gm$ such that $\IMP_0(\Gm^A)$ becomes \coNPc\ for any specified collection of subsets $A$ while remains easy for the rest of the subsets.
\end{example}

In the case of a 2-element $D$ we can show a more definitive result.

\begin{proposition}[see also \cite{Mastrolilli19}]\label{pro:2-element-constants}
Let $\Gm$ be a constraint language on the set $\{0,1\}$. Then
\begin{itemize}
\item[(1)]
$\IMP_d(\Gm^*)$ is polynomial time equivalent to $\IMP_{d+2}(\Gm^*)$.
\item[(2)]
$\IMP_0(\Gm)$ is polynomial time [\coNPc] if and only if  $\CSP(\Gm)$ is polynomial time [\textbf{NP}-complete].
\item[(3)]
If $\CSP(\Gm^{\{0\}})$ or $\CSP(\Gm^{\{1\}})$ is \textbf{NP}-complete then $\IMP_1(\Gm)$ is \coNPc.
\end{itemize}
\end{proposition}

Items (1),(3) follow from Proposition~\ref{pro:adding-constant} and item (2) follows from Lemma~\ref{lem:imp-0}. Moreover, replacing the relation $\NEQ_2$ in Example~\ref{exa:small-d} with the \textsc{Not-All-Equal} relation, one can construct constraint languages $\Gm$ such that the borderline between easiness and hardness in the sequence $\IMP_0(\Gm),\IMP_1(\Gm),\IMP_2(\Gm)$ lies in any desirable place.

Proposition~\ref{pro:adding-constant} also provides a connection between the decision version of the $\IMP$ and its search version. Since $\IMP(\Gm)$ is in \textbf{coNP}, here by the search $\IMP$ we understand the following problem. Let $(f_0,\cP)$ be an instance of $\IMP(\Gm)$ such that $f_0\not\in\I(\cP)$, the problem is to find an assignment $\vf\in \mb{V}(\I(\cP))$ such that $f_0(\vf)\neq 0$.  

\begin{corollary}\label{cor:search}
A decision problem $\IMP(\Gamma)$ is tractable [$d$-tractable] if and only if the corresponding search problem is tractable [$d$-tractable].
\end{corollary}

\begin{proof}
One direction is trivial as the tractability of the search problem implies the  tractability of the corresponding decision problem.
    
For the converse, let $\Gm$ be a constraint language over a finite set $D$ such that $\IMP(\Gm)$ is \mbox{($d$-)} tractable. Consider $(f_0,\cP)$, an instance of $\IMP(\Gm)$, where $\cP=(X,D,C)$ is an instance of $\CSP(\Gm)$. By the choice of $\Gm$, we can decide in polynomial time whether there exists $\vf$ such that $\vf\in \mb{V}(\I(\cP))$ but  $f_0(\vf)\neq 0$. Suppose such $\vf$ exists and hence $f_0\not\in \I(\cP)$. Then for each $x \in X$ there must be some $a \in D$, for which in the following instance $(f'_0,\cP')$ of the \IMP~we have $f'_0\not\in \I(\cP')$: 
\begin{enumerate}
\item 
define $f'_0$ to be the polynomial obtained from $f_0$ by substituting $a$ for $x$.
\item 
define $\cP'$ with $\cP'=(X,D,C'=C\cup \{\langle x,\{a\}\rangle\})$.
\end{enumerate}
Checking whether $f'_0\in\I(\cP')$ is an instance of $\IMP(\Gm^*)$ and therefore can be done in polynomial time.  Hence, by considering each possible value $a \in D$ we can find a value for $x$ that is a part of $\vf\in\mb{V}(\I(\cP))$ such that $f_0(\vf)\neq 0$. Repeating the process for each variable in turn we can find a required $\vf$. The algorithm requires solving at most $|X|\cdot|D|$ instances of $\IMP(\Gm^*)$, each of which can be solved in polynomial time. 
\end{proof}

\subsection{Primitive positive definability}\label{sec:pp-definitions}

One of the most useful reductions between CSPs is by means of primitive-positive definitions. 

\begin{definition}[pp-definability]
\label{def:pp-def}
    Let $\Gm,\Dl$ be constraint languages on the same set $D$. We say that $\Gm$ pp-defines $\Dl$ (or $\Dl$ is pp-definable from $\Gm$) if for each relation (predicate) $R\subseteq D^k$ in $\Dl$ there exists a first order formula $L$  over variables $\{x_1,\dots, x_m,x_{m+1},\dots,x_{m+k}\}$ that uses predicates from $\Gm$, equality relations, and conjunctions such that 
    \[
        R(x_{m+1},\dots,x_{m+k})= \exists x_1 \dots \exists x_m  L
    \]
Such an expression is often called a \emph{primitive positive (pp-) formula}.
\end{definition}

Mastrolilli showed that there is an analogue of existential quantification on the IMP side.

\begin{definition}
Given $\I=\langle f_1,\dots,f_s\rangle \subseteq \Field[X]$, for $Y\sse X$, the $Y$-elimination ideal $\I_{X\setminus Y}$ is the ideal of $\Field[X\setminus Y]$ defined by 
\[
\I_{X\setminus Y} =\I\cap \Field[X\setminus Y]
\]
In other words, $\I_{X\setminus Y}$ consists of all consequences of $f_1 = \dots = f_s = 0$ that do not depend on variables from $Y$.
\end{definition}

\begin{theorem}[\cite{Mastrolilli19}]\label{extension-theorem}
    Let $\cP=(X,D,C)$ be an instance of the \CSP$(\Gamma)$, and let $\I(\cP)$ be the corresponding ideal. For any $Y\sse X$ let $\I_Y$ be the $(X\setminus Y)$-elimination ideal. Then, for any partial solution $\vf_Y\in \mb{V} (\I_Y)$ there exists an extension $\psi: X\setminus Y$ such that $(\vf,\psi) \in \mb{V} (\I(\cP))$.
\end{theorem}

Let $\Gm,\Dl$ be constraint languages on the same domain $D$ such that $\Gm$ pp-defines $\Dl$. This means that for every relation $\rel$ from $\Dl$ there is a pp-definition in $\Gm$
\[
R(x_{m_\rel+1},\dots,x_{m_\rel+k_\rel})= \exists x_1 \dots \exists x_{m_\rel}  L_\rel.
\]
Suppose $\mathcal{P}_{\Dl}=(X,D,C)$ is an instance of $\CSP(\Dl)$. This instance can be converted into an instance $\cP_\Gm=(X',D,C')$ of $\CSP(\Gm)$, see Theorem 2.16 in \cite{BulatovJK05}, in such a way that $X\sse X'$ and the instance $\cP_\Dl$ has a solution if and only if $\cP_\Gm$ does. Moreover, it can be shown that $\cP_\Gm,\cP_\Dl$ satisfy the following condition.
\begin{quote}
\label{extention-condition}
{\bf The Extension Condition.} Every solution of $\cP_\Dl$ can be extended to a solution of $\cP_\Gm$, and, vice versa, the restriction of every solution of $\cP_\Gm$ onto variables from $X$ is a solution of $\cP_\Dl$. 
\end{quote}
As usual, let $\I(\cP_\Dl)$ be the ideal of $\Field[X]$ corresponding to $\cP_\Dl$ and $\I(\cP_\Gm)$ the ideal of $\Field[X']$ corresponding to $\cP_\Gm$. We would like to relate the set of solutions of $\cP_\Dl$ to the variety of the $X'\setminus X$-elimination ideal of $\I(\cP_\Gm)$. The next lemma states that the variety of the $X'\setminus X$-elimination ideal of $\I(\cP_\Gm)$ is equal to the the variety of $\I(\cP_\Dl)$.  

\begin{lemma}[\cite{Mastrolilli19}, Lemma 6.1, paraphrased]\label{variety-m-elimination}
Let $\I_X= \I(\cP_\Gm)\cap \Field [X]$ be the $X'\setminus X$-elimination ideal of $\I(\cP_\Gm)$. Then  
\[
\mb{V}(\I(\cP_\Dl)) = \mb{V}(\I_X).
\]
\end{lemma}

We can now prove a reduction for pp-definable constraint languages.

\begin{theorem}\label{pp-define-reduction}
If $\Gm$ pp-defines $\Dl$, then $\IMP(\Dl)$ [$\IMP_d(\Dl)$] is polynomial time reducible to $\IMP(\Gm)$ [respectively, to $\IMP_d(\Gm)$].
\end{theorem}

\begin{proof}
    Let $(f_0,\cP_\Dl)$, $\cP_\Dl=(X,D,C_\Dl)$, be an instance of $\IMP(\Dl)$ where $X=\{x_{m+1},\dots, x_{m+k}\}$, $f_0\in \Field[x_{m+1},\dots,x_{m+k}]$, $k=|X|$, and  $m$ will be defined later, and $\I(\cP_\Dl)\subseteq \Field[x_{m+1},\dots,x_{m+k}]$. From this we construct an instance $(f'_0,\cP_\Gm)$ of $\IMP(\Gm)$ where $f'_0\in \Field[x_{1},\dots,x_{m+k}]$ and $\I(\cP_\Gm)\subseteq \Field[x_{1},\dots,x_{m+k}]$ such that $f_0\in \I(\cP_\Dl)$ if and only if $f'_0\in \I(\cP_\Gm)$.

 Using pp-definitions of relations from $\Dl$ we convert the instance $\cP_\Dl$ into an instance $\cP_{\Gm}=(\{x_{1},\dots,x_{m+k}\},D,C_\Gm)$ of $\CSP(\Gm)$ such that every solution of $\cP_\Dl,\cP_\Gm$ satisfy the Extension Condition above. Such an instance $\cP_\Gm$ can be constructed in polynomial time as follows.
    
    By the assumption each $S\in\Dl$, say, $t_S$-ary, is pp-definable in $\Gm$ by a pp-formula involving relations from $\Gm$ and the equality relation, $=_D$. Thus,
    \[
        S(y_{q_S+1},\dots,y_{q_s+t_S})=\exists \vc y{q_S} (R_1(w^1_1,\dots, w^1_{l_1})\wedge\dots\wedge R_r(w^r_1,\dots, w^r_{l_r})),
    \]
where $w^1_1,\dots, w^1_{l_1},\dots,w^k_1,\dots, w^k_{l_k}\in \{\vc y{m_S+t_S}\}$ and $\vc Rr\sse\Gm\cup\{=_D\}$. 

Now, for every constraint $B=\langle \bs,S\rangle\in C_\Dl$, where $\bs=(x_{i_1},\dots,x_{i_t})$ create a fresh copy of $\{\vc y{q_S}\}$ denoted by $Y_B$, and add the following constraints to $C_\Gm$
    \[
         \langle (w^1_1,\dots, w^1_{l_1}),R_1\rangle,\dots, \langle (w^r_1,\dots, w^r_{l_r}),R_r\rangle.
    \]
We then set $m=\sum_{B\in C}|Y_B|$ and assume that $\cup_{B\in C}Y_B=\{\vc xm\}$. Note that the problem instance obtained by this procedure belongs to $\CSP(\Gm\cup \{=_D\})$. All constraints of the form $\langle(x_i, x_j),=_D\rangle$  can be eliminated by replacing all occurrences of the variable $x_i$ with $x_j$.  Moreover, it can be checked (see also Theorem 2.16 in \cite{BulatovJK05}) that $\cP_\Dl,\cP_\Gm$ satisfy the Extension Condition. 
    
    Let $\I(\cP_\Gm)\subseteq\Field[x_{1},\dots,x_{m+k}]$ be the ideal corresponding to $\cP_{\Gm}$ and set $f'_0=f_0$. Since $f_0\in \Field[x_{m+1},\dots,x_{m+k}]$ we also have $f_0\in \Field[x_{1},\dots,x_{m+k}]$. Hence, $(f_0,\cP_\Gm))$ is an instance of $\IMP(\Gm)$. We prove that $f_0\in \I(\cP_\Dl)$ if and only if $f_0\in \I(\cP_\Gm)$.
    
    Suppose $f_0\not\in \I(\cP_\Dl)$, this means there exists $\vf\in \mb{V}(\I(\cP_\Dl))$ such that $f(\vf)\neq 0$. By Theorem~\ref{extension-theorem}, $\vf$ can be extended to a point $\vf'\in\mb{V}(\I(\cP_\Gm))$. This in turn implies that $f_0\not\in \I(\cP_\Gm)$. Conversely, suppose $f_0\not\in \I(\cP_\Gm)$. Hence, there exists $\vf'\in\mb{V}(\I(\cP_\Gm))$ such that $f_0(\vf')\neq 0$. Projection of $\vf'$ to its last $k$ coordinates gives a point $\vf\in \mb{V}(\I_X)$. By Lemma~\ref{variety-m-elimination}, $\vf\in \mb{V}(\I(\cP_\Dl))$ which implies $f_0\not\in \I(\cP_\Dl)$.
\end{proof}

\begin{remark}
    The smallest set of all relations pp-defined from a constraint language $\Gamma\subseteq \mb{R}_A$ is called the \emph{relational clone} of $\Gm$, denoted by $\langle\Gamma\rangle$. Hence, as a corollary to \Cref{pp-define-reduction}, for a finite set of relations $\Gamma$, $\IMP(\Gamma)$ is tractable [$d$-tractable] if and only if $\IMP(\Dl)$ is tractable [$d$-tractable] for any finite $\Dl\sse\langle\Gamma\rangle$.
    Similarly, $\IMP(\Gamma)$ is \coNPc\ if and only if $\IMP(\Dl)$ is \coNPc\ for some finite $\Dl\sse\langle\Gamma\rangle$.
\end{remark}

\subsection{Primitive positive interpretability}\label{sec:pp-interpretations}

Pp-definability is a useful technique that tells us what additional relations can be added to a constraint language without changing the complexity of the corresponding problem class, and provides a tool for comparing different languages on the same domain. Next, we discuss a more powerful tool that can be used to compare the complexity of the IMP for languages over different domains.

\begin{definition}[pp-interpretability]\label{pp-interpret}
    Let $\Gm,\Dl$ be constraint languages over finite domains $D,E$, respectively, and $\Dl$ is finite. We say that $\Gm$ pp-interprets $\Dl$ if there exists a natural number $\ell$, a set $F \subseteq D^\ell$, and an onto mapping $\pi : F \to E$ such that $\Gm$ pp-defines the following relations
    \begin{enumerate}
        \item the relation $F$,
        \item the $\pi$-preimage of the equality relation on $E$, and
        \item the $\pi$-preimage of every relation in $\Dl$,
    \end{enumerate}
    where by the $\pi$-preimage of a $k$-ary relation $S$ on $E$ we mean the $\ell k$-ary relation $\pi^{-1}(S)$ on $D$ defined by
        \[
            \pi^{-1}(S)(x_{11},\ldots , x_{1k},x_{21},\ldots,x_{2k},\ldots,x_{\ell1},\ldots,x_{\ell k})\qquad \text{is true}
        \]
    if and only if
        \[
            S(\pi(x_{11},\dots,x_{\ell1}),\dots,\pi(x_{1k},\dots,x_{\ell k})) \qquad\text{is true}.
        \]
\end{definition}

\begin{example}\label{exa:partial-order}
    Suppose $D=\{0,1\}$ and $E=\{0,1,2\}$ and define relations $R_D=\{(0,0),(0,1),(1,1)\}$ and $R_E=\{(0,0),(0,1),(0,2),(1,1),(1,2),(2,2)\}$. Note that relations $R_D,R_E$ are orders $0\le1$ and $0\le1\le2$ on $D,E$, respectively. Set $\Gm=\{R_D\}$ and $\Dl=\{R_E\}$.
    
    Let $n=2$ and define $F=\{(0,0),(0,1),(1,1)\}\subseteq D^2$. The language $\Gm$ pp-defines $F$ i.e. $F=\{(x,y)\mid x\leq y \text{ and } x,y\in\{0,1\}\}$. Now define mapping $\pi:F\to E$ as follows $\pi((0,0))=0,\pi((0,1))=1,\pi((1,1))=2$. The $\pi$-preimage of the relation $R_E$ is the relation 
    \[
        R_F=\{(0,0,0,0),(0,0,0,1),(0,0,1,1),(0,1,0,1),(0,1,1,1),(1,1,1,1)\}.
    \]
 The language $\Gm$ pp-defines $\Gm'=\{R_F\}$ through the following pp-formula 
    \[
        R_F=\{(x_1,x_2,y_1,y_2)\mid (x_1\leq y_1)\land (x_2\leq y_2)\land (x_1\le x_2)\land (y_1\le y_2),  \text{ and } (x_1,x_2,y_1,y_2)\in\{0,1\}^4\}.
    \]
    Consider instance $(\{x,y,z\},E,C)$ of $\CSP(\Dl)$ where the set of constraints is $C=\{\langle (x,y),R_E\rangle,\langle (y,z),R_E\rangle\}$. This basically means the requirements $(x\leq y)\land (y\leq z)$. This instance is equivalent to the following instance of $\CSP(\Gm')$:
    \begin{align}
        &\langle (x_1,x_2,y_1,y_2),R_F\rangle \land 
        \langle (y_1,y_2,z_1,z_2),R_F\rangle 
    \end{align}
    As was pointed out, $\Gm$ pp-defines $F$ as well as $R_F$. Hence, we define $(\{x_1,x_2,y_1,y_2,z_1,z_2\},D,C')$, an instance of $\CSP(\Gm)$, with the constraints
    \begin{align}
    \label{Ex:pp1}
        &\langle (x_1,x_2),R_D\rangle \land 
        \langle (y_1,y_2),R_D\rangle \land
        \langle (z_1,z_2),R_D\rangle \\
        \label{Ex:pp2}
        &\land
        \langle (x_1,y_1),R_D\rangle \land
        \langle (x_2,y_2),R_D\rangle \\
        \label{Ex:pp3}
        &\land
        \langle (y_1,z_1),R_D\rangle \land
        \langle (y_2,z_2),R_D\rangle 
    \end{align}
    Note that \eqref{Ex:pp1} is a pp-definition of relation $F$ and forces $(x_1x_2),(y_1y_2),(z_1z_2)\in F$. Moreover, equations \eqref{Ex:pp2} and \eqref{Ex:pp3} are equivalent to
    \begin{align}
        (x_1\leq y_1)\land (x_2\leq y_2)\land (y_1\leq z_1)\land (y_2\leq z_2)
    \end{align}
    Applying the mapping $\pi$, every solution of the instance $(\{x_1,x_2,y_1,y_2,z_1,z_2\},D,C')$ can be transformed to a solution of instance $(\{x,y,z\},E,C)$ and back.
\end{example}

 One can interpolate mapping $\pi$ in \Cref{pp-interpret} by a polynomial of low degree. It is a known fact that given $N+1$ distinct $\mb{x}_0,\dots,\mb{x}_N\in \mathbb{R}^\ell$ and corresponding values $y_0,\dots,y_N$, there exists a polynomial $p$ of degree at most $\ell N$ that interpolates the data i.e. $p(\mb{x}_j)=y_j$ for each $j\in\{0,\dots,N\}$ (such a polynomial can be obtained by a straightforward generalization of the Lagrange interpolating polynomial, see, e.g., \cite{phillips2003interpolation}). Hence, we can interpolate the mapping $\pi$ by a polynomial of total degree at most $\ell|E|$.

\begin{theorem}\label{the:pp-interpretability}
Let $\Gm,\Dl$ be constraint languages on sets $D,E$, respectively, and let $\Gm$ pp-interprets $\Dl$.   Then $\IMP_d(\Dl)$ is polynomial time reducible to  $\IMP_{\ell|E|}(\Gm)$.
\end{theorem}

\begin{proof}
    Let $(f_0,\cP_\Dl)$ be an instance of $\IMP_d(\Dl)$ where $f_0\in \Field[x_{1},\dots,x_n]$, $\cP_{\Dl}=(\{x_{1},\dots,x_n\},E,C_\Dl)$, an instance of $\CSP(\Dl)$, and $\I(\cP_\Dl)\subseteq \Field[x_{1},\dots,x_n]$. 

    The properties of the mapping $\pi$ from Definition~\ref{pp-interpret} allow us to rewrite an instance of $\CSP(\Dl)$ to an instance of $\CSP(\Gamma')$ over the constraint language $\Gm'$. Recall that, by \Cref{pp-interpret}, $\Gm'$ contains all the $\ell k$-ary relations $\pi^{-1}(S)$ on $D$ where $S\in\Dl$ is $k$-ary relation.   
    
    Note that $\Gm'$ is pp-definable from $\Gm$. By Theorem~\ref{pp-define-reduction}, $\IMP(\Gamma')$ is reducible to $\IMP(\Gm)$. It remains to show $\IMP_d(\Dl)$ is reducible to $\IMP_d(\Gamma')$. To do so, from instance $(f_0,\cP_\Dl)$ of $\IMP_d(\Dl)$ we construct an instance $(f'_0,\cP_{\Gm'})$ of $\IMP_d(\Gamma')$ such that $f_0\in\I(\cP_\Dl)$ if and only if $f'_0\in\I(\cP_{\Gm'})$. 
    
    Let $p$ be a polynomial of total degree at most $\ell|E|$ that interpolates mapping $\pi$. For every $f_0\in \Field[x_1,\dots,x_n]$, let $f'_0\in \Field [x_{11},\ldots , x_{\ell1},\ldots,x_{1n},\ldots,x_{\ell n}]$ be the polynomial that is obtained from $f_0$ by replacing each indeterminate $x_i$ with $p(x_{1i},\dots,x_{\ell i})$. 
    Clearly, for any assignment $\vf:\{\vc xn\}\to E$, $f_0(\vf)=0$ if and only if $f'_0(\psi)=0$ for every $\psi:\{x_{11}\zd x_{\ell n}\}\to D$
such that 
\[
\vf(x_i)=\pi(\psi(x_{1i}),\dots,\psi(x_{\ell i}))
\]
for every $i\le n$. Moreover, 
for any such $\vf,\psi$ it holds $\vf\in \mb{V}(\I(\cP_\Dl))$ if and only if  $\psi\in \mb{V}(\I(\cP_{\Gm'}))$. This yields that 
    \[
        (\exists \vf\in \mb{V}(\I(\cP_\Dl)) \land f_0(\vf)\neq 0) \iff (\exists \psi\in \mb{V}((\cP_{\Gm'})) \land f'_0(\psi)\neq 0)
    \]
Note that the condition that $f_0$ has bounded degree is important here, because otherwise $f'_0$ may have exponentially more monomials than $f_0$. This completes the proof of the theorem.
\end{proof}

\section{Polymorphisms and algebras}

\subsection{Polymorphisms and a necessary condition for tractability}
\label{sec:polymorphisms}

Sets of relations closed under pp-definitions allow for a succinct representation through polymorphisms. Let $R$ be a $k$-ary relation on a set $D$ and $\psi$ an $n$-ary operation on the same set. Operation $\psi$ is said to be a \emph{polymorphism} of $R$ if for any $\mb a^1\zd\mb a^n\in R$ the tuple $\psi(\mb a^1\zd\mb a^n)$ belongs to $R$. Here by $\psi(\mb a^1\zd\mb a^n)$ we understand the component-wise action of $\psi$, that is, if $\mb a^i=(a^i_1\zd a^i_k)$ then 
\[
    \psi(\mb a^1\zd\mb a^n)=(\psi(a^1_1\zd a^n_1)\zd\psi(a^1_k\zd a^n_k)).
\]
For more background on polymorphisms, their properties, and links to the CSP the reader is referred to a relatively recent survey \cite{Barto17:polymorphisms}. Most of the standard results we use below can be found in this survey.

A polymorphism of a constraint language $\Gm$ is an operation that is a polymorphism of every relation in $\Gm$. The set of all polymorphisms of the language $\Gm$ is denoted by $\Pol(\Gm)$. For a set $\Psi$ of operations by $\Inv(\Psi)$ we denote the set of relations $R$ such that every operation from $\Psi$ is a polymorphism of $R$. The operators $\Pol$ and $\Inv$ induces so called \emph{Galois correspondence} between sets of operations and constraint languages. There is a rich theory of this correspondence, however, for the sake of this paper we only need one fact.

\begin{theorem}\label{the:galois}
Let $\Gm,\Dl$ be constraint languages on a finite set $D$. Then $\Pol(\Gm)\sse\Pol(\Dl)$ if and only if $\Gm$ pp-defines $\Dl$. In particular, $\Inv(\Pol(\Gm))$ is the set of all relations pp-definable in $\Gm$. 
\end{theorem}

Combining Theorem~\ref{the:galois} and Theorem~\ref{pp-define-reduction}, polymorphisms of constraint languages provide reductions between \IMP s.

\begin{corollary}\label{cor:imp-polymorphisms}
Let $\Gm,\Dl$ be constraint languages on a finite set $D$ and $\Dl$ finite. If $\Pol(\Gm)\sse\Pol(\Dl)$ then $\IMP(\Dl)$ [$\IMP_d(\Dl)$] is polynomial time reducible to $\IMP(\Gm)$ [$\IMP_d(\Gm)$].
\end{corollary}

Corollary~\ref{cor:imp-polymorphisms} amounts to saying that similar to $\CSP(\Gm)$ the complexity of $\IMP(\Gm)$ is determined by the polymorphisms of $\Gm$.

Next we use the known necessary condition for CSP tractability \cite{BulatovJK05} to obtain some necessary conditions for tractability of $\IMP(\Gm)$. 

A \emph{projection} is an operation $\psi:D^k\to D$ such that there is $i\in[k]$ with $\psi(\vc xk)=x_i$ for any $\vc xk\in D$. If the only polymorphisms of a constraint language are projections, every relation is pp-definable in $\Gm$ implying that $\IMP(\Gm)$ is \coNPc. 

Theorem~\ref{the:adding-constants} is another ingredient for our necessary condition. Recall that for a language $\Gm$ by $\Gm^*$ we denote the language with added \emph{constant relations} $R_a$ for all $a\in D$.                                                                                                                                                                                                                  It is known that every polymorphism $\psi$ of all the constant relations is \emph{idempotent}, that is, satisfies the condition $\psi(x\zd x)=x$. Therefore, by Theorem~\ref{the:adding-constants} it suffices to focus on idempotent polymorphisms.

\begin{proposition}\label{pro:imp-projections}
    Let $\Gm$ be a constraint language over a finite set $D$. If the only idempotent polymorphisms of $\Gm$ are projections then $\IMP_{|D|(|D|-1)}(\Gamma)$ is \coNPc.
\end{proposition}

\begin{example}
    Consider the relation $N=\{0,1\}^3\setminus\{(0,0,0),(1,1,1)\}$. This relation corresponds to \textsc{Not-All-Equal Satisfiability} problem and it is known that the idempotent operations from $\Pol(\{N\})$ are projections \cite{post1941two}. $\CSP(\{N\})$ was shown to be \textbf{NP}-complete by Schaefer~\cite{Schaefer78} and $\IMP(\{N\})$ is shown to be \coNPc\ in \cite{Mastrolilli19}. 
\end{example}

\subsection{Algebras and a better necessary condition}
\label{sec:algebraic}

In this section we briefly review the basics of structural properties of (universal) algebra in application to the IMP. Universal algebras have been instrumental in the study of CSPs, and, although we do not go deeper into this theory in this paper, we expect they should be useful for IMPs as well. We follow textbooks \cite{Burris81:universal,mckenzie2018algebras} and texts on the algebraic theory of the CSP, see, e.g., \cite{Barto14:constraint,Barto15:constraint,BulatovJK05,Bulatov07:towards}. 

\paragraph{Algebras.}
An \emph{algebra} is a pair $\cD=(D,\Psi)$ where $D$ is a set (always finite in this paper) and $\Psi$ is a set of operations on $D$ (perhaps multi-ary). The operations from $\Psi$ are called \emph{basic}, and any operation that can be obtained from operations in $\Psi$ by means of composition is called a \emph{term} operation. The set of all term operations will be denoted by $\Term(\cD)$. For example, $\Psi$ can be the set $\Pol(\Gm)$ for some constraint language $\Gm$ on $D$, in which case $\cD$ is called the \emph{algebra of polymorphisms} of $\Gm$ and will be denoted $\Alg(\Gm)$. Thus, $\Alg(\Gm)=(D,\Pol(\Gm))$. 

By Corollary~\ref{cor:imp-polymorphisms} the algebra $\Alg(\Gm)$ determines the complexity of $\IMP(\Gm)$ and $\IMP_d(\Gm)$ for sufficiently large $d$. The advantage of using algebras rather than just polymorphisms is that it unlocks a variety of structural methods that cannot be easily applied if we only use polymorphisms. Algebra $\cD=(D,\Psi)$ is said to be \emph{tractable} [\emph{$d$-tractable}] if for any finite constraint language $\Gm$ such that $\Psi\sse\Pol(\Gm)$ the problem $\IMP(\Gm)$ is tractable [$d$-tractable]. Algebra $\cD$ is said to be \emph{\coNPc} if for some finite language $\Gm$ with $\Psi\sse\Pol(\Gm)$ the problem $\IMP(\Gm)$ is \coNPc. In the rest of this section apart from another necessary condition of tractability we prove several results that deduce the tractability [$d$-tractability, \coNPc ness] of a certain algebra derivative from $\cD$ from a similar property of $\cD$. 

\paragraph{Idempotent algebras.}
The first step will be to reduce the kind of algebras we have to study. By Theorem~\ref{the:adding-constants} idempotent polymorphisms determine the complexity of $\IMP(\Gm)$. On the algebraic side, an algebra $\cD$ is said to be \emph{idempotent} if each of its basic operations (and therefore each of its term operations) is idempotent. Every algebra $\cD=(D,\Psi)$ can be converted into an idempotent algebra simply by throwing out all the non-idempotent term operations. Let $\Term_{id}(\cD)$ denote the set of all idempotent operations from $\Term(\cD)$. Then the \emph{full idempotent reduct} of $\cD=(D,\Psi)$ is the algebra $\Id(\cD)=(D,\Term_{id}(\cD))$.

\begin{proposition}\label{pro:idempotent-reduct}
For any finite algebra $\cD=(D,\Psi)$,  $\cD$ is tractable [$d$-tractable] if and only if $\Id(\cD)$  is tractable [$d$-tractable]. Also $\Id(\cD)$ is \coNPc\ if and only if $\cD$ is  \coNPc. 
\end{proposition}

\begin{proof}
Note that an operation $\psi$ on a set $D$ is idempotent if and only if it preserves all the relations in the set $\Gm_{\text{CON}} = \{R_a \mid a \in D\}$, consisting of all constant relations $R_a$ on $D$. Hence, $\Inv(\Term_{id}(\cD))$ is the relational clone generated by $\Inv(\Psi)\cup \Gamma_{\text{CON}}$, or, in other words, every relation $R$ such that $\Term_{id}(\cD)\sse\Pol(R)$ is pp-definable in $\Inv(\Psi)\cup \Gamma_{\text{CON}}$. 

Let $\Dl$ be a finite set from $\Inv(\Term_{id}(\cD))$. By the observation above there is a finite $\Gm\sse \Inv(\Psi)\cup \Gamma_{\text{CON}}$ such that $\Gm$ pp-defines $\Dl$. By Theorem~\ref{the:adding-constants} for any $d$ the problem $\IMP_d(\Dl)$ can be reduced in polynomial time to $\IMP_{d+|D|(|D|-1)}(\Gm)$, and the result follows.
\end{proof}

\paragraph{Subalgebras, homomorphisms, and direct powers.}
The following standard algebraic constructions have been very useful in the study of the CSP.

\begin{definition}\label{def:HSP}
Let $\cD=(D,\Psi)$ be an algebra.
\begin{itemize}
\item[--]
{\bf (Subalgebra)} \ \ 
Let $E\sse D$ such that, for any $\psi\in \Psi$ and for any $b_1, \dots,b_k \in E$, where $k$ is the arity of $\psi$, we have $\psi(b_1, \dots, b_k) \in E$. In other words, $\psi$ is a polymorphism of $E$ or $E\in\Inv(\Psi)$. The algebra $\cE = (E,\Psi|_E)$, where $\Psi|_E$ consists of the restrictions of all operations in $\Psi$ to $E$, is called a \emph{subalgebra} of $\cD$. 
\item[--]
{\bf (Direct power)} \ \ 
For a natural number $k$ the \emph{$k$-th direct power} $\cD^k$ of $\cD$ is the algebra $\cD^k=(D^k,\Psi^k)$, where $\Psi^k$ consists of all the operations from $\Psi$ acting on $D^k$ component-wise (see the definition of polymorphism). 
\item[--] 
{\bf (Homomorphic image)} \ \ 
Let $E$ be a set and $\chi:D\to E$ a mapping such that for any (say, $k$-ary) $\psi\in\Psi$ and any $\vc ak,\vc bk\in D$, if $\chi(a_i)=\chi(b_i)$, $i\in[k]$, then $\chi(\psi(\vc ak))=\chi(\psi(\vc bk))$. The algebra $\cE=(E,\Psi_\chi)$ is called a \emph{homomorphic image} of $\cD$, where for every $\psi\in\Psi$ the set $\Psi_\chi$ contains $\psi/_\chi$ given by $\psi/_\chi(\vc ck)=\chi(\psi(\vc ak))$ and $\vc ak\in D$ are such that $c_i=\chi(a_i)$, $i\in[k]$. 
\end{itemize} 
\end{definition}

If an algebra is the algebra of polymorphisms of some constraint language, the concepts above are related to pp-definitions and pp-interpretations. We will use the following easy observation. 

\begin{lemma}\label{lem:HSP-pp-definitions}
Let $\cD=(D,\Psi)=\Alg(\Gm)$ for a constraint language $\Gm$ over $D$.
\begin{itemize}
\item[--]
If $\cE=(E,\Psi|_E)$ is a subalgebra of $\cD$ then $E$ is pp-definable in $\Gm$.
\item[--]
Every relation from $\Inv(\Psi^k)$ is pp-interpretable in $\Gm$.
\item[--]
Let $\cE=(E,\Psi_\chi)$ be a homomorphic image of $\cD$. Then every relation from $\Inv(\Psi_\chi)$ is pp-interpretable in $\Gm$. 
\end{itemize}
\end{lemma}

The standard algebraic constructions also include direct products of different algebras. Direct products also have a strong connection to the CSP and therefore IMP. However, they require a more general framework, multi-sorted CSPs and IMPs. These are beyond the scope of this paper.

We are now ready to prove the reductions induced by subalgebras, direct powers, and homomorphic images.

\begin{theorem}\label{the:HSP-reduction}
Let $\cD$ be an algebra and $\cE$ its subalgebra [direct power, homomorphic image]. If $\cD$ is $d$-tractable, then so is $\cE$. If $\cE$ is \coNPc, then $\cD$ is also \coNPc. Moreover, if $\cE$ is a subalgebra of $\cD$, then $\cE$ is tractable whenever $\cD$ is. 
\end{theorem}

\begin{proof}
The theorem is almost straightforward from Lemma~\ref{lem:HSP-pp-definitions} and Theorems~\ref{pp-define-reduction},~\ref{the:pp-interpretability}. Let $\cD=(D,\Psi)$, $\cE=(E,\Psi')$ and $\Dl\sse\Inv(\Psi')$, a finite set. Note that $\cD=\Alg(\Gm)$ for $\Gm=\Inv(\Psi)$. 

If $\cE$ is a subalgebra of $\cD$, that is, $\Psi'=\Psi|_E$ then by Lemma~\ref{lem:HSP-pp-definitions} $E$ is pp-definable in $\Gm$ and therefore $\Dl\sse\Inv(\Psi|_E)\sse\Inv(\Psi)$. The result follows.

If $\cE$ is a direct power, say, $\cE=(D^k,\Psi^k)$, then since every relation from $\Inv(\Psi^k)$ is pp-interpretable in $\Gm$, there is a finite set $\Gm'\sse\Gm$ that pp-interprets $\Dl$. Then by Theorem~\ref{the:pp-interpretability} $\IMP_d(\Dl)$ can be reduced to $\IMP_d(\Gm')$ in polynomial time. The result follows.

In the case when $\cE$ is a homomorphic image of $\Gm$, the proof is identical to the previous case due to Lemma~\ref{lem:HSP-pp-definitions}.
\end{proof}

\paragraph{Stronger necessary condition for tractability.}
Subalgebras, direct powers, and homomorphic images allow us to state a stronger condition for tractability of constraint languages. In the case of the CSP, when a constraint language $\Gm$ contains all the constant relations, $\CSP(\Gm)$ is \textbf{NP}-complete if and only if $\Alg(\Gm)$ has a homomorphic image $\cD$ of a subalgebra such that all the term operations of $\cD$ are projections. Using Theorem~\ref{the:adding-constants} we can make this condition even stronger (although only necessary).

\begin{theorem}\label{the:HSP-necessary} 
Let $\Gm$ be a constraint language with the property that there exists a homomorphic image $\cE$ of a subalgebra of a direct power of $\Id(\Alg(\Gm))$ such that all the term operations of $\cE$ are projections. Then $\IMP_{|D|(|D|-1)}(\Gm)$ is \coNPc.
\end{theorem}

\begin{proof}
Let $\cE=(E,\Psi)$. Since the term operations of $\cE$ are only projections, by Proposition~\ref{pro:imp-projections} there is a finite set $\Dl$ such that $\IMP_0(\Dl)$ is \coNPc. Then, by
Lemma~\ref{lem:HSP-pp-definitions}, $\Gm^*$ pp-interprets $\Dl$, and we obtain the result by Theorems~\ref{the:adding-constants} and~\ref{the:pp-interpretability}.
\end{proof} 

\section{Sufficient conditions for tractability}
\label{sec:sufficient}
\subsection{The Ideal Membership Problem and \GBs}
 A possible way to solve the \IMP~is via polynomial division. Informally, if a remainder of division of $f_0$ by generating polynomials of $\I(\cP)$ is zero then $f_0\in\I(\cP)$. Let us recall some standard notations from algebraic geometry that are needed to present a division algorithm and the notion of \GBs. We follow notation in \cite{Cox}.

 A monomial ordering $\succ$ on $\Field[\vc x n]$ is a relation $\succ$ on $\zZ_{\geq 0}^n$, or equivalently, a relation on the set of monomials $\bx^{\alpha}$, $\alpha \in \zZ_{\geq 0}^n$ (see~\cite{Cox}, Definition 1, p.55). Each monomial $\bx^\alpha=x_1^{\alpha_1}\cdots x_n^{\alpha_n}$ corresponds to an $n$-tuple of exponents $\alpha =(\alpha_1,\ldots,\alpha_n)\in \mathbb{Z}^n_{\geq0}$. This establishes a one-to-one correspondence between the monomials in $\Field[x_1,\ldots,x_n]$ and $\mathbb{Z}^n_{\geq0}$. Any ordering $\succ$ we establish on the space $\mathbb{Z}^n_{\geq0}$ will give us an ordering on monomials: if $\alpha \succ \beta$ according to this ordering, we will also say that $\bx^\alpha \succ \bx^\beta$.

\begin{definition}\label{def:lex and grlex} 
Let $\alpha =(\alpha_1,\ldots,\alpha_n),\beta=(\beta_1,\ldots,\beta_n)\in \mathbb{Z}^n_{\geq0}$ and $|\alpha| = \sum_{i=1}^n\alpha_i$, $|\beta| = ~\sum_{i=1}^n\beta_i$. \emph{Lexicographic order} and \emph{graded lexicographic order} are defined as follows. 
  \begin{enumerate}
      \item We say $\alpha\succ_\lex \beta$ if the leftmost nonzero entry of the vector difference $\alpha -\beta \in \mathbb{Z}^n$ is positive. We will write $\bx^\alpha\succ_\lex \bx^\beta$ if $\alpha\succ_\lex\beta$. 
      \item We say $\alpha\succ_\grlex \beta$ if $|\alpha| >|\beta|$, or $|\alpha| =|\beta|$ and $\alpha\succ_\lex \beta$.
  \end{enumerate}
\end{definition}

\begin{definition}
  For any $\alpha=(\alpha_1,\cdots,\alpha_n)\in \mathbb{Z}^n_{\geq0}$ let $\bx^\alpha\overset{\mathrm{def}}{=} \prod_{i=1}^{n}x_i^{\alpha_i}$. Let $f= \sum_{\alpha} a_{\alpha}\bx^\alpha$ be a nonzero polynomial in $\Field[x_1,\ldots,x_n]$ and let $\succ$ be a monomial order.
  \begin{enumerate}
    \item The \emph{multidegree} of $f$ is $\multideg(f)\overset{\mathrm{def}}{=} \max(\alpha\in \mathbb{Z}^n_{\geq0}:a_\alpha\not = 0)$.
    \item The \emph{degree} of $f$ is deg$(f)=|\multideg(f)|$ where $|\alpha| = \sum_{i=1}^n\alpha_i$. In this paper, this is always according to \textsf{grlex} order.
    \item The \emph{leading coefficient} of $f$ is $\LC(f)\overset{\mathrm{def}}{=} a_{\multideg(f)}\in \Field$.
    \item The \emph{leading monomial} of $f$ is $\LM(f)\overset{\mathrm{def}}{=} \bx^{\multideg(f)}$ (with coefficient 1).
    \item The \emph{leading term} of $f$ is $\LT(f)\overset{\mathrm{def}}{=} \LC(f)\cdot \LM(f)$.
  \end{enumerate}
\end{definition}

\begin{definition}[A division algorithm]
\label{def:division-alg}
    Let $\succ$ be a monomial order on $\zZ_{\geq 0}^n$, and let $F = \{\vc f s\}\subset\Field[\vc x n]$. Then every $f \in \Field[\vc x n]$ can be written as $f = q_1f_1 + \dots + q_s f_s + r$, where $q_i, r \in \Field[\vc x n]$, and either $r = 0$ or $r$ is a linear combination, with coefficients in 
    $\Field$, of monomials, none of which is divisible by any of $\LT(f_1),\dots , \LT( f_s)$. Furthermore, if $q_i f_i\neq 0$, then
    $\multideg(f) \succeq  \multideg(q_if_i)$. We call $r$ a \emph{remainder} of $f$ on division by $F$. Also, we say that $f$ \emph{reduces to} $r$ \emph{modulo} $F$, written $f \to_F r$.
\end{definition}
The above definition suggests the following procedure to compute a remainder. Repeatedly, choose an $f_i\in F$ such that $\LT(f_i)$ divides some term $t$ of $f$ and replace $f$ with $f-\frac{t}{\LT(f_i)}f_i$, until it cannot be further applied. Note that the order we choose the polynomial $f_i$ is not specified. Unfortunately, depending on the generating polynomials of the ideal a remainder of division may not be unique. Moreover, such a remainder depends on the order we do division. 

\begin{example}
  Let $f=x^2y-xy^2+y$ and $\I=\Ideal{f_1,f_2}$ with $f_1=x^2$ and $f_2=xy-1$. Consider the \grlex order with $x\succ_\lex y$. On one hand, $f= 0\cdot f_1+(x-y)\cdot f_2- x$. On the other hand, $f=y\cdot f_1 -y\cdot f_2+0$. 
\end{example}
The Hilbert Basis Theorem states that every ideal has a finite generating set (see, e.g., Theorem~4 on page 77 \cite{Cox}). Fortunately, for every ideal there is a finite generating set that suits our purposes. That is, there is a generating set so that the remainder of division by that set is uniquely defined, no matter in which order we do the division.

\begin{definition}\label{def:GB}
   Fix a monomial order on the polynomial ring $\Field[\vc x n ]$. A finite subset $G = \{\vc g t\}$ of an ideal $\I \subseteq \Field[\vc x n ]$ different from $\{0\}$ is said to be a \GB\ (or \emph{standard basis}) if 
   \[
    \Ideal{\LT(g_1),\dots,\LT(g_t)} = \Ideal{\LT(I)}\]
    where $\Ideal{\LT(I)}$ denotes the ideal generated by the leading terms of elements of $\I$. 
\end{definition}
\begin{definition}[$d$-truncated \GB]
\label{def:d-GB}
    If $G$ is a \GB\ of an ideal, the $d$-truncated \GB\ $G'$ of $G$ is defined as
    $$G' =G\cap \Field[\vc x n]_d,$$ where $\Field[\vc x n]_d$ is the set of polynomials of degree less than or equal to $d$.
\end{definition}
\begin{proposition}[\cite{Cox}, Proposition~1, p.83 ]
\label{prop:division-GB}
    Let $\I \subseteq \Field[\vc x n ]$ be an ideal and let $G = \{\vc g t\}$ be a \GB\ for $\I$. Then given $f \in \Field[\vc x n ]$, there is a unique $r \in \Field[\vc x n ]$ with the following two properties:
    \begin{enumerate}
        \item No term of $r$ is divisible by any of $\LT(g_1),\dots,\LT(g_t)$,
        \item There is $g\in\I$ such that $f=g+r$.
    \end{enumerate}
    In particular, $r$ is the remainder on division of $f$ by $G$ no matter how the elements of $G$ are listed when using the division algorithm.
\end{proposition}

The remainder $r$ is called the \emph{normal form of} $f$ \emph{by} $G$, denoted by $f|_G$. Note that, although the remainder $r$ is unique, even for a \GB, the "quotients" $q_i$ produced by the division algorithm $f = q_1g_1+\dots+ q_tg_t + r$ can change if we list the generators in a different order. As a corollary of \Cref{prop:division-GB}, we get the following criterion for when a given polynomial lies in an ideal.
\begin{corollary}[\cite{Cox}, Corollary~2, p.84]
    Let $G = \{\vc g t\}$ be a \GB\ for an ideal $\I \subseteq \Field[\vc x n]$ and let $f \in \Field[\vc x n]$. Then $f \in \I$ if and only if the remainder on division of $f$ by $G$ is zero.
\end{corollary}

There is a criterion, known as Buchberger's criterion, that tells us whether a given generating set of an ideal is a \GB. In order to formally express this criterion, we need to define the notion of $S$-polynomials.

\begin{definition}[S-polynomial]
\label{s-poly}
    Let $f,g\in \Field[x_1,\ldots,x_n]$ be nonzero polynomials. If $\multideg(f)=\alpha$ and $\multideg(g)=\beta$, then let $\gamma=(\gamma_1,\ldots,\gamma_n)$, where $\gamma_i = \max(\alpha_i,\beta_i)$ for each $i$. We call $x^\gamma$ the \emph{least common multiple} of $\LM(f)$ and $\LM(g)$, written $x^\gamma = \LCM(\LM(f),\LM(g))$. The \emph{$S$-polynomial} of $f$ and $g$ is the combination 
    \[
        S(f,g) = \frac{x^\gamma}{\LT(f)}\cdot f - \frac{x^\gamma}{\LT(g)}\cdot g.
    \]
\end{definition}
 
 \begin{theorem}[Buchberger's Criterion \cite{Cox}, Theorem 3, p.105]
 \label{th:crit}
     Let $\I$ be a polynomial ideal. Then a basis $G = \{\vc g t\}$ of $\I$ is a \GB\ of $\I$ if and only if for all pairs $i\neq j$, the remainder on division of $S(g_i, g_j)$ by $G$ (listed in some order) is zero.
 \end{theorem}

\begin{proposition}[\cite{Cox}, Proposition 4, p.106]
\label{prop:prime-LM}
    We say the leading monomials of two polynomials $f,g$ are relatively prime if $\LCM(\LM(f),\LM(g))=\LM(f)\cdot\LM(g)$. Given a finite set $G \subseteq \Field[\vc x n]$, suppose that we have $f,g \in G$ such that the leading monomials of $f$ and $g$ are relatively prime. Then $S(f, g)\to_G 0$.
\end{proposition}

\subsection{The dual-discriminator}\label{sec:dual-discriminator}

Here we deal with a \emph{majority} operation. Over the Boolean domain there is only one majority operation, called the \emph{dual-discriminator}. In the Boolean case, Mastrolilli~\cite{Mastrolilli19} proved that the $\IMP(\Gm)$~is tractable when the constraint language $\Gm$ is closed under the dual-discriminator operation. Later, Bharathi and Mastrolilli~\cite{Bharathi-Dual-Disc} expand this tractability result to constraint languages over the ternary domain. We establish tractability result for any finite domain. We point out that constraint languages closed under the dual-discriminator operation are also known as 0/1/all constraints. We start off with the definition of a majority polymorphism and explain an appealing structure of majority closed relations.

\begin{definition}
  Let $\mu$ be a $3$-ary operation from $D^3$ to $D$. If for all $x,y\in D$ we have $\mu(x,x,y)=\mu(x,y,x)=\mu(y,x,x)=x$, then $\mu$ is called a \emph{majority operation}.
\end{definition}

For a ($n$-ary) relation $R$ and $T\sse[n]$ by $\pr_TR$ we denote the \emph{projection} of $R$ onto $T$, that is, the set of tuples $(a_i)_{i\in T}$ such that there is $(\vc bn)\in R$ with $b_i=a_i$ for each $i\in T$.
 
\begin{proposition}[\cite{JeavonsCG97}]
\label{prop:majority-to-binary}
Let $R$ be a relation of arity $n$ with a majority polymorphism, and let $C=\langle S,R\rangle$ constraining the variables in $S$ with relation $R$. For any problem $\mc{P}$ with constraint $C$, the problem $\mc{P}'$ which is obtained by replacing $C$ by the set of constraints 
\[
    \{((S[i],S[j]),\pr_{i,j}(R)) \mid 1\leq i\leq j\leq n\}
\]
has exactly the same solutions as $\mc{P}$.
\end{proposition}

The above proposition suggests that, without loss of generality, we may assume that all the constraints are binary when a constraint language has a majority polymorphism  $\mu$. Let $\Gamma$ be a language over a set $D$ such that $\mu\in\Pol(\Gm)$, and let $\cP=(X,D,C)$ be an instance of $\CSP(\Gamma)$. We assume each constraint in $C$ is binary. That is 
\[
C=\{C_{ij}=\langle (x_i,x_j),R_{ij}\rangle\mid R_{ij}\subseteq D_i\times D_j \text{ where }D_i,D_j\subseteq D\}.
\] 

Relations closed under the dual-discriminator operation admit a much nicer structure that has been characterized, see~\cite{szendrei1986clones}. Indeed, such a characterization states that constraints can only be of three types. Let us first define the dual-discriminator operation before formulating the characterization. The dual-discriminator operation is defined as follows.
\begin{align*}
    \nabla(x,y,z)=
    \begin{cases} 
      y & \text{if } y=z,\\
      x & \text{otherwise.} 
   \end{cases}
\end{align*}

\begin{lemma}[\cite{CooperCJ94,szendrei1986clones}]
    Suppose $\nabla\in \Pol(\Gamma)$. Then each constraint $C_{ij}=\langle (x_i,x_j),R_{ij}\rangle$ is one of the following three types.
\begin{enumerate}
\item 
A complete constraint: $R_{ij}=D_i\times D_j$ for some $D_i,D_j\subseteq D$,
\item 
A permutation constraint: $R_{ij}=\{ (a,\pi(a))\mid a\in D_i\}$ for some $D_i\subseteq D$ and some bijection $\pi:D_i\to D_j$, where $D_j\subseteq D$, 
\item 
A two-fan constraint: $R_{ij}=\{ (\{a\}\times D_j) \cup (D_i\times\{b\})\}$ for some $D_i,D_j\subseteq D$ and $a\in D_i,b\in D_j$.
    \end{enumerate}
\end{lemma}

Here, we discuss a preprocessing step in order to handle permutation constraints and transform the instance into an instance without any permutation constraints. Such a transformation simplifies the problem and yields to an efficient computation of a \GB. Let $\cP=(X,D,C)$ be an instance of $\CSP(\Gamma)$ such that $\nabla\in \Pol(\Gamma)$. Suppose the set of constraints $C$ contains a permutation constraint $C_{ij}=\langle (x_i,x_j),R_{ij}\rangle$ with $R_{ij}=\{ (a,\pi(a))\mid a\in D_i,\pi(a)\in D_j\}$.  Define instance $\cP'=(X\setminus \{x_j\},D,C')$ as follows. 
\begin{enumerate}
\item 
$C'_{st} = C_{st}$ if $s\neq j$ and $t\neq j$,
\item 
replace each constraint $C_{sj}=\langle (x_s,x_j),R_{sj}\rangle$ by  $C'_{si}=\langle (x_s,x_i),R'_{si}\rangle$ where $R'_{si}= \{(a,\pi^{-1}(b)) \mid (a,b)\in R_{sj}\}$,
\item 
replace each constraint $C_{js}=\langle (x_j,x_s),R_{js}\rangle$ by  $C'_{is}=\langle (x_i,x_s),R'_{is}\rangle$ where $R'_{is}= \{(\pi^{-1}(a),b) \mid (a,b)\in R_{js}\}$.
\end{enumerate}
Note that instance $\cP$ has a solution if and only if instance $\cP'$ has a solution. The next lemma states the above preprocessing step does not change the complexity of the \IMP. In the next lemma we use a polynomial interpolation of permutation $\pi$. The permutation $\pi$ can be interpolated by a polynomial $p:\zR\to\zR$ such that $p(a)=\pi(a)$ for all $a\in D_i$.

\begin{lemma}
\label{lem:no-permutation}
    Let $\I(\cP)$ and $\I(\cP')$ be the corresponding ideals to instances $\cP$ and $\cP'$, respectively. Given a polynomial $f_0$, define polynomial $f'_0$ to be the polynomial obtained from $f_0$ by replacing every occurrence of $x_j$ with $p(x_i)$. Then $f_0\in \I(\cP)$ if and only if $f'_0\in \I(\cP')$.
\end{lemma}

\begin{proof}
    Note that, by our construction, there is a one-to-one correspondence between the points in $ \mb{V}(\I(\cP))$ and the points in $ \mb{V}(\I(\cP'))$. That is, each $\mb{a}=(a_1,\dots,a_i,\dots, a_j=p(a_i),\dots, a_n)\in \mb{V}(\I(\cP))$ corresponds to $\mb{a}'=(a_1,\dots,a_{j-1},a_{j+1},\dots,a_n)\in \mb{V}(\I(\cP'))$. Furthermore, for all
    $\mb{a}\in \mb{V}(\I(\cP))$ and the corresponding $\mb{a}'\in \mb{V}(\I(\cP'))$ we have $f_0(\mb{a})=f'_0(\mb{a}')$. This yields that 
    \[
        (\exists \mb{a}\in \mb{V}(\I(\cP)) \land f_0(\mb{a})\neq 0) \iff (\exists \mb{a}'\in \mb{V}(\I(\cP')) \land f'_0(\mb{a}')\neq 0)
    \]
    This finishes the proof of the lemma.
\end{proof}

The preprocessing step and \Cref{lem:no-permutation} suggest that we can massage any instance of $\IMP(\Gm)$ and obtain an instance $\IMP(\Gm)$ without permutation constraints. This can be carried out in polynomial time by processing permutation constraints one by one in turn, provided the original polynomial has bounded degree, as otherwise the number of monomials in the resulting polynomial may be exponentially greater that that of the original one.

\begin{lemma}
     Let $\cP=(X,D,C)$ be an instance of $\CSP(\Gamma)$ such that $\nabla\in \Pol(\Gamma)$ and $C$ contains no permutation constraint. Then a \GB\ of the corresponding ideal $\I(\cP)$ can be computed in polynomial time.
\end{lemma}

\begin{proof}
    The satisfiability of the instance $\cP$ can be decided in polynomial time~\cite{CooperCJ94} so we may assume that $1\not \in \I(\cP)$, else $G=\{1\}$ is a \GB. Moreover, we assume for every $x_i,x_j,x_k\in X$ if $(a,b)\in R_{ij}$ then there exists $c$ such that $(a,c)\in R_{ik}$ and $(c,b)\in R_{kj}$. This corresponds to the so-called \emph{arc consistency} notion which is used in solving \CSP s extensively. Note that such an assumption does not change $\Sol(\cP)$. Equivalently, it does not change $\Variety{\I(\cP)}$ which, by \Cref{th:nullstz}, means the ideal $\I(\cP)$ does not change. This arc consistency assumption has a consequence in terms of polynomials which enables us to prove our set of polynomials is indeed a \GB.
    
    First, let us give a set $G$ of polynomials that represent binary constraints. Initially, $G$ contains all the domain polynomials i.e. $G=\{\prod\limits_{a\in D}(x_i-a)\mid x_i\in X\}$. We proceed as follows.
\begin{enumerate}
\item[$i)$] 
To each complete constraint $C_{ij}=\langle (x_i,x_j),R_{ij}=D_i\times D_j\rangle$ we associate two polynomials $g_{i}=\prod\limits_{a\in D_i}(x_i-a)$ and $g_{j}=\prod\limits_{b\in D_j}(x_j-b)$ and replace $\prod\limits_{a\in D}(x_i-a)$ by $g_i$, and replace $\prod\limits_{a\in D}(x_j-a)$ by $g_j$.
\item[$ii)$] 
To each two-fan constraint $C_{ij}=\langle (x_i,x_j),R_{i,j}=\{ (\{a\}\times D_j) \cup (D_i\times\{b\})\}\rangle$ we associate polynomial $g_{ij}=(x_i-a)(x_j-b)$ and set $G=G\cup\{g_{ij}\}$. Furthermore, we replace domain polynomial $\prod\limits_{a\in D}(x_i-a)$ by $\prod\limits_{a\in D_i}(x_i-a)$, and replace domain polynomial $\prod\limits_{a\in D}(x_j-a)$ by $\prod\limits_{a\in D_j}(x_j-a)$.
\end{enumerate}
     Observe that, as a consequence of arc consistency, for every two polynomials $f=(x_i-a)(x_j-b)$ and $g=(x_i-c)(x_k-d)$ in $G$ with $a\neq c$ polynomial $h=(x_j-b)(x_k-d)$ is also in $G$.
    
     Consider \grlex order with $x_1\succ_\lex\dots \succ_\lex x_n$. Now, we prove that for any two polynomials $f,g\in G$ we have $S(f,g)\to_{G} 0$. We dismiss the cases where $\LM(f)$ and $\LM(g)$ are relatively prime. In these cases, by \Cref{prop:prime-LM}, we have $S(f,g)\to_{G} 0$. Hence, we focus on the cases where  $\LM(f)$ and $\LM(g)$ are not relatively prime.
\begin{enumerate}
\item 
Suppose $f=(x_i-a)(x_j-b)$ and $g=(x_i-a)(x_k-d)$. Then 
\[
    S(f,g)=x_k\cdot f-x_j\cdot g=d\cdot x_i\cdot x_j-b\cdot x_i\cdot x_k+a\cdot b\cdot x_k-a\cdot d\cdot x_j=d\cdot f-b\cdot g.
\]
Observe that $\multideg(S(f, g)) \succeq \multideg(d\cdot f)$ and $\multideg(S(f, g)) \succeq \multideg( b\cdot g)$. Hence, by \Cref{def:division-alg}, we have $S(f,g)\to_{\{f,g\}} 0$.
\item 
Suppose $f=(x_i-a)(x_j-b)$ and $g=(x_i-c)(x_k-d)$ where $a\neq c$. Then $S(f,g)=x_k\cdot f- x_j\cdot g$ and $S(f,g)\to_{\{f,g\}}(c-a)(x_j-b)(x_k-d)$. However, $h=(x_j-b)(x_k-d)$ is in $G$ and hence $S(f,g)\to_{G}0$.
\item 
Suppose $f=\prod\limits_{a\in D_i}(x_i-a)$ and $g=(x_i-c)(x_j-b)$ where $c\in D_i\subseteq D$. Define $f_1=\prod\limits_{a\in D_i\setminus\{c\}}(x_i-a)$ and $g_1=(x_j-b)$. Hence, $f=(x_i-c)\cdot f_1$ and $g=(x_i-c)\cdot g_1$.
        \begin{align*}
            S(f,g)
            &=x_j \cdot f - (x_i^{|D_i|-1})\cdot g\\
            &= \left[(x_j-b)+b\right]\cdot f - \left[\prod\limits_{a\in D_i\setminus \{c\}}(x_i-a)-\left(\prod\limits_{a\in D_i\setminus\{c\}}(x_i-a)-x_i^{|A_i|-1}\right)\right]\cdot g\\
            &= (x_j-b)\cdot f-(b)\cdot f - \prod\limits_{a\in D_i\setminus\{c\}}(x_i-a)\cdot g + \left(\prod\limits_{a\in D_i\setminus\{c\}}(x_i-a)-x_i^{|D_i|-1}\right)\cdot g\\
            &= (-b)\cdot f + \left(\prod\limits_{a\in D_i\setminus\{c\}}(x_i-a)-x_i^{|D_i|-1}\right)\cdot g=(g_1-\LT(g_1))\cdot f + (f_1-\LT(f_1))\cdot g 
        \end{align*}
\end{enumerate}

     We show that $\multideg(S(f, g)) \succeq \multideg((\LT(g_1)-g_1)\cdot f)$ and $\multideg(S(f, g)) \succeq \multideg( (f_1-\LT(f_1))\cdot g)$. This follows by showing $\LM((\LT(g_1)-g_1)\cdot f)\neq \LM((f_1-\LT(f_1))\cdot g)$. By contradiction, if $\LM((\LT(g_1)-g_1)\cdot f)= \LM((f_1-\LT(f_1))\cdot g)$ then
     \begin{align*}
         \LM(\LT(g_1)-g_1)\cdot\LM (f)= \LM(f_1-\LT(f_1))\cdot \LM(g)\implies x_i^{|A_i|}= x_i^{|A_i|-2}\cdot x_ix_j
     \end{align*}
     The latter is impossible, hence $\multideg(S(f, g)) \succeq \multideg((\LT(g_1)-g_1)\cdot f)$ and $\multideg(S(f, g)) \succeq \multideg( (f_1-\LT(f_1))\cdot g)$. Therefore, we have $S(f,g)\to_{\{f,g\}} 0$ (recall \Cref{def:division-alg}).
     
     We have shown for any two polynomials $f,g\in G$ we have $S(f,g)\to_{G} 0$. Hence, by Buchberger's criterion (\Cref{th:crit}), $G$ is a \GB\ for $\I(\cP)$.
\end{proof}

\begin{theorem}\label{the:dual-discriminator}
    Let $\Gamma$ be a constraint language. If $\Gamma$ has the dual-discriminator polymorphism, then $\IMP_d(\Gamma)$ can be solved in polynomial time.
\end{theorem}

\begin{remark}
Note that unlike in the results of \cite{Mastrolilli19,Bharathi-Dual-Disc}, due to the preprocessing step, Theorem~\ref{the:dual-discriminator} does not always allow one to find a proof that a polynomial belongs to the ideal, but this is resolved in Section~\ref{sec:ximp-truncated}.
\end{remark}

\subsection{Semilattice polymorphisms}
 In this section we study the \IMP($\Gamma$) for languages $\Gamma$ where \Pol($\Gamma$) contains a \emph{semilattice} operation. A binary operation $\psi(x, y)$ satisfying the following three conditions is said to be a semilattice operation:
 \begin{enumerate}
     \item  Associativity: $\psi(x,\psi(y,z))=\psi(\psi(x,y),z)$
     \item  Commutativity: $\psi(x,y)=\psi(y,x)$ 
     \item  Idempotency: $\psi(x,x)=x$ 
 \end{enumerate}
 
 Mastrolilli \cite{Mastrolilli19} considered this problem for languages over the Boolean domain i.e., $D=\{0,1\}$, and proved the following. We remark that a Boolean relation is closed under a semilattice operation if and only if it can be defined by a conjunction of \emph{dual-Horn} clauses or can be defined by a conjunction of \emph{Horn} clauses~\cite{JeavonsCG97}.
 
\begin{theorem}[\cite{Mastrolilli19}]\label{IMP-Bool}
    Let $\Gamma$ be a finite Boolean constraint language. If $\Gamma$ has a semilattice polymorphism, then $\IMP_d(\Gamma)$ can be solved in $n^{O(d)}$ time for $d \geq 1$.
\end{theorem}

We extend this tractability result to languages over any finite domain $D$. That is, we prove that $\IMP_d(\Gm)$ is polynomial time solvable when $\Gm$ is a language over $D$ and it has a semilattice polymorphism. To do so, we use a well-known result in semilattice theory. A \emph{semilattice} is an algebra $\cD=(D,\{\psi\})$, where $\psi$ is a semilattice operation. Informally speaking, every semilattice is a subalgebra of a direct power of a 2-element semilattice. 

\begin{theorem}[\cite{Papert64:congruence}]\label{embed-to-boolean}
    Let $\cD=(D,\psi)$ be a finite semilattice where $\psi$ is a semilattice operation. Then there is $k$ such that $\cD$ is a subalgebra of the direct power $\cB^k$ of $\cB=(\{0, 1\},\vf)$, where $\vf$ is a semilattice operation on $\{0,1\}$.
\end{theorem}

Armed with Theorem~\ref{embed-to-boolean} a proof of tractability of semilattice IMPs is straightforward.

\begin{theorem}\label{the:semilattice}
    Let $\Gamma$ be a finite constraint language over domain $D$. If $\Gamma$ has a semilattice polymorphism, then $\IMP_d(\Gamma)$ can be solved in polynomial time.
\end{theorem}

\begin{proof}
Let $\psi$ be a semilattice polymorphism of $\Gm$. Then $\cD=(D,\psi)$ is a semilattice, and therefore is a subalgebra of $\cB^k$, where $\cB=(\{0,1\},\vf)$ is a 2-element semilattice. By Lemma~\ref{lem:HSP-pp-definitions}, there is a finite constraint language $\Dl$ over $\{0,1\}$ such that $\vf$ is a polymorphism of $\Dl$. By Theorem~\ref{the:pp-interpretability}, $\IMP_d(\Gm)$ reduces to $\IMP_{ud}(\Dl)$ in polynomial time for a constant $u$. By Theorem~\ref{IMP-Bool}, we get the result.
\end{proof}

\begin{example}[Totally Ordered Domain]\label{exa:total-order}
    Let $D=\{1,\dots,t\}$ be a finite domain and $\Gamma$ be a language defined on $D$. A semilattice polymorphism $\psi$ is \emph{conservative} if $\psi(x,y)\in\{x,y\}$. Suppose $\Gamma$ has a conservative semilattice polymorphism $\psi:D^2\to D$. Note that $\psi$ defines a total ordering on $\{1,\dots,t\}$ so that $u\le v$ if and only if $\psi(u,v)=u$. Define $\pi:D\to \{0,1\}^t$ to be the following mapping 
\[
    \pi(i)=(0,\dots,0,\overbrace{1,\dots,1}^{i}).
\]
Let $\cP = (X, D, C)$ denote an instance of \CSP($\Gamma$) where $X =\{x_1,\dots , x_n\}$. Construct $\CSP$ instance $\cP' = (X', \{0,1\}, C')$ with $X' =\{x_{11},\dots,x_{1t},\dots , x_{n1},\dots,x_{nt}\}$ with the following set of constraints
\begin{itemize}
    \item [1.] 
    $x_{ij} \leq x_{ik}$ for all $1\leq i\leq n$ and $1\leq k\leq j\leq t$,
    \item [2.] 
    if $R(x_{i_1},\dots , x_{i_k})\in C$ then $\pi(R)(x_{i_11},\dots,x_{i_1t},\dots , x_{i_k1},\dots,x_{i_kt})\in C'$.
\end{itemize}
Observe that $\mb{a}\in \mb{V}(\I(\cP))$ if and only if $\pi(\mb{a})\in \mb{V}(\I(\cP'))$. 
Given $f_0\in\Field[x_1,\dots , x_n]$, define $f'_0\in \Field[x_{11},\dots,x_{1t},\dots , x_{n1},\dots,x_{nt}]$ to be the polynomial obtained from $f_0$ where we replace each indeterminate $x_i$ with $x_{i1}+\dots+x_{it}$. It is easy to check that, for $\mb{a}\in \mb{V}(\I(\cP))$, we have $f_0(\mb{a})=0$ if and only if $\pi(\mb{a})\in \mb{V}(\I(\cP'))$ and  $f'_0(\pi(\mb{a}))=0$. Therefore, deciding if $f_0\in\I(\cP)$ is equivalent to deciding if $f'_0\in\I(\cP')$ where the later one is polynomial time solvable by Theorem~\ref{IMP-Bool}. 
\end{example}


\subsection{Linear system in $\GF(p)$}\label{sec:lin-mod-p-ROOTsUNITY}

In this section we focus on constraint languages that are expressible as a system of linear equations modulo a prime number. Let $\Gm$ be a constraint language over a set $D$ with $|D|=p$, and $p$ a prime number. Suppose $\Gm$ has an \emph{affine} polymorphism modulo $p$ (i.e.\ a ternary operation $\psi(x,y,z)=x\ominus y\oplus z$, where $\oplus,\ominus$ are addition and subtraction modulo $p$, or, equivalently, of the field $\GF(p)$). In this case every \CSP\ can be represented as a system of linear equations over $\GF(p)$. Without loss of generality, we may assume that the system of linear equations at hand is already in the \emph{reduced row echelon} form. Transforming system of linear equations mod $p$ in its reduced row echelon form to a system of polynomials in $\zR[X]$ that are a \GB\ is not immediate and requires substantial work. This is the case even if we restrict ourselves to lexicographic order. 

A proof based on the \GBs\ conversion technique is given in Appendix~\ref{sec:lin-mod-p}. That proof is of an independent interest as it relies on the existence of the so-called \emph{independent p-expressions} and it is quite technical. Here, we present an alternative simple algorithm that checks the membership in polynomial time. In section \ref{sec:ximp-truncated}, we will see how this algorithm can be used to construct a $d$-truncated \GB. 

Let $\cP$ be an instance of $\CSP(\Gm)$ that is expressed as a system of linear equations $\mathscr{S}$ over $\zZ_{p}$ with variables $x_1,\dots,x_n$. A system of linear equations over $\zZ_{p}$ can be solved by Gaussian elimination (this immediately tells us if $1 \in \I(\cP)$ or not, and we proceed only if $1 \not\in \I(\cP)$). We assume a lexicographic order $\succ_\lex$ with $x_1 \succ_\lex \dots \succ_\lex x_n$. We also assume that the linear system has $r \le n$ equations and it is already in its reduced row echelon form with $x_i$ as the leading monomial of the $i$-th equation. Let $Supp_i \subset [n]$ such that $\{x_j : j \in Supp_i\}$ be the set of variables appearing in the $i$-th equation of the linear system except for $x_i$.
 
 Fix a prime $p$ and let $\oplus$, $\ominus$, $\odot$ denote addition, subtraction, and multiplication modulo $p$, respectively. We will call a linear polynomial over $\zZ_p$ a \emph{$p$-expression}. Let the $i$-th equation be $g_i = 0\pmod p$ where $g_i := x_i\ominus f_i$, with $i \in [r]$ and $f_i$ is the $p$-expression $(\bigoplus_{j\in Supp_i} \alpha_jx_j)\oplus \alpha_i$ and $\al_j,\alpha_i \in \zZ_p$. We will assume that each variable $x_i$
is associated with its $p$-expression $f_i$ which comes from the mod $p$ equations. This is clear for $i \le r$; for $i > r$ the $p$-expression $f_i = x_i$ itself. Hence, we can write down the reduced \GB\ in the \lex order in an implicit form as follows.
\begin{align}
\label{eq:G1}
    G_1=\{ x_1 \ominus f_1,\dots, x_r\ominus f_r ,\prod_{i\in\zZ_p}(x_{r+1}-i),\dots,\prod_{i\in\zZ_p}(x_{n}-i)\}
\end{align} 
    Let $U_{p}=\{\om,\om^2,\dots,\om^{p}=\om^0=1\}$ be the set of $p$-th roots of unity where $\om$ is a primitive $p$-th root of unity. For a primitive $p$-th root of unity $\om$ we have $\om^a=\om^b$ if and only if $a \equiv b ~(\mathrm{mod}~p)$. From $\mc{P}$ we construct a new CSP instance $\mc{P}'= (V,U_{p},\widetilde{C})$ where for each equation $x_i\ominus f_i=0$ with $f_i=(\bigoplus_{j\in Supp_i} \alpha_jx_j)\oplus \alpha_i$ we add the constraint $x_i-f_i'=0$ with $$f_i'=\om^{\alpha_i}(\prod\limits_{j\in Supp_i}x_j^{\alpha_j}).$$ 
    Moreover, the domain constraints are different. For each variable $x_j$, $r+1\le j\leq n$, the domain polynomial is $(x_j)^{p}-1$. Therefore, we write $G$ over complex number domain as follows.

    \begin{align}
        \label{eq:G1-complex}
            G'=\{ x_1 - f_1',\dots, x_n-f_n' ,(x_{r+1})^{p}-1,\dots,(x_n)^{p}-1\}
    \end{align} 

    Define univariate polynomial $\phi\in \zC[x]$ so that it interpolates points $(0,\om^0),(1,\om),\dots,(p-1,\om^{p-1})$. This polynomial provides a one-to-one mapping between solutions of instance $\mc{P}$ and instance $\mc{P}'$. That is, $(a_1,\dots,a_n)$ is a solution of $\mc{P}$ if and only if $(\phi(a_1),\dots,\phi(a_n))$ is a solution of $\mc{P}'$.

\begin{lemma}
    \label{lem:equivalent}
    For a polynomial $f\in\Real[x_1,\dots,x_n]$ define polynomial $f'\in\Complex[{x_1},\dots,{x_n}]$  to be
    \[
        f'({x_1},\dots,{x_n})=f(\phi^{-1}(x_1),\dots,\phi^{-1}(x_n)).
    \]
    Then $f\in\I(\mc{P})$ if and only if $f'\in \I(\mc{P}')$.
\end{lemma}
\begin{proof}
    As $\I(\mc{P})$ is radical, by the Strong Nullstellensatz we have $f\not\in \I(\mc{P})$ if and only if there exists a point $\mb{a}\in \zZ_{p}^n$ such that $\mb{a}$ is in $Sol(\mc{P})$ and $f(\mb{a})\neq 0$. Similarly, for $\I(\mc{P}')$ we have $f'\not\in \I(\mc{P}')$ if and only if there exists a point $\mb{a}'\in U_{p}^n$ such that $\mb{a}'\in Sol(\mc{P}')$ and  $f'(\mb{a}')\neq 0$.
    
    Suppose $f\not\in \I(\mc{P})$ and consider a point $\mb{a}=({a_1},\dots,{a_n})\in \zZ_{p}^n$ so that $\mb{a}\in Sol(\mc{P})$ and  $f(\mb{a})\neq 0$. 
     Recall that $\mb{a}\in Sol(\mc{P})$ if and only if $\mb{a}'=(\om^{a_1},\dots,\om^{a_n})\in Sol(\mc{P}')$. Furthermore,
    \[
        f'(\om^{a_1},\dots,\om^{a_n})  = f(a_1,\dots,a_n) \neq 0
    \]
    Therefore, for $\mb{a}'\in Sol(\mc{P}')$ we have $f'(\mb{a}')\neq 0$ which implies $f'\not\in \I(\mc{P}')$. This finishes the proof.
\end{proof}

The next lemma states that $\IMP(\I(\mc{P}'))$ is polynomial time solvable by showing that the set of polynomials $G'$ is in fact a \GB\ for $\I(\mc{P}')$.
\begin{lemma}
    \label{lem:GB-P'}
    $G'$ is a \GB\ for $\I(\mc{P}')$ with respect to \lex order $x_1 \succ_\lex \dots \succ_\lex x_n $.
\end{lemma}    
\begin{proof}
     First, we proceed to show $G'$ is a \GB\ for $\Ideal{G'}$. Note that for each $x_i - f_i'$, we have $\LM(x_i - f_i') = x_i$. Moreover, the leading monomial of $(x_j)^{p}-1$, $r+1\leq j\leq n$, is $(x_j)^{p}$. Hence, for every pair of polynomials in $G'$ the reduced S-polynomial is zero as the leading monomials of any two polynomials in $G'$ are relatively prime. By Buchberger's Criterion it follows that $G'$ is a \GB\ for $\langle G'\rangle$ over $\Complex[x_1,\dots,x_n]$ (according to the \lex order).
    
    It remains to show $\Ideal{G'}= \mb{I}(Sol(\mc{P}'))$. According to our construction we have $\Variety{\langle G'\rangle}= Sol(\mc{P}')$ which implies $\Ideal{G'}\subseteq \mb{I}(Sol(\mc{P}'))$. In what follows, we prove $\mb{I}(Sol(\mc{P}'))\subseteq \Ideal{G'}$. Consider a polynomial $f\in \mb{I}(Sol(\mc{P}'))$. Note that $f(\mb{a})=0$ for all $\mb{a}\in Sol(\mc{P}')$. Let $r = f|_{G'}$. $r$ does not contain variables $x_1,\dots,x_r$, and hence it is a polynomial in $x_{r+1},\dots,x_n$. Now note that any $\mb{b}=(b_{1},\dots,b_{n-r})\in U_p^{n-r}$ extends to a unique point in $Sol(\mc{P}')$. Therefore, all the points in $U_p^{n-r}$ are zeros of $r$, hence $r=f|_{G'}$ is the zero polynomial. Since $f\in \mb{I}(Sol(\mc{P}'))$ was arbitrary chosen, it follows that for every $f\in \mb{I}(Sol(\mc{P}'))$ we have $f|_{G'}=0$. Hence $G'$ is a \GB\ of $\mb{I}(Sol(\mc{P}'))$.
\end{proof}

We remark that the substitution technique used in Lemma \ref{lem:equivalent} may result in a polynomial with exponentially many monomials. However, for polynomials of bounded degree $d$ it yields a polynomial of degree at most $p d$ and polynomially many monomials. Thus, we obtain the following

\begin{theorem}\label{the:linear}
Let $\Gm$ be a constraint language over domain $D$. If $\Gm$ has an affine polymorphism, then $\IMP_d(\Gm)$ can be solved in polynomial time.
\end{theorem}
\section{A variation of IMP and its applications}
In this section, we study a variation of the $\IMP$ and show similar reductions for pp-definability notion as well as  pp-interpretability notion. These reductions provide a very strong tool for solving the search version of the $\IMP$. More generally, using the algebraic approach we present a general framework for computing $d$-truncated \GBs~for combinatorial ideals. Moreover, in the later sections we discuss how this variation of the $\IMP$ leads to $\Sos$ proofs with a milder restriction on the degree of proofs.

\subsection{The IMP with indeterminate coefficients}

    In the new variant of the $\IMP$ we are given a polynomial with unknown coefficients and the goal is to decide if there is an assignment for coefficients such that the resulting polynomial belongs to an ideal. Formally speaking, 
    
\begin{definition}[\xIMP]
    Given an ideal $\I\sse\Field[\vc x n]$ and a vector of $\ell$ polynomials $M=(g_1,\dots,g_\ell)$, the \xIMP~asks if there exist coefficients $\mathbf{c}=(\vc c \ell)\in \Field^\ell$ such that $\mathbf{c}M=\sum_{i=1}^\ell c_i g_i$ belongs to the ideal $\I$.
\end{definition}

    In a similar fashion, \xIMP~associated with a constraint language $\Gamma$  over a set $D$ is the problem \xIMP$(\Gamma)$ in which the input is a pair $(M,\cP)$ where $\cP = (X, D, C)$ is a
    $\CSP(\Gm)$ instance and $M$ is a vector of $\ell$ polynomials. The goal is to decide
    whether there are coefficients $\mathbf{c}=(\vc c \ell)\in \Field^\ell$ such that $\mathbf{c}M$ lies in the combinatorial ideal $\I(\cP)$. We use $\chi\IMP_d(\Gamma)$ to denote $\chi\IMP(\Gamma)$ when the vector $M$ contains polynomials of total degree at most $d$. 
    
    Next, we show that the main reductions from Section~\ref{sec:IMPs} work for the \xIMP\ as well.

\begin{theorem}\label{thm:X-pp-define-reduction}
    If $\Gm$ pp-defines $\Dl$, then \xIMP$(\Dl)$ is polynomial time reducible to \xIMP$(\Gm)$.
\end{theorem}

\begin{proof}
The proof of Theorem~\ref{thm:X-pp-define-reduction} closely follows that of Theorem~\ref{pp-define-reduction}.
    Let $(M,\cP_\Dl)$, $\cP_\Dl=(X,D,C_\Dl)$, be an instance of \xIMP$(\Dl)$ where $X=\{x_{m+1},\dots, x_{m+k}\}$, $M$ is a vector of $\ell$ polynomials in $x_{m+1},\dots,x_{m+k}$, $k=|X|$, and $m$ will be defined later, and $\I(\cP_\Dl)\subseteq \Field[x_{m+1},\dots,x_{m+k}]$. From this we construct an instance $(M',\cP_\Gm)$ of \xIMP$(\Gm)$ where $M'$ is a vector of $\ell'$ polynomials in $x_{1},\dots,x_{m+k}$ and $\I(\cP_\Gm)\subseteq \Field[x_{1},\dots,x_{m+k}]$ such that 
    \[
        \exists \mathbf{c}\in \Field^\ell \text{~with~} \mathbf{c}M\in \I(\cP_\Dl) \iff \exists \mathbf{c}'\in \Field^{\ell'} \text{~with~} \mathbf{c}'M'\in \I(\cP_\Gm).
    \]
    Using pp-definitions of relations from $\Dl$ we convert the instance $\cP_\Dl$ into an instance $\cP_{\Gm}=(\{x_{1},\dots,x_{m+k}\},D,C_\Gm)$ of $\CSP(\Gm)$ such that every solution of $\cP_\Dl,\cP_\Gm$ satisfy the Extension Condition~\ref{extention-condition}. Such an instance $\cP_\Gm$ can be constructed in polynomial time as follows.
    
    By the assumption each $S\in\Dl$, say, $t_S$-ary, is pp-definable in $\Gm$ by a pp-formula involving relations from $\Gm$ and the equality relation, $=_D$. Thus,
    \[
        S(y_{q_S+1},\dots,y_{q_s+t_S})=\exists \vc y{q_S} (R_1(w^1_1,\dots, w^1_{l_1})\wedge\dots\wedge R_r(w^r_1,\dots, w^r_{l_r})),
    \]
    where $w^1_1,\dots, w^1_{l_1},\dots,w^k_1,\dots, w^k_{l_k}\in \{\vc y{m_S+t_S}\}$ and $\vc Rr\sse\Gm\cup\{=_D\}$. 

    Now, for every constraint $B=\langle \bs,S\rangle\in C_\Dl$, where $\bs=(x_{i_1},\dots,x_{i_t})$ create a fresh copy of $\{\vc y{q_S}\}$ denoted by $Y_B$, and add the following constraints to $C_\Gm$
    \[
         \langle (w^1_1,\dots, w^1_{l_1}),R_1\rangle,\dots, \langle (w^r_1,\dots, w^r_{l_r}),R_r\rangle.
    \]
    We then set $m=\sum_{B\in C}|Y_B|$ and assume that $\cup_{B\in C}Y_B=\{\vc xm\}$. Note that the problem instance obtained by this procedure belongs to $\CSP(\Gm\cup \{=_D\})$. All constraints of the form $\langle(x_i, x_j),=_D\rangle$  can be eliminated by replacing all occurrences of the variable $x_i$ with $x_j$.  Moreover, it can be checked (see also Theorem 2.16 in \cite{BulatovJK05}) that $\cP_\Dl,\cP_\Gm$ satisfy the Extension Condition~\ref{extention-condition}. 
    
    Let $\I(\cP_\Gm)\subseteq\Field[x_{1},\dots,x_{m+k}]$ be the ideal corresponding to $\cP_{\Gm}$ and set $M'=M$. Now, $(M',\cP_\Gm)$ is an instance of \xIMP$(\Gm)$. We prove that, for every $\mathbf{c}\in \Field^{\ell}$, $\mathbf{c}M\in \I(\cP_\Dl)$ if and only if $\mathbf{c}M\in \I(\cP_\Gm)$.
    
    Consider an arbitrary $\mathbf{c}\in \Field^{\ell}$ and set $f_0=\mathbf{c}M$. Suppose $f_0\not\in \I(\cP_\Dl)$, this means there exists $\vf\in \mb{V}(\I(\cP_\Dl))$ such that $f_0(\vf)\neq 0$. By Theorem~\ref{extension-theorem}, $\vf$ can be extended to a point $\vf'\in\mb{V}(\I(\cP_\Gm))$. This in turn implies that $f_0\not\in \I(\cP_\Gm)$. Conversely, suppose $f_0\not\in \I(\cP_\Gm)$. Hence, there exists $\vf'\in\mb{V}(\I(\cP_\Gm))$ such that $f_0(\vf')\neq 0$. Projection of $\vf'$ to its last $k$ coordinates gives a point $\vf\in \mb{V}(\I_X)$. By Lemma~\ref{variety-m-elimination}, $\vf\in \mb{V}(\I(\cP_\Dl))$ which implies $f_0\not\in \I(\cP_\Dl)$.
\end{proof}

The reduction for pp-interpretable languages remain valid in the case of \xIMP\ as well.

\begin{theorem}\label{thm:X-pp-interpretability}
Let $\Gm,\Dl$ be constraint languages on sets $D,E$, respectively, and let $\Gm$ pp-interprets $\Dl$. Then $\chi\IMP_d(\Dl)$ is polynomial time reducible to $\chi\IMP_{d\ell|E|}(\Gm)$.
\end{theorem}

\begin{proof}
    Let $(M,\cP_\Dl)$ be an instance of \xIMP$(\Dl)$ where $M$ is a vector of $r$ polynomials in $x_{1},\dots,x_n$, $\cP_{\Dl}=(\{x_{1},\dots,x_n\},E,C_\Dl)$, an instance of $\CSP(\Dl)$, and $\I(\cP_\Dl)\subseteq \Field[x_{1},\dots,x_n]$. 

    The properties of the mapping $\pi$ from Definition~\ref{pp-interpret} allow us to rewrite an instance of $\CSP(\Dl)$ to an instance of $\CSP(\Gamma')$ over the constraint language $\Gm'$. Recall that, by \Cref{pp-interpret}, $\Gm'$ contains all the $\ell k$-ary relations $\pi^{-1}(S)$ on $D$ where $S\in\Dl$ is $k$-ary relation, as well as the $2\ell$-ary relation $\pi^{-1}(=_E)$.   
    
    Note that $\Gm'$ is pp-definable from $\Gm$. By Theorem~\ref{thm:X-pp-define-reduction}, \xIMP$(\Gamma')$ is reducible to \xIMP$(\Gm)$. It remains to show \xIMP$(\Dl)$ is reducible to \xIMP$(\Gamma')$. To do so, from instance $(M,\cP_\Dl)$ of \xIMP
    $(\Dl)$ we construct an instance $(M',\cP_{\Gm'})$ of \xIMP$(\Gamma')$ such that 
    \[
        \exists \mathbf{c}\in \Field^r \text{~with~} \mathbf{c}M\in \I(\cP_\Dl) \iff \exists \mathbf{c}'\in \Field^{r'} \text{~with~} \mathbf{c}'M'\in \I(\cP_\Gm').
    \]
    Let $p$ be a polynomial of total degree at most $\ell|E|$ that interpolates mapping $\pi$. For each monomial $\bx^{\alpha}=\prod x_i^{\alpha_i}$ in $M$ replace each indeterminate $x_i$ with $p(x_{1i},\dots,x_{\ell i})$. This yields the following polynomial
    \begin{align}
    \label{eq:sub-pp-inter}
        \prod\limits_{i=1}^n [p(x_{1i},\dots,x_{\ell i})]^{\alpha_i}
    \end{align}
    Note that for a monomial of total degree at most $d$, the maximal degree of monomials appearing in the polynomial \eqref{eq:sub-pp-inter} is at most $d\ell|E|$. Let $M'$ be the vector of monomials consisting monomials in \eqref{eq:sub-pp-inter} for all monomials in $M$. Observe that $M'$ contains at most $O(n^{d\ell|E|})$ monomials and each monomial in $M'$ consists of indeterminates $x_{11},\ldots , x_{\ell1},\ldots,x_{1n},\ldots,x_{\ell n}$. Now, $(M',\cP_{\Gm'})$ is an instance of \xIMP$(\Gamma')$.

    Consider an arbitrary $\mathbf{c}\in \Field^{r}$ and set $f_0=\mathbf{c}M\in \Field[x_1,\dots,x_n]$. Let $f'_0\in \Field [x_{11},\ldots , x_{\ell1},\ldots,x_{1n},\ldots,x_{\ell n}]$ be the polynomial that is obtained from $f_0$ by replacing each indeterminate $x_i$ with $p(x_{1i},\dots,x_{\ell i})$. Note that there exists $\mathbf{c}'$ such that $f'_0=\mathbf{c}'M'$.
    Clearly, for any assignment $\vf:\{\vc xn\}\to E$, $f_0(\vf)=0$ if and only if $f'_0(\psi)=0$ for every $\psi:\{x_{11}\zd x_{\ell n}\}\to D$
    such that 
    \[
    \vf(x_i)=\pi(\psi(x_{1i}),\dots,\psi(x_{\ell i}))
    \]
    for every $i\le n$. Moreover, 
    for any such $\vf,\psi$ it holds $\vf\in \mb{V}(\I(\cP_\Dl))$ if and only if  $\psi\in \mb{V}(\I(\cP_{\Gm'}))$. This yields that 
    \[
        (\exists \vf\in \mb{V}(\I(\cP_\Dl)) \land f_0(\vf)\neq 0) \iff (\exists \psi\in \mb{V}((\cP_{\Gm'})) \land f'_0(\psi)\neq 0)
    \]
    This completes the proof of the theorem.
\end{proof}


 Recall that one drawback of the reductions for the \IMP, Theorems~\ref{pp-define-reduction} and \ref{the:pp-interpretability}, is the issue of recovering a proof which is a subtle point in the search version of the $\IMP$. A nice property of the reductions in Theorems~\ref{thm:X-pp-define-reduction} and \ref{thm:X-pp-interpretability} is that they reductions for the search version of the \xIMP as well. To elaborate, consider the reduction for pp-interpretablility in the proof of Theorem~\ref{thm:X-pp-interpretability}. The entries of vector $\mathbf{c}'$ are linear combination of $c_1,\ldots,c_\ell$. Hence, if there exists a polynomial time algorithm that finds $\mathbf{c}'$ such that $\mathbf{c}'M'\in\I(\mc{P}_{\Gamma'})$ then a vector $\mathbf{c}$ with $\mathbf{c}M\in \I(\mc{P}_{\Delta})$ can be computed by simply solving a system of linear equations with $c_1,\ldots,c_\ell$ as unknowns. This is formalized as follows.
 \begin{theorem}
 \label{thm:search-xIMP-reduction}
    Let $\Gamma$ and $\Delta$ be constraint languages on (possibly similar) sets $D$, $E$, respectively. Suppose there exists a polynomial time algorithm that solves the search version of \xIMP$(\Gamma)$. Then, there exists a polynomial time algorithm that solves the search version of \xIMP$(\Delta)$ if
    \begin{itemize}
        \item [1.] $D=E$ and $\Gamma$ pp-defines $\Delta$, or
        \item [2.] $\Gamma$ pp-interprets $\Delta$.
    \end{itemize}
 \end{theorem}

 \begin{proof}
     It follows from a similar argument in the proofs of Theorems~\ref{thm:X-pp-define-reduction}, \ref{thm:X-pp-interpretability}, and noting that $\mathbf{c}'$ is a linear combination of $c_1,\ldots,c_\ell$.
 \end{proof}

\subsection{Sufficient conditions for tractability of \xIMP}

We first show that having a \GB\ yields a polynomial time algorithm for solving the search version of \xIMP. Next, we use the reductions from Theorem~\ref{thm:search-xIMP-reduction} to establish the tractability of $\chi\IMP_d(\Gamma)$ for languages closed under a various polymorphisms.
\begin{theorem}
    Let $\I$ be an ideal, and let $\{\vc g s\}$ be a \GB~for $\I$ with respect to some monomial ordering. Then the (search version of) \xIMP~ is polynomial time solvable.
 \end{theorem}
\begin{proof}
    Recall that a polynomial $p$ belongs to $\I$ if and only if the remainder on division of $p$ by $\vc g s$ is zero. Let $M=(\vc m \ell)$ be a vector of $\ell$ polynomials and $\mathbf{c}=(\vc c \ell)\in\Field^{\ell}$ be a vector of unknown coefficients. Set $f=\mathbf{c}M=\sum c_im_i$. We do the division algorithm to obtain the reminder of dividing $f$ by $\vc g s$. Repeatedly, choose a $g_i \in \{\vc g s\}$ such that $\LT(g_i)$ divides some term $t$ of $f$ and replace $f$ with $f - \frac{t}{\LT(g_i)} g_i$, until it cannot be further applied. Hence,
    \[
        f= q_1g_1+\dots+q_rg_s + r
    \]
    where $r$  is a linear combination, with unknown coefficients in $\Field$, of monomials, none of which is divisible by any of $\LT(g_1),\dots,\LT(g_s)$. The key observation is that the coefficients of monomials in $r$ are linear combination of $\vc c \ell$. Now, we want $r$ to be the zero polynomial. Hence, we set every unknown coefficient of monomials in $r$ to be zero. This in turn yields a system of linear equations in $\vc c \ell$. Such a system of linear equations has a solution if and only if there exists $\mathbf{c}=(\vc c \ell)$ such that $f=\mathbf{c}M \in \I$.
\end{proof}
The above theorem and the results by Mastrolilli~\cite{Mastrolilli19,Bharathi-Minority} give the following corollary.
\begin{corollary}
\label{cor:xIMP-bool}
    Let $\Gm$ be a finite constraint language over domain $\{0,1\}$. Then the (search version of) $\chi\IMP_d(\Gamma)$ can be solved in polynomial time if
    \begin{itemize}
        \item[1.] $\Gamma$ has a semilattice polymorphism, or
        \item[2.] $\Gamma$ has a majority polymorphism, or
        \item[3.] $\Gamma$ has a minority polymorphism.
    \end{itemize}
\end{corollary}
Now we use our reductions to prove the same tractability results for languages over arbitrary finite domain. Note that the only majority polymorphism over $\{0,1\}$ is the dual-discriminator.

\begin{theorem}
\label{thm:xIMP-finite-domain}
    Let $\Gm$ be a finite constraint language over domain $D$. Then the (search version of) $\chi\IMP_d(\Gamma)$ can be solved in polynomial time if
    \begin{itemize}
        \item[1.] $\Gamma$ has a semilattice polymorphism, or
        \item[2.] $\Gamma$ has the dual-discriminator polymorphism, or
        \item[3.] $\Gamma$ is expressed as a system of linear equations over $\GF(p)$, $p$ prime.
    \end{itemize}
\end{theorem}

\begin{proof}
    In the first case where $\Gamma$ has a semilattice polymorphism we reduce the problem to the Boolean case similar to Theorem~\ref{the:semilattice}. Then using Corollary~\ref{cor:xIMP-bool} we can solve the corresponding \xIMP\ problem. For the other two cases we use reductions similar to Theorems~\ref{the:linear}, \ref{the:dual-discriminator} to transform the instances into instances where we can compute truncated \GBs. The key observation here is that under these reductions, similar to the reduction for pp-interpretability in Theorem~\ref{thm:search-xIMP-reduction}, we end up with system of linear equations over unknown variables $c_1,\dots,c_{\ell}$.
\end{proof}

\subsection{A framework for constructing d truncated \GBs}\label{sec:ximp-truncated}

We observe that constructing a $d$-truncated \GB~for an ideal $\I$ is reducible to solving $\xIMP_d$ for the ideal $\I$.  With this reduction at hand, we design algorithms to construct $d$-truncated \GB~for many combinatorial ideals, namely, combinatorial ideals arising from languages invariant under a semilattice, or the dual-discriminator, or languages expressible as linear equations over $\GF(p)$. Some basic notation are in order.

Let $\I\in \Field[X]$ be an ideal. We say two polynomials $f,g$ are congruent modulo $\I$ and write $f \equiv g ~mod~\I$ if $f-g\in \I$. It is easy to see that congruence modulo $\I$ is an equivalence relation on $\Field[X]$. The quotient of $\Field[X]$ modulo $\I$, written $\Field[X]/\I$ is a ring with the base set consisting of the cosets $[f]=f+\I=\{f+q \mid q\in \I\}$. $\Field[X]/\I$ is a commutative ring under addition $[f]+[g] = [f+g]$ and multiplication $[f]\cdot[g] = [fg]$ (product in $\Field[X]$). We consider $\Field[X]/\I$ as a $\Field$-vector space with addition defined as above and scalar multiplication given by $c\cdot[f]=[c\cdot f]$, $c\in\Field$. We also consider the subset $\Field[X]_d/\I$ of all polynomials of total degree at most $d$. As is easily seen it is also an $\Field$-vector space. Note that $\Field[X]/\I$ is infinitely dimensional in general. However, if $\I$ is zero-dimensional and radical then the quotient ring $\Field[X]/\I$ is a finite dimensional vector space. Moreover, for any bound $d$ on the total degree of polynomials $\Field[X]_d/\I$ is finitely dimensional. We also have a natural basis for those spaces. 

\begin{proposition}[Proposition 1 on page 248 of~\cite{Cox}]
\label{prop:quotient-basis}
    Fix a monomial ordering on $\Field[X]$ and let $\I\subseteq\Field[X]$ be an ideal. Let $\Ideal{\LT(I)}$ denote the ideal generated by the leading terms of elements of $\I$. 
     \begin{itemize}
         \item [1.] Every $f\in \Field[X]$ is congruent modulo $\I$ to a \emph{unique} polynomial $r$ which is a $\Field$-linear combination of the monomials in the complement of $\Ideal{\LT(\I)}$,
         \item[2.] The elements of $\{\bx^\alpha\mid \bx^\alpha\not\in \Ideal{\LT(\I)}\}$ are linearly independent modulo $\I$, i.e., if we have 
         \[\sum_{\alpha}c_\alpha \bx^\alpha\equiv 0 ~mod~\I,\]
        where the $\bx^\alpha$ are all in the complement of $\Ideal{\LT(\I)}$, then $c_\alpha = 0$ for all $\alpha$.
     \end{itemize}
\end{proposition}

Proposition~\ref{prop:quotient-basis} suggests a simple algorithm to construct a $d$-truncated \GB. Let $\I\sse\Field[X]$ be an ideal. At the beginning of the algorithm, there will be two sets: $G$, which is initially empty but will become the $d$-truncated \GB\ with respect to the \grlex order, and $B(G)$, which initially contains $1$ and will grow to be the \grlex monomial basis of the quotient ring $\Field[\vc x n]/\I$ as a $\Field$-vector space i.e., $B(G)=\{\bx^\alpha\mid |\alpha|\leq d, \bx^\alpha\not\in \Ideal{\LT(\I)}\}$. In fact, $B(G)$ contains the reduced monomials (of degree at most $d$) with respect to $G$. Every $f \in \Field[\vc x n]$ is congruent modulo $\I$ to a unique polynomial $r$ which is a $\Field$-linear combination of the monomials in the \emph{complement} of $\Ideal{\LT(\I)}$. Furthermore, $\Field[X]_d/\I$ is isomorphic as a $\Field$-vector space to $\texttt{Span}(\bx^\alpha\mid \bx^\alpha\not\in \Ideal{\LT(\I)})$ via mapping $\Phi([f])=\reduce f G$. Here, $\texttt{Span}(\bx^\alpha\mid \bx^\alpha\not\in \Ideal{\LT(\I)})$ means the set of all $\Field$-linear combinations of $\{\bx^\alpha\mid \bx^\alpha\not\in \Ideal{\LT(\I)})\}$. Hence, for every $f\in \Field[\vc x n]$, we have
\[
    \reduce f G \in \texttt{Span}(\bx^\alpha\mid \bx^\alpha\not\in \Ideal{\LT(\I)}).
\]
 This suggests the following. In Algorithm~\ref{alg:GB+xIMP}, $Q$ is the list of all monomials of degree at most $d$ arranged in increasing order with respect to \grlex ordering. The algorithm iterates over monomials in $Q$ in increasing \grlex order and at each iteration decides exactly one of the following actions given the current sets $G$ and $B(G)$. 
\begin{enumerate}
\item 
$q$ should be discarded (if $q$ is divisible by some \LM~in $G$), or
\item 
a polynomial with $q$ as its leading monomial should be added to $G$, or
\item 
$q$ should be added to $B(G)$.
\end{enumerate} 

\begin{algorithm}[t]
  \caption{$d$-Truncated \GBs}\label{alg:GB+xIMP}
  \begin{algorithmic}[1]
    \Require{$\I$, degree $d$.}
    \State Let $Q$ be the list of all monomials of degree at most $d$ arranged in increasing \grlex order.
    \State $G=\emptyset, B(G)=\{1\}$ (we assume $1\not\in\I$). 
    \State Let $b_i$ (arranged in increasing \grlex order) be the elements of $B(G)$.
    \For{$q \in Q$}
    \If{$q$ is divisible by some \LM~in $G$}
        \State Discard it and go to Step 4,
    \EndIf
    \State Let $M$ be the vector of length $\ell$ whose entries are monomials in $B(G)$,
    \If{there exists $\mathbf{c}\in \Field^\ell$ such that $q-\mathbf{c}M\in \I$}
        \State $G = G \cup \{q-\mathbf{c}M\}$
    \Else 
        \State $B(G) = B(G) \cup \{q\}$ 
    \EndIf
    \EndFor
    \State \textbf{return} $G$
\end{algorithmic}
\end{algorithm}

\begin{theorem}
\label{thm:GB+xIMP}
    Let $\mc{H}$ be a class of ideals for which the search version of $\chi\IMP_d$ is polynomial time solvable. Then there exists a polynomial time algorithm (see \Cref{alg:GB+xIMP}) that constructs a degree $d$ \GB~of an ideal $\I\in \mc{H}$, $\I\subseteq\Field[\vc xn]$, in time $O(n^d)$.  
\end{theorem}
\begin{proof}
    \Cref{alg:GB+xIMP} is clearly a polynomial time algorithm assuming $\chi\IMP_d$ for ideal $\I$ is polynomial time solvable. We prove $G$ returned by the algorithm is $d$-truncated \GB, and set of monomials $B(G)$ is so that 
    \[
        B(G)=\{\bx^\alpha\mid |\alpha|\leq d, \bx^\alpha\not\in \Ideal{\LT(\I)}\}.
    \]
    We prove this by induction. The induction base is correct as $B(G)=\{1\}$ and $G=\emptyset$. Suppose sets $G$ and $B(G)$ are computed correctly up to the $i$-th iteration and let $q$ be the current monomial. 
    
    First, if $q$ is a multiple of some $\LM$ in $G$ then $q\in \Ideal{\LT(G)}$. Furthemore, no polynomial with $q$ as its leading monomial is in a reduced \GB~of $\I$ (recall the definition of a reduced \GB). Therefore, in this case, \Cref{alg:GB+xIMP} correctly discards monomial $q$.
    
    Second, suppose $q$ is not divisible by any $\LM$ in $G$ then by the division algorithm the normal form of $q$ by $G$, $\reduce q G$, is $q$ itself. Now the algorithm decides if a polynomial with $q$ as its leading monomial can be in $G$. Let $g=q+f$ be a polynomial such that $\LM(g) = q$. Therefore, by \Cref{prop:quotient-basis} and the inductive hypothesis, if $g\in \I$ then with the current $G$ and $B(G)$ we must have 
    \begin{align*}
        0 = \reduce g G &=\reduce q G + \reduce f G = q + \sum k_i b_i
    \end{align*}
    where all 
    \[
        b_i\in M =\{\bx^\alpha\mid deg(\bx^\alpha)< deg(q), \bx^\alpha\not\in \Ideal{\LT(\I)}\}
    \] 
    and $k_i\in \zR$. This yields $g\in\I$ if there exists $\mathbf{c}\in \zR^\ell$ such that $q-\mathbf{c}M\in \I(\cP)$ then $\{q-\mathbf{c}M\}\in \I$. Furthermore, if such $\mathbf{c}\in \zR^\ell$ does not exists then it implies there is no polynomial in $\I$ with $q$ as its leading monomial. Hence, $q$ must be added to $B(G)$.
    \end{proof}
    
    We point out that in Theorem~\ref{thm:GB+xIMP}, if only the decision version of \xIMP\ is polynomial time solvable then a slight modification of Algorithm~\ref{alg:GB+xIMP} returns basis monomials $\{\bx^\alpha\mid |\alpha|\leq d, \bx^\alpha\not\in \Ideal{\LT(\I)}\}$. 
\begin{theorem}
    \label{thm:GB+xIMP+languages}
    Let $\Gamma$ be a finite constraint language over domain $D$. For an instance $\mc{P}$ of $\CSP(\Gm)$ a $d$-truncated $\GB$ of $\I(\mc{P})$ can be computed in time $O(n^d)$ if
    \begin{itemize}
        \item[1.] $\Gamma$ has a semilattice polymorphism, or
        \item[2.] $\Gamma$ has the dual-discriminator polymorphism, or
        \item[3.] $\Gamma$ is expressed as a system of linear equations over $\GF(p)$, $p$ prime.
    \end{itemize}
\end{theorem}
\begin{proof}
    Follows from Theorems \ref{thm:xIMP-finite-domain} and \ref{thm:GB+xIMP}.
\end{proof}
\section{\Sos~proofs: bit complexity and automatizability}\label{sec:SOS}

The focus of this section is on designing efficient algorithms to find proofs of nonnegativity of polynomials over (semi)algebraic sets. Sum-of-squares certificates of nonnegativity is a popular and powerful framework to provide a proof that a polynomial is nonnegative. In this section we present a very light introduction to \Sos\ proofs of nonnegativity and lay down some notation and background. The main appeal of this proof system is that it can be transformed into an SDP feasibility problem and hopefully be solved efficiently using methods such as the Ellipsoid method. However, we discuss a recently discovered issue with the bit complexity of the coefficients appearing in polynomials in an \Sos\ proof which could cause the Ellipsoid method to run in exponential time. This issue affects the automatizability of \Sos\ proofs. Our objective is to characterize algebraic sets i.e., constraint languages, for which the Ellipsoid method is guaranteed to run in polynomial time. We first observe that the existence of a degree $d$ \Sos\ proof implies existence of a degree $d$ \Sos\ proof on quotient rings. Hence, we only need to look at \Sos\ on quotient rings. Using this we then prove that for ideals arising from \CSP\ instances we can guarantee \Sos\ proofs with low bit complexity, provided a low degree one exists. This leads us to the third part of this section where we show that degree bounds on \Sos\ proofs can be relaxed preserving automatizability provided the \IMP\ part is polynomial time solvable. We believe this could potentially lead to more expressive \Sos\ proofs. Throughout this section we work with $\Field=\zR$.

\subsection{\Sos\ proofs on quotient ring}
In this subsection we provide a light introduction to \Sos\ proofs and discuss the issue with the bit complexity of the coefficients of polynomials appearing in a proof. At the end of this subsection we demonstrate that one can restrict himself to \Sos\ polynomials on quotient rings i.e., polynomials that are squares of $\Field$-linear combination of monomials from $\{\bx^\alpha\mid \bx^\alpha\not\in \Ideal{\LT(\I)}\}$.

It is known that, in general, nonnegativity of a polynomial is not equivalent to having a representation as a sum-of-squares polynomials \cite{hilbert1888darstellung}. The most famous example of these type of polynomials is the \emph{Motzkin polynomial} \cite{motzkin1967arithmetic}. The Motzkin polynomial $r(x,y)= 1+ x^4y^2+x^2y^4-3x^2y^2$ is nonnegative while it cannot be represented as a sum-of-squares \cite{Razborov98}. However, the situation is different when one is concerned with proofs of nonnegativity over (semi)algebraic sets. Let $S$ be the following semialgebraic set.
\begin{align}
\label{eq:def-S}
    S=\{\bx\in \zR^n\mid p_1(\bx)=0,\dots,p_{m}(\bx)=0, q_{1}(\bx)\geq 0,\dots,q_{\ell}(\bx)\geq 0\}
\end{align}
where the ideal $\I=\Ideal{p_1,\dots,p_m}$ is zero-dimensional and radical. In this case we have the following lemma, Note that radicality is crucial in the next lemma. 
\begin{lemma}[\cite{parrilo2002explicit}]
    Every nonnegative polynomial on $S$ is of the form $s_0^2+\sum_{j=1}^\ell s_j^2q_j + f$ with $f\in \I$.  
\end{lemma}

 The main appeal of \Sos\ proofs of nonnegativity is that the existence of a degree $d$ \Sos\ certificate can be formulated as the \emph{feasibility} of a semidefinite program (SDP). We say a polynomial is \Sos\ if it is equal to sum of square polynomials. Set $\zN_t^n=\{\alpha\in \zN^n\mid |\alpha|\leq t\}$. It is known, see e.g. \cite{choi1995sums,powers1998algorithm}, that a polynomial $r(\bx)=\sum_{\alpha\in \zN_{2d}^n}r_\alpha \bx^\alpha$ of degree at most $2d$ is \Sos\ if and only if the following system in the matrix variable $M = (M_{\alpha,\beta})_{\alpha,\beta\in \zN_{d}^n}$ is feasible 
        \[
            \begin{cases} 
               0\preceq M  \\
                \sum\limits_{\substack{\beta,\gamma\in \zN_{d}^n,\\ \beta+\gamma =\alpha}} M_{\beta,\gamma}  = r_{\alpha} 
            \end{cases}
        \]
This leads to a uniform and powerful tool for low degree polynomial optimization. Formally,
\begin{align*}
    &\text{Minimize} \quad r(\bx) \\
    &\text{Subject to} \quad  \bx\in S=\{\bx\in \zR^n\mid p_1(\bx)=0,\dots,p_{m}(\bx)=0, q_{1}(\bx)\geq 0,\dots,q_{\ell}(\bx)\geq 0\}
\end{align*}
An \Sos\ proof of a lower bound $r(\bx) \geq \theta$ is given by a polynomial identity of the form
\begin{align}
\label{eq:sos-proof}
    r(\bx) - \theta = \sum_{i=1}^{t_0} h_i^2(\bx) + \sum_{k=1}^\ell (\sum_{j=1}^{t_k} s_j^2(\bx)) q_k(\bx) + \sum_{i=1}^m \lambda_i(\bx) p_i(\bx).
\end{align}
The degree of an \Sos\ certificate is often defined to be the maximum degree of the polynomials involved in the proofs i.e., $\max\{deg(h_i^2),deg(s_j^2q_k),deg(\lambda_ip_i)\}$. We will see that under some conditions this bound on the degree of a proof can be relaxed which could potentially lead to a more expressive and powerful \Sos\ proof of nonnegativity framework. The discussion around this issue is presented at the end of this section. For now, our focus is on providing conditions whose presence guarantees an efficient algorithm i.e., guarantees \emph{low bit complexity} of the coefficients involved in an \Sos\ certificate of nonnegativity. 

A common misconception was that if a degree $d$ \Sos\ proof of nonnegativity exists then the corresponding SDP feasibility is solvable by the Ellipsoid method in time $O(n^d)$. This is sometimes referred to as the \Sos\ proof system being \emph{automatizable}. Unfortunately, this is not true in general \cite{ODonnell17}. Technically, the Ellipsoid method is guaranteed to work in time $poly(n^d)$ if the SDP's feasible region (should it exist) intersects a ball of radius $2^{poly(n^d)}$ \cite{GrotschelLS81}. Thus, it is not sufficient for an \Sos\ proof to exist, we also need one to exist in which all the \Sos\ polynomials can be written down with $poly(n^d)$ bits. Hence, O'Donnell posed the following question.
 
\begin{question}[\cite{ODonnell17}]
\label{Ques:d-Sos}
    Suppose there is a degree $d$ \Sos~proof that $r(\bx) \geq \theta$ subject to constraints $S=\{\bx\in \zR^n\mid p_1(\bx)=0,\dots,p_{m}(\bx)=0, q_{1}(\bx)\geq 0,\dots,q_{\ell}(\bx)\geq 0\}$, of the form
    \begin{align}
        r(\bx) - \theta = \sum_{i=1}^{t_0} h_i^2(\bx) + \sum_{k=1}^\ell (\sum_{j=1}^{t_k} s_j^2(\bx)) q_k(\bx) + \sum_{i=1}^m \lambda_i(\bx) p_i(\bx).
    \end{align}
    Is there a poly($N=n^d$)-time algorithm (presumably, a version of the Ellipsoid algorithm) that finds polynomials $h_i,s_j$ certifying $r(\bx)\geq \theta - o_N(1)$?
\end{question}

As pointed out in \cite{ODonnell17}, every SDP feasibility problem can be viewed as an \Sos\ feasibility problem modulo an ideal. A polynomial $f\in \zR[X]$ is called \emph{\Sos\ mod $\I$} if there exist $s_1,\ldots,s_t\in \zR[X]$ such that $f-\sum_{i=1}^t s_i^2\in \I$. Consequently, Question~\ref{Ques:d-Sos} is asking if the Semidefinite Feasibility Problem (SDFP) is in \textbf{P}. This is a well-known open question \cite{Ramana97,PorkolabK97,TarasovV08}.

%
Let $\I\subseteq \zR[X]$ be a zero-dimensional ideal. Fix a monomial order, here we consider \grlex order. Suppose we know $\mc{B}=\{\bx^\alpha\mid \bx^\alpha\not\in \Ideal{LT(\I)}\}$ so that $\zR[X] = \texttt{Span}_{\zR}(\mc{B})\oplus \I$. If $r= \sum s_i^2 +q$ with $s_i\in \zR[X]$, $q\in \I$, write $s_i=u_i+v_i$ with $u_i\in \texttt{Span}_{\zR}(\mc{B})$ and $v_i\in\I$. Hence, $r = \sum u_i^2 +g $  where $g = q+\sum (v_i^2 + 2u_iv_i)$ is in $\I$. This observation implies that to check the existence of an \Sos\ proof modulo $\I$ we can work with matrices indexed by $\mc{B}$ instead of the full set of monomials. Computing monomials of degree at most $d$ in $\mc{B}$ is polynomial time solvable for many combinatorial ideals as pointed out in Theorem~\ref{thm:GB+xIMP+languages}. The following lemma suggests that the existence of a degree $d$ \Sos\ proof implies existence of an \Sos\ proof where all the \Sos\ polynomials are linear combination of the monomials from $\{\bx^\alpha\mid \bx^\alpha\not\in \Ideal{\LT(\I)}\}$.

\begin{lemma}
\label{lem:sos-quotient}
     Let $\I=\Ideal{p_1\dots,p_m}$ be an ideal in $\zR[X]$ and $\mc{Q} = \{q_1,\dots,q_\ell\}$. Define $S=\{\bx\in \zR^n\mid p_1(\bx)=0,\dots,p_{m}(\bx)=0, q_{1}(\bx)\geq 0,\dots,q_{\ell}(\bx)\geq 0\}$.
     
     Suppose $r(\bx)$ has a degree $d$ \Sos\ proof of nonnegativity on $S$. Then, it has a degree $d$ \Sos\ proof of nonnegativity
    \begin{align*}
        r(\bx) = \sum_{i=1}^{t_0} h_i^2(\bx) + \sum_{k=1}^\ell (\sum_{j=1}^{t_k} s_j^2(\bx)) q_i(\bx) + \sum_{i=1}^m \lambda_i(\bx) p_i(\bx).
    \end{align*}
    such that all $h_i,s_j$ are in $\texttt{Span}(\bx^\alpha\mid |\alpha|\leq \frac{d}{2}, \bx^\alpha\not\in\Ideal{\LT(\I)})$.
\end{lemma}
\begin{proof}
    Let $G$ be a \GB\ of $\I$ and $\mc{B}_{\frac{d}{2}}$ contain the reduced monomials of degree at most $\frac{d}{2}$ with respect to $G$. Then $\mc{B}_{\frac{d}{2}}$ is a basis for $\texttt{Span}_{\zR}(\bx^\alpha\mid |\alpha|\leq \frac{d}{2}, \bx^\alpha\not\in\Ideal{\LT(\I)})$. Each non-zero polynomial $g$ of degree at most $\frac{d}{2}$ can be written as $g=u+v$ with $u\in \texttt{Span}_{\zR}(\mc{B}_{\frac{d}{2}})$ and $v\in\I$ such that $deg(g)\geq deg(u)$ and $deg(g)\geq deg(v)$. Hence,
    \begin{align*}
        r(\bx) 
        &= \sum_{i=1}^{t_0} h_i^2(\bx) + \sum_{k=1}^\ell (\sum_{j=1}^{t_k} s_j^2(\bx)) q_k(\bx) + \sum_{i=1}^m \lambda_i(\bx) p_i(\bx)\\
        &= \sum_{i=1}^{t_0} (u_{h_i}+v_{h_i})^2(\bx) + \sum_{k=1}^\ell (\sum_{j=1}^{t_i} (u_{s_j}+v_{s_j})^2(\bx)) q_k(\bx) + \sum_{i=1}^m \lambda_i(\bx) p_i(\bx)
    \end{align*}
   with $u_{h_i}$ and $u_{s_j}$ all belong to $\texttt{Span}_{\Field}(\mc{B}_{\frac{d}{2}})$ and all $v_{h_i},v_{s_j}$ are in $\I$. Expanding the above and collecting the terms that are in $\I$ in the last summation finishes the proof.
\end{proof}

\subsection{Automatizability for almost all CSP-based ideals}\label{sec:automatizability}
Here we wish to provide conditions that guarantee low bit complexity for polynomials appearing in an \Sos\ proof. This in turn guarantees automatizability i.e., if there exists a degree $d$ \Sos\ proof of nonnegativity then the Ellipsoid method is guaranteed to find such a proof. The first systematic approach to this problem is due to Raghavendra and Weitz~\cite{RaghavendraW17}. Let $\mc{P}=\{p_1\dots,p_m\}$ and $\mc{Q} = \{q_1,\dots,q_\ell\}$ and define $S\subseteq\{\mathbf{a}\in \zR^n\mid \forall p\in \mc{P}:p(\mathbf{a})=0\}$. We write $\mathbf{u}_d$ for the vector of polynomials whose entries are the elements of the usual monomial basis of $\zR[X]_d$. Similarly, we use $\mathbf{u}_d(\mathbf{a})$ for the vector of reals whose entries are the entries of $\mathbf{u}_d$ evaluated at $\mathbf{a}$. Let $\mc{U}$ denote the uniform distribution over $S$ and define the moment matrix as
\begin{align}
\label{eq:moment-matrix}
    M_{d} = \zE_{\mathbf{a}\sim \mc{U}}[\mathbf{u}_d(\mathbf{a})\mathbf{u}_d(\mathbf{a})^T].
\end{align}

In some sense, the main condition that Raghavendra and Weitz provided is that \emph{every} polynomial of degree $d$ has a proof of membership in the ideal $\I=\Ideal{\mc{P}}$ of degree at most $k$. In this case we say that $(\mc{P},\mc{Q})$ is $k$-complete on $S$ up to degree $d$. Ideally, we want $k$ to be small e.g., $k=poly(d)$. This is a very strong condition on the structure of the ideal $\I$ while the other conditions in \cite{RaghavendraW17} are considered to be mild. The conditions are as follows, 
\begin{itemize}
    \item [1.] $(\mc{P},\mc{Q})$ is $k$-complete on $S$ up to degree $d$,
    \item[2.] $S$ is $\varepsilon$-robust for $Q$. This means $\forall q\in \mc{Q},\forall \mathbf{a}\in S: q(\mathbf{a})>\varepsilon$,
    \item[3.] $S$ is $\delta$-spectrally rich for $(\mc{P},\mc{Q})$ up to degree $d$. This means every nonzero eigenvalue of $M_d$ is at least $\delta$\footnote{A zero eigenvector of $M$ corresponds to a polynomial which is zero on $S$.}. 
\end{itemize}

Surprisingly, we observe that the $k$-completeness condition is avoidable for ideals arising from \CSP s. Furthermore, the discrete nature of the varieties of these ideals implies the $\delta$-richness. In fact, one can prove similar to Lemma 7 of \cite{RaghavendraW17} that $\frac{1}{\delta}=2^{poly(n^{|D|})}$ where $D$ is the domain of the constraint language at hand. We follow a similar approach to \cite{RaghavendraW17} proving the following theorem.

\begin{theorem}
\label{thm:bit-complexity-bound}
    Let $\mc{P}$ be an instance of $\CSP(\Gamma)$ and $\I(\mc{P})=\Ideal{p_1\dots,p_m}$ be the corresponding ideal to $\mc{P}$. Define $S=\Variety{\I(\mc{P})}=\{\mathbf{a}\in \zR^n\mid \forall p\in \mc{P}:p(\mathbf{a})=0\}$ and suppose $S$ is $\varepsilon$-robust for $\mc{Q} = \{q_1,\dots,q_\ell\}$.
    
    Let $r(\bx)$ be a polynomial nonnegative on $S$, and assume $r$ has a degree $d$ \Sos\ proof of nonnegativity
    \begin{align*}
         r(\bx) = \sum_{i=1}^{t_0} h_i^2(\bx) + \sum_{k=1}^\ell (\sum_{j=1}^{t_k} s_j^2(\bx)) q_k(\bx) + \sum_{i=1}^m \lambda_i(\bx) p_i(\bx).
    \end{align*}
    Then $r$ has a degree $d$ \Sos\ proof of nonnegativity such that the coefficients of every polynomial appearing in the proof are bounded by $2^{poly(n^d, \log\frac{1}{\varepsilon})}$. In particular, if $\mc{Q}=\emptyset$ then every coefficient can be written down with only $poly(n^d)$ bits.
\end{theorem}
\begin{proof}
    Let $\I=\I(\mc{P})$ and $\mc{B}_{\frac{d}{2}}=\{\bx^\alpha\mid |\alpha|\leq \frac{d}{2}, \bx^\alpha\not\in \Ideal{\LT(I)}\}$. Let $\bv$ be the vector whose entries are the elements of $\mc{B}_d$, so any polynomial in $\texttt{Span}_{\zR}(\mc{B}_{\frac{d}{2}})$ can be expressed as $c^T \bv$, where $c$ is a vector of reals. By Lemma~\ref{lem:sos-quotient}, there is a proof of nonnegativity as follows.
    \begin{align*}
        r(\bx) &= \sum_{i=1}^{t_0} (c_i^T\bv)^2(\bx) + \sum_{k=1}^\ell \left(\sum_{j=1}^{t_k} (d_{kj}^T\bv)^2(\bx)\right) q_k(\bx) + \sum_{i=1}^m \lambda_i(\bx) p_i(\bx) \\
        &= \langle C,\bv\bv^T\rangle + \sum_{k=1}^\ell \langle D_k,\bv\bv^T\rangle q_k + \sum_{i=1}^m \lambda_i(\bx) p_i(\bx) 
    \end{align*}
    for PSD matrices $C, D_1,\dots,D_\ell$. Next, we average this polynomial identity over all the points $\mathbf{a}\in S$:
    \begin{align*}
        \zE_{\mathbf{a}\in S}[r(\mathbf{a})] = \langle C,\zE_{\mathbf{a}\in S}[\bv(\mathbf{a})\bv(\mathbf{a})^T] \rangle + \sum_{k=1}^\ell \langle D_k, \zE_{\mathbf{a}\in S}[\bv(\mathbf{a})\bv(\mathbf{a})^T q_k(\mathbf{a})] \rangle + 0
    \end{align*}
    Let $||r||$ denote the maximum absolute value of coefficients of $r$ and $||S||=\max_{\mathbf{a}\in S}||\mathbf{a}||_{\infty}$. Then the LHS is at most $poly(||r||,||S||)\leq 2^{poly(n^d)}$. The RHS is a sum of positive numbers, since the inner products are over pairs of PSD matrices. Thus the LHS is an upper bound on each term of the RHS. We would like to say that the LHS also provides an upper bound on the entries of matrices $C$, $D_k$ or an upper bound on the trace of these matrices. Recall that the trace of a matrix is sum of its eigenvalues. Hence, we prove that the averaged matrix $M'= \zE_{\mathbf{a}\in S}[\bv(\mathbf{a})\bv(\mathbf{a})^T]$ has no zero eigenvector so none of the eigenvalues of matrices $C$ and $D_k$'s are being dismissed by a zero eigenvector of $M'$.
    
    We now claim that the averaged matrix $M'= \zE_{\mathbf{a}\in S}[\bv(\mathbf{a})\bv(\mathbf{a})^T]$ has no zero eigenvector. Any zero eigenvector $c$ of $M'$ can be associated with a polynomial $c^T \bv$. Since $c^T M'c = \zE_{\mathbf{a}\in S}[(c^T \bv(\mathbf{a}))^2]$ and $c^TM'c=0$, we must have $c^T\bv(\mathbf{a})=0$ for each $\mathbf{a}\in S$. Therefore, as $\I$ is radical, by the Strong Nullstellensatz \ref{th:nullstz}, $c^T\bv$ must be in the ideal $\I$. This is a contradiction to $\bv$ being a vector of monomials from $\mc{B}_{\frac{d}{2}}$. 
    
    We proceed by assuming the averaged matrix $M'= \zE_{\mathbf{a}\in S}[\bv(\mathbf{a})\bv(\mathbf{a})^T]$ has no zero eigenvector. Therefore, none of the eigenvalues of matrix $C$ are being dismissed by a zero eigenvector of $M'$. The same is true for all $D_k$. Furthermore, similar to $M$ \eqref{eq:moment-matrix}, every nonzero eigenvalue of $M'$ is at least $\delta$, so 
    \[
        \langle C,\zE_{\mathbf{a}\in S}[\bv(\mathbf{a})\bv(\mathbf{a})^T] \rangle \geq \delta \cdot Tr(C).      
    \]
    Also, $q_k(\mathbf{a}) >\varepsilon$  for each $k$ and $\mathbf{a}$. Thus,
    \begin{align*}
        \langle D_k, \zE_{\mathbf{a}\in S}[\bv(\mathbf{a})\bv(\mathbf{a})^T q_k(\mathbf{a})] \rangle  &\geq \varepsilon \cdot \langle D_k, \zE_{\mathbf{a}\in S}[\bv(\mathbf{a})\bv(\mathbf{a})^T] \rangle\\
        & \geq \varepsilon\cdot\delta\cdot Tr(D_k)
    \end{align*}
    Thus, after averaging we have
    \[
        \delta \cdot Tr(C) + \varepsilon\cdot\delta\cdot \sum_{k=1}^\ell Tr(D_k) \leq poly(||r||,|||S|)\leq 2^{poly(n^d)}  
    \]
    Every entry of a PSD matrix is bounded by the trace, so $C$ and each $D_k$ have entries bounded by $poly(||r||, ||S||,\frac{1}{\varepsilon},\frac{1}{\delta})$. Noting that $poly(||r||,||S||)\leq 2^{poly(n^d)}$ and $\frac{1}{\delta}=2^{poly(n^d)}$ yields the desire bound on the entries of $C$ and each $D_k$. It remains to give an upper bound on the coefficients of polynomials $\lambda_i$'s. Let $||\mc{P}||$ denote the maximum absolute value of coefficients appearing in polynomials $p_1,\dots,p_m$. Consider the system of linear equations induced by 
    \[
        r(\bx) - \langle C,\bv\bv^T\rangle + \sum_{k=1}^\ell \langle D_k,\bv\bv^T\rangle q_k  =\sum_{i=1}^m \lambda_i(\bx) p_i(\bx) 
    \]
    where we take the coefficients of the $\lambda_i$ as variables. Note that this system is feasible. There are at most $O(n^d)$ variables and the coefficients on the LHS are bounded by $poly(||r||,||S||,\frac{1}{\delta},\frac{1}{\varepsilon})$. By Cramer's rule, the coefficients of $\lambda_i$'s can be taken to be bounded by $poly(||\mc{P}||,||r||,||S||,\frac{1}{\delta},\frac{1}{\varepsilon}, n!)$. Noting that $poly(||\mc{P}||,||r||,||S||)\leq 2^{poly(n^d)}$, as they are part of the input, and $\frac{1}{\delta}=2^{poly(n^d)}$ gives that this bound is at most $2^{poly(n^d,\log\frac{1}{\varepsilon})}$ as desired.
\end{proof}


The required conditions in Theorem~\ref{thm:bit-complexity-bound} are quite mild and covers almost all algebraic sets arising from \CSP-based ideals, particularly when $\mc{Q}=\emptyset$. This covers a wide range of problems for which such a low bit complexity guarantee was not known. Here we give an example of such problems, the \textsc{H-Coloring} problem. 

\paragraph{H-Coloring.} Fix a (di)graph $H$ with $V(H)=\{0,1,\dots,d\}$. For any (di)graph $G$ with vertex set $V(G)=\{x_1,\dots,x_n\}$ we are interested in homomorphisms from $G$ to $H$. That is a mapping $\phi:V(G)\to V(H)$ that preserves the edges i.e., if $x_ix_j\in E(G)$ then $\phi(x_i)\phi(x_j)\in E(H)$. The set of all homomorphisms from $G$ to $H$ can be captured by the following polynomial system. The first set of polynomials are the domain polynomials and ensures each vertex of $G$ is assigned a label from $V(H)$. The second set of polynomials are the edge constraints where they ensure that each edge $x_ix_j$ of $G$ is mapped to an edge $\alpha\beta\in E(H)$.
\begin{align*}
    \mc{P}= \Big\{\prod_{\alpha\in V(H)}(x_i-\alpha)\mid x_i\in V(G)\Big\}\cup\Big\{\prod_{\alpha\beta\in E(H)}(1-\prod_{\substack{\lambda\in V(H),\\ \lambda\neq \alpha}}\frac{\lambda-x_i}{\lambda-\alpha}\prod_{\substack{\lambda\in V(H),\\ \lambda\neq\beta}}\frac{\lambda-x_j}{\lambda-\beta})\mid x_ix_j\in E(G)\Big\}
\end{align*}

Here $S$ is the set of all homomorphisms from $G$ to $H$. This setting is very general and many important optimization problems are captured by this specification: \textsc{k-Coloring} where $H$ is a clique of size $k$, \textsc{Vertex Cover} where $H=(V=\{0,1\},E=\{(0,1),(1,0),(1,1)\})$. For both of these examples it is known that $\mc{P}$ is in fact a \GB. Our results imply that if there exists a degree $d$ \Sos\ proof of nonnegativity on $S$ then there exists one of degree $d$ with low bit complexity. Furthermore, our result implies that when $H$ is closed under a semilattice polymorphism then a truncated \GB\ can be constructed in polynomial time. This includes many graph classes such as bi-arc digraphs, interval graphs, signed interval digraphs, threshold tolerance graphs, etc. \cite{hell2016bi}.


\subsection{Relaxed degree bound for \Sos}

Degree $d$ \Sos\ relaxation insists on  $\max\{deg(h_i^2),deg(s_j^2q_k),deg(\lambda_ip_i)\} \leq d$ in \eqref{eq:sos-proof}. However, here we discuss how this restriction can be relaxed when the ideal membership part is well behaved. More precisely, we prove that if the corresponding \xIMP\ to the ideal $\Ideal{P}$ is polynomial time decidable then the existence of an \Sos\ proof with a milder degree restriction $\max\{deg(h_i^2),deg(s_j^2)\} \leq d$ is polynomial time decidable as well. 

\begin{theorem}
    Let $\Gm$ be a constraint language such that $\chi\IMP(\Gamma)$ is polynomial time decidable. Let $\mc{P}$ be an instance of $\CSP(\Gamma)$ and $\I(\mc{P})=\Ideal{p_1\dots,p_m}$ be the corresponding ideal to $\mc{P}$. Define $S=\Variety{\I(\mc{P})}$ and suppose $S$ is $\varepsilon$-robust for $\mc{Q} = \{q_1,\dots,q_\ell\}$. 
    
    For a polynomial $r$, the existence of a proof
    \[
        r(\bx) = \sum_{i=1}^{t_0} h_i^2(\bx) + \sum_{k=1}^\ell (\sum_{j=1}^{t_k} s_j^2(\bx)) q_i(\bx) + \sum_{i=1}^m \lambda_i(\bx) p_i(\bx)
    \]
    with $\max\{deg(h_i^2),deg(s_j^2)\} \leq d$ is polynomial time decidable.
\end{theorem}
\begin{proof}
    First of all, by Theorem~\ref{thm:bit-complexity-bound}, we may assume there exists a proof with \Sos\ polynomials that have coefficients with polynomial bit complexity. Second, by rearranging terms, we have
    \[
        r(\bx) - \sum_{i=1}^{t_0} h_i^2(\bx) - \sum_{k=1}^\ell (\sum_{j=1}^{t_k} s_j^2(\bx)) q_i(\bx) = \sum_{i=1}^m \lambda_i(\bx) p_i(\bx). 
    \]
    This in turn is equivalent to asking if there exist coefficients (for the \Sos\ polynomials) such that polynomial 
    \[
        r(\bx) - \sum_{i=1}^{t_0} h_i^2(\bx) - \sum_{k=1}^\ell (\sum_{j=1}^{t_k} s_j^2(\bx)) q_i(\bx) 
    \]
    is in the ideal $\I(\mc{P})$. Now this is an instance of $\chi\IMP(\Gamma)$ (note that the unknown coefficients have polynomial bit complexity.). However, this can be checked in polynomial time since $\chi\IMP(\Gamma)$ is polynomial time decidable. 
\end{proof}
Note that the above theorem is correct for any radical ideal over an algebraically closed field as long as the corresponding \xIMP\ is polynomial time decidable. Moreover, the above theorem motivates the problem of if such a relaxation could potentially lead to more powerful proof systems of nonnegativity. For instance, problems such as \textsc{Vertex Cover}, \textsc{Clique}, and \textsc{Stable Set}, are all instances of Boolean \CSP\ for which the corresponding \xIMP\ is polynomial time decidable. More generally, \xIMP\ is polynomial time decidable for the \textsc{2-SAT} problem, hence  polynomial optimization problems over \textsc{2-SAT} can enjoy the more relaxed version of \Sos\ proofs introduced here. An interesting research direction would be studying the power of this relaxed setting in terms of approximation algorithms. 
\section{Theta bodies for combinatorial ideals}
\label{sec:theta-body}
One of the core problems in optimization is to understand the $\texttt{conv}(S)$ or a relaxation of $\texttt{conv}(S)$, where $S$ the set of feasible solutions to a given problem and $\texttt{conv}(S)$ denotes the convex hull of $S$. For the \CSP s, an instance $\mathcal{P}$ of $\CSP(\Gm)$ can be associated with an ideal $\I(\mc{P})=\Ideal{f_1,\ldots,f_s}$ and the $\Sol(\mc{P})=\Variety{\I(\mc{P})}$ is a \emph{finite} subset of $\zR^n$. Hence, in this setting we are interested in computing $\texttt{conv}(\Variety{\I(\mc{P})})$. This is cut out by the inequalities $f(\bx)\geq 0$ as $f$ runs over all linear polynomials that are nonnegative over $\Variety{\I(\mc{P})}$. Thus, a natural relaxation for $\texttt{conv}(\Variety{\I(\mc{P})})$ is \footnote{Note that, in general, for an ideal $\I\in \zR[X]$ the convex hull of $\Variety{\I}$ may not be closed. Hence, \eqref{eq:cl(S)} is a relaxation for the closure of $\texttt{conv}(\Variety{\I})$, denoted by $cl(\texttt{conv}(\Variety{\I}))$. Here, closure of the convex hull of a set $C$ means the intersection of the closed halfspaces containing $C$.}
\begin{align}
\label{eq:cl(S)}
    \{ \bx\in \zR^n \mid f(\bx)\geq 0 \text{ for all } f \text{ linear and \Sos\ mod } \I\} 
\end{align}

Here we wish to provide a systematic approach using the algebraic property of the solution sets and our results on the \IMP\ and \xIMP\ to compute/approximate $\texttt{conv}(\Sol(\mc{P}))$. We follow the notion of \emph{theta bodies} introduced and advanced by Gouveia, Parrilo and Thomas \cite{GouveiaPT10-ThetaBody}. Throughout this section we work with $\Field=\zR$. Some definitions are in order.
\begin{definition}
    Let $\I$ be an ideal in $\zR[X]$. 
    \begin{itemize}
        \item [1.] A polynomial $f\in \zR[X]$ is called \emph{$k$-\Sos\ mod $\I$} if there exist $s_1,\ldots,s_t\in \zR[X]_k$ such that $f-\sum_{i=1}^t s_i^2\in \I$.
        \item[2.] The ideal $\I$ is called $k$-\Sos\ if \emph{every} nonnegative polynomial on $\Variety{\I}$ is $k$-\Sos\ mod $\I$. If \emph{every degree $d$} nonnegative polynomial on $\Variety{\I}$ is $k$-\Sos\ mod $\I$ we say $\I$ is $(d,k)$-\Sos.
    \end{itemize}
\end{definition}

Lov{\'a}sz \cite{lovasz2003semidefinite} asked the following question: Which ideals in $\zR[X]$ are $(1, 1)$-\Sos ? How about $(1,k)$-\Sos ?. We propose a restricted version of this question for vanishing ideals of constraint languages. Formally, we say a language $\Gm$ is $(d, k)$-\Sos\ if for every instance $\mc{P}$ of $\CSP(\Gm)$ the corresponding ideal $\I(\mc{P})$ is $(d, k)$-\Sos. An analogue of Lov{\'a}sz's question for constraint languages is,

\begin{problem}
\label{problem:theta-body}
    Which languages are $(1, 1)$-\Sos ? How about $(1,k)$-\Sos ?
\end{problem}

Gouveia, Parrilo and Thomas \cite{GouveiaPT10-ThetaBody} quite elegantly present an equivalent geometric notion to an ideal being $(1,k)$-\Sos. 
\begin{definition}
    For a positive integer $k$, the $k$-th \emph{theta body} of an ideal $\I\in\zR[X]$ is 
    \begin{align}
        TH_k(\I):= \{\bx\in \zR^n \mid f(\bx)\geq 0 \text{ for every linear } f \text{ that is k-\Sos\ mod } \I \}.
    \end{align}
\end{definition}

Observe that, by definition, $TH_1(\I)\supseteq TH_2(\I)\supseteq \dots \supseteq \texttt{conv}(\Variety{\I})$. We say a \emph{combinatorial} ideal is \emph{$TH_k$-exact} if $TH_k(\I)=\texttt{conv}(\Variety{\I})$\footnote{In \cite{GouveiaPT10-ThetaBody}, an ideal in $\zR[X]$ is called \emph{$TH_k$-exact} if $TH_k(\I)=cl(\texttt{conv}(\Variety{\I}))$. In general, $\texttt{conv}(\Variety{\I})$ may not be closed while the theta bodies are. Therefore, the theta body sequence of $\I$ can converge, if at all, only to $cl(\texttt{conv}(\Variety{\I}))$. }. Zero-dimensional ideals are $TH_k$-exact for some finite $k$ \cite{Laurent07}. We say a language $\Gm$ is $TH_k$-exact if for any instance $\mc{P}$ of $\CSP(\Gm)$ the ideal $\I(\mc{P})$ is $TH_k$-exact. An intriguing question is characterizing $TH_k$-exact languages, for constant $k$.
\begin{problem}
    Which languages are $TH_k$-exact, for some constant $k$?
\end{problem}

Gouveia, Parrilo and Thomas \cite{GouveiaPT10-ThetaBody} proved that a radial ideal is $TH_k$-exact if and only if it is $(1, k)$-\Sos. They have also provided an answer to Problem~\ref{problem:theta-body} characterizing finite sets $S\subset \zR^n$ that are $(1, 1)$-\Sos\ (see Theorem 4.2 in \cite{GouveiaPT10-ThetaBody}). 

Although the results of \cite{GouveiaPT10-ThetaBody} provide a characterization of $TH_k$-exact ideal, they do little to the computational side of theta bodies. Theta bodies provide a set of relaxations of solution sets and helps to formulate many combinatorial problems as optimizing a (linear) polynomial over theta bodies which could lead to a better understanding of approximability of many combinatorial problems in a unified way. In order to use theta bodies in this way, we need a method of efficiently `construct' a theta body.  
One way to describe the $k$-th theta body of $\I$ is in terms of the so called \emph{combinatorial moment matrices} which are matrices indexed by a basis of $\mc{B}_k$ \cite{Laurent07}. Let us elaborate on this via an example, see \cite{Laurent07,GouveiaPT10-ThetaBody} for further details. First we explain the construction of combinatorial moment matrices and then discuss how we can use them to describe theta bodies for the \textsc{Clique} problem. We may assume $\mc{B}=\{\bx^\alpha \mid \bx^\alpha \not\in \Ideal{\LT(\I)}\}$ and $\mc{B}_k=\{\bx^\alpha \mid |\alpha|\le k,\bx^\alpha \not\in \Ideal{\LT(\I)}\}$ are vectors with their elements listed in increasing \grlex order. Hence, any polynomial modulo $\I$ is $\mathbf{c}\cdot \mc{B}$ with $\mathbf{c}\in \zR^{|\mc{B}|}$. Now, let $\mathbf{y}=(y_1,\dots,y_m)\in \zR^{|\mc{B}_{2k}|}$ and define  $M_{\mc{B}_k}(\mathbf{y})$ to be the matrix indexed by $\mc{B}_k$ whose $(\bx^{\alpha_i},\bx^{\alpha_j})$ entry is $\mathbf{c}\cdot \mathbf{y}$ where $\mathbf{c}\in \zR^{|\mc{B}_{2k}|}$ is so that
\[
    \bx^{\alpha_i+\alpha_j}=\bx^{\alpha_i}\bx^{\alpha_j}  \equiv \mathbf{c}\cdot \mc{B}_{2k} ~~ mod ~~\I .
\]
The matrix $M_{\mc{B}_k}(\mathbf{y})$ is known as the \emph{$k$-th truncated combinatorial moment matrix}. Having $M_{\mc{B}_k}$, we define the following 
\begin{align}
    TH_k(\I)= \{ \mathbf{y}=(1,y_1,\dots,y_m)\mid \mathbf{y}\in \zR^{\mc{B}_{2k}} \text{ and } M_{\mc{B}_k}(\mathbf{y}) \succeq 0\}.
\end{align}
We now explain this construction using the \textsc{Clique} problem. Let $G=(V,E)$ be a graph with vertex set $V=\{\vc xn\}$ and edge set $E$. A clique in $G$ is a set of vertices $U$ such that for all $x_i,x_j\in U$, $x_ix_j\in E$.  The corresponding ideal is 
\[
    \I_{\textsc{Clique}}=\Ideal{x_i^2-x_i : x_i\in V, x_i\cdot x_j: x_ix_j\not\in E}.
\]
 For a subset $U\subseteq V$ let $\bx^{U}=\prod_{x_i\in U}x_i$. From definition of $\I$, it is clear that $\mc{B}=\{\bx^U\mid U \text{ is a clique in }G\}$. In particular, $\mc{B}_1=\{1,x_1,\dots,x_n\}$. Denote by $y_0,y_1,\dots,y_n$ the first $n+1$ coordinates of $\mathbf{y}\in \zR^{|\mc{B}_{2k}|}$. Then the $k$-th theta body of $\I$ is described as 
\[
    TH_k(\I_{\textsc{Clique}})=\{(1,y_1,\dots,y_n) \mid \mathbf{y}\in \zR^{|\mc{B}_{2k}|} \text{ and } M_{\mc{B}_k}(\mathbf{y}) \succeq 0\}.
\]
Equivalently, in a more intuitive way, we can describe the $k$-th theta body as follows too.

\begin{align*}
    TH_k(\I_{\textsc{Clique}}) = \left\{ \mathbf{y}\in \zR^n : 
    \begin{array}{l} 
    \exists \, M \succeq 0, \, M \in \zR^{|\mc{B}_k| \times |\mc{B}_k|}
    \,\text{ such that} \\ 
    M_{\emptyset \emptyset} = 1,\\
   M_{\emptyset \{i\}} = M_{\{i\} \emptyset} = M_{\{i\} \{i\}} = y_i \\
  M_{U U'} = 0 \,\,\text{if} \,\,U \cup U' \,\, \text{is not clique
    in}  \,\, G\\
  M_{U U'} = M_{W W'} \,\,\text{if} \,\, U \cup U' = W \cup W' 
    \end{array}
    \right\}.
\end{align*}

Another way to compute a representation of theta bodies is to work with the \emph{truncated quadratic module}. The \emph{quadratic module of $\I$} is
\[
    \mc{M}(\I) = \{ s+\I \mid s \text{ is \Sos\ in } \zR[X]\},
\]
and the $k$-th truncation of $\mc{M}(\I)$ is
\[
    \mc{M}_k(\I) = \{ s+\I \mid s \text{ is } k\text{-\Sos\ in } \zR[X]\}.
\]

Both $\mc{M}(\I)$ and $\mc{M}_k(\I)$ are cones in the $\zR$-vector space $\zR[X]/I$. Recall that $\zR[X]/I$ as a $\zR$-vector space is isomorphic to $\texttt{Span}(\bx^\alpha \mid \bx^\alpha \not\in \Ideal{\LT(\I)})$. Hence, having a monomial basis $\mc{B}=\{\bx^\alpha \mid \bx^\alpha \not\in \Ideal{\LT(\I)}\}$ or $\mc{B}_k=\{\bx^\alpha \mid |\alpha|\le k,\bx^\alpha \not\in \Ideal{\LT(\I)}\}$ would lead to an efficient computation of sum-of-squares in $\zR[X]/I$ \cite{parrilo2005exploiting}.

If the ideals arising from the combinatorial problem are produced through instance of $\CSP(\Gm)$ for some language $\Gm$, our methods make it possible to efficiently compute a representation of the corresponding theta body of one of the types described above. We say the class of $k$-th theta bodies arising from instances of $\CSP(\Gm)$, denoted by $TH_k(\Gm)$, is efficiently computable if for any instance $\mc{P}$ of $\CSP(\Gm)$ a representation of the $k$-th theta body $TH_k(\I(\mc{P}))$ can be constructed in polynomial time.  Therefore, we propose the following research problem. 

\begin{problem}
    For which languages $\Gm$ the $k$-th theta body $TH_k(\I(\mc{P}))$ is computable in polynomial time where $\mc{P}$ is an instance of $\CSP(\Gm)$?
\end{problem}

Our results on \xIMP\ imply that a monomial basis $\mc{B}_k$ as well as a $k$-truncated \GB\ can be computed in polynomial time. This together with our result on the bit complexity of \Sos\ proofs obtain the following. 
\begin{theorem}\label{the:theta}
    Let $\mathcal{H}$ be a class of ideals for which $\chi\IMP_d$ is polynomial time decidable. Then there exists a polynomial time algorithm that constructs the $d$-th theta body of an ideal $\I\in \mathcal{H}$.
\end{theorem}
\begin{proof}
    Let $\I$ be an ideal in $\mathcal{H}$. By Theorem~\ref{thm:GB+xIMP} we can construct the $d$-truncated \GB\ as well as $\mc{B}_d$. This together with the result of Laurent \cite{Laurent07} leads to a polynomial time algorithm to compute the $d$-th theta body of $\I$. 
\end{proof}
Using the reductions from Theorems~\ref{thm:search-xIMP-reduction} and \ref{thm:GB+xIMP+languages} we obtain the following results.  

\begin{corollary}
 \label{cor:Theta-body-reduction}
    Let $\Gamma$ and $\Delta$ be constraint languages on (possibly different) sets $D$, $E$, respectively. Suppose there exists a polynomial time algorithm that solves the search version of $\chi\IMP_k(\Gamma)$. Then, $O(k)$-th theta body for language $\Delta$ can be constructed in polynomial time if
    \begin{itemize}
        \item [1.] $D=E$ and $\Gamma$ pp-defines $\Delta$, or
        \item [2.] $\Gamma$ pp-interprets $\Delta$.
    \end{itemize}
 \end{corollary}

 \begin{corollary}
    \label{cor:theta-body+languages}
    Let $\Gamma$ be a finite constraint language over domain $D$. For constant $k$ and an instance $\mc{P}$ of $\CSP(\Gm)$ the $k$-th theta body of $\I(\mc{P})$, $TH_k(\I(\mc{P}))$, can be computed in polynomial time if
    \begin{itemize}
        \item[1.] $\Gamma$ has a semilattice polymorphism, or
        \item[2.] $\Gamma$ has the dual-discriminator polymorphism, or
        \item[3.] $\Gamma$ is expressed as a system of linear equations over $\GF(p)$, $p$ prime.
    \end{itemize}
\end{corollary}

Our results yield efficient construction of theta bodies for many well-studied combinatorial problems. Note that, in this case, Theorem~\ref{thm:bit-complexity-bound} guarantees low bit complexity of the coefficients in \Sos\ proofs which leads to a polynomial time execution of the Ellipsoid method. As a result optimizing a linear polynomial over such theta bodies is guaranteed to be polynomial time using the Ellipsoid method. In what follows we provide examples of well-studied problems for which the objective is optimizing a linear polynomial over integer points and (re)discover multiple positive results on computation of theta bodies. We point out there are only a handful of specific problems for which theta bodies are known to be efficiently computed. 

\paragraph{The Maximum Stable Set Problem.}
Let $G=(V,E)$ be an undirected graph with vertex set $V=\{x_1,\ldots,x_n\}$ and edge set $E$. A stable set (a.k.a independent set) in $G$ is a set of vertices $U$ such that for all $x_ix_j\in U$, $x_ix_j\not\in E$. The \textsc{Maximum Stable Set} problem seeks a stable set of largest cardinality in $G$. This can be seen as an optimization problem over language $\Gm=\{(0,0),(0,1),(1,0)\}$. More precisely, the objective is maximizing $\sum_i x_i$ subject to $(x_ix_j)\in \Gm$, for all $x_ix_j\in E$. One can observe that $\Gm$ admits the semilattice polymorphism \Min\ and hence by Corollary~\ref{cor:theta-body+languages} the $k$-th theta body of the corresponding ideal to this problem, $TH_k(G)$, can be constructed in polynomial time. We point out that, for this problem, it is known that the $1$-st theta body provides a convex relaxation for the (characteristic vectors) of all stable sets in $G$ \cite{Lovasz79-Shannon}. Therefore, the problem $\max\limits_{\bx\in TH_1(G)}\sum x_i$ is a SDP which can be solved to arbitrary precision in polynomial time in the size of $G$. The optimal value of this SDP provides an upper bound on the size of a maximum stable set. 

\paragraph{Binary Matroid and its theta bodies.}
Let $\mc{M}=(E, \mc{C})$ be a binary matroid where $E$ is called the ground set and $\mc{C}$ is a collection of subsets of $E$ that is closed under taking symmetric differences. Each member of $\mc{C}$ is called a \emph{cycle}. (Often members of $\mc{C}$ are called independent sets, however, here we call them cycles to avoid confusion.) One can view the binary matroid $\mc{M}=(E, \mc{C})$ as the $\GF(2)$-vector subspaces of $\GF(2)^E$. Let $\mathbf{1}_F\in \{0,1\}^E$ denote the characteristic vector of $F\subseteq E$, thus the cycles of the binary matroid $\mc{M}$ arise as the solutions in $\GF(2)^E$ of a linear system $M\bx = 0$, where $M$ is a matrix with columns indexed by $E$. The convex hull of the vectors $\mathbf{1}_F$, $F\in \mc{C}$, is called the \emph{matroid polytope} or \emph{cycle polytope}. Let $\I(\mc{M})\subseteq \zR[x_i\mid i\in E]$ be the vanishing ideal of the cycle vectors of $\mc{M}$. Our result in Theorem~\ref{thm:GB+xIMP+languages} implies the following theorem.
\begin{theorem}
    The set $\mc{B}_k=\{\bx^\alpha\mid |\alpha|\leq k, \bx^\alpha\not\in\Ideal{\LT(\I(\mc{M}))}\}$,  and the $k$-truncated \GB\ of the ideal $\I(\mc{M})$ can be computed in polynomial time. 
\end{theorem}

This result has numerous consequences including a polynomial time algorithm to construct $TH_k(\I(\mc{M}))$. This in turn can be viewed as computing the moment matrices for the cycle ideal $\I(\mc{M})$. One classical and well studied application is \emph{cut function} of graphs. When $\mc{M}$ is the cut polytope then $TH_k(\I(\mc{M}))$ provides a relaxation for the \emph{cut polytope}. In particular, $TH_1(\I(\mc{M}))$ coincides with the edge-relaxation considered by Rendl and Wiegele \cite{wiegele2006nonlinear} and numerical experiments there indicates that it is often tighter than the Goemans-Williamson SDP relaxation \cite{GoemansW95}. See \cite{GouveiaLPT12} for more detailed discussions and applications. 

\paragraph{The Min/Max Ones Problem and a generalization.} 
\textsc{Min Ones} and \textsc{Max Ones} are Boolean \CSP\ problems where the objective is to find a feasible solution (a 0-1 assignment satisfying all constraints) minimizing/maximizing the number of variables assigned the label $1$ \cite{KhannaSTW00}. Classical examples are the \textsc{Minimum Vertex Cover}, the \textsc{Maximum Stable Set}, and many packing and covering problems \cite{Srinivasan95}. In almost all the cases where Boolean \CSP\ is polynomial time decidable our results imply that we can construct theta bodies as relaxations of the solution space. This motivates the question of studying the power of the theta body relaxations in designing approximation algorithms for this class of problems.

\paragraph{The Min-Cost H-coloring Problem.}
In the \textsc{Maximum Stable Set} problem we have a linear cost function defined over homomorphisms from graph $G$ to the graph $H=(V=\{0,1\},E=\{(0,0),(0,1),(1,0)\})$. However, one can consider a more general setting as follows. Let $H=(V,E)$ with $V(H)=\{0,1,\ldots,d\}$ be a fixed digraph and $G=(V,E)$ with $V(G)=\{x_1,\ldots,x_n\}$ be an input digraph. Given a cost function $c:V(G)\times V(H)\to \zR_+\cup\{+\infty\}$, the objective is to find a homomorphism from $G$ to $H$ with minimum cost which is defined as $\sum_{i\in[n],j\in V(H)}c(x_i,j)$. This problem is studied under name of \textsc{Min-Cost H-coloring} and many important optimization problems are captured by this specification: \textsc{Minimum Vertex Cover}, \textsc{Minimum Sum 3-Coloring}, \textsc{Minimum Sum K-Coloring} to name a few. Only a few results are known in terms of approximation for this problem in a unified framework \cite{esa/HellMNR12,icalp/RafieyRS19}. Both of these results are based on rounding a particular LP. Our result implies that we can construct theta bodies for many cases of $H$ and have a convex relaxation of all possible homomorphisms. This motivates the question of checking if formulating the \textsc{Min-Cost H-coloring} problem as an optimization problem over theta bodies help to achieve (better) approximation algorithms.

\paragraph{The Strict \CSP s Problem.}

A natural generalization of the above problems, both in terms of alphabet and arity of the relations, is the class of \textsc{Strict CSPs} \cite{KumarMTV11}. As usual we have a constraint language $\Gamma$ on domain $D=\{0,\dots, d\}$. Given an instance $\mc{P}=(X,D,C)$ of $\CSP(\Gm)$ and a cost function $c:X\times D\to \zR_+\cup\{+\infty\}$, the objective is to minimize/maximize $\sum_{x_i\in X,j\in D}c(x_i,j)$ subject to $\Variety{\Ideal{\mc{P}}}$ i.e., finding a solution to $\mc{P}$ that minimizes/maximises $\sum_{x_i\in X} x_i$. The approximability of this problem and it special cases have been studied intensively and (\textsc{Unique Game}) hardness results are known \cite{KumarMTV11}. It is tempting to study the power of theta body relaxations for approximation of \textsc{Strict CSPs} as our results imply in many cases a theta body can be constructed in polynomial time.

\section{Conclusion and future work}

The study of \CSP-related Ideal Membership Problems is in its infancy, and pretty much all research directions are open. These include expanding the range of tractable \IMP s. A number of candidates for such expansions are readily available from the existing results about the CSP. There are however several questions that seem to be more intriguing; they mainly concern with relationship between the \IMP, \CSP\ and other problems.

The first one is how the tractability of the IMP can be used in applications such as Nullstellensatz and \Sos\ proofs. The several results we obtain here barely scratch the surface. Establishing connections of this kind seem important, because it would allow for using a much larger toolbox than the usual \GB. Note also that constructing an explicit \GB\ beyond Boolean case is getting very hard very quickly; such techniques may be not very useful in more general cases.

One of the principal techniques in solving the \CSP\ is constraint propagation, that is, the study how local interaction between constraints may tighten them, and even sometimes refute the existence of a solution. In some cases such as IMPs with the dual-discriminator polymorphism, computing the $S$-polynomial, and therefore constructing a \GB\ is equivalent to establishing so-called arc consistency. This however is not the case in general. On the other hand, constraint propagation is done through some very simple pp-definitions, and so Theorem~\ref{the:pp-definitions-intro}(1) and its proof imply that there likely are some parallel constructions with polynomials and ideals.

The main tool for solving restricted degree problems $\IMP_d(\Gm)$ is constructing a $d$-truncated \GB, in which the degrees of polynomials are also bounded by $d$. It is interesting what effect restricting the degree of polynomials in a generating set has on the properties of the underlying \CSP. More precisely, if $\I(\cP)$ is the ideal corresponding to a CSP instance $\cP$, and $\I_d(\cP)$ is the ideal generated by the reduced \GB, then $\I_d(\cP)$ can be translated back, to a less constrained \CSP\ $\cP'$. What is the connection between $\cP$ and $\cP'$? For example, every instance $\cP$ of $\CSP(\Gm)$ where $\Gm$ is Boolean and has a semilattice polymorphism, then $\cP$ is equivalent to a \textsc{Horn}- or \textsc{AntiHorn-Satisfiability}. Restricting the degrees of polynomials in $\I_d(\cP)$ is apparently equivalent to removing all clauses of length more than $d$ in $\cP$.

\section*{Acknowledgement}
The second author would like to thank Reza Dastbasteh for many valuable discussions.
\clearpage
\addcontentsline{toc}{section}{References}
\bibliographystyle{plain}
\bibliography{IMP-references}

\begin{thebibliography}{10}

\bibitem{Atserias19:proof}
Albert Atserias and Joanna Ochremiak.
\newblock Proof complexity meets algebra.
\newblock {\em {ACM} Trans. Comput. Log.}, 20(1):1:1--1:46, 2019.

\bibitem{Barto14:constraint}
Libor Barto.
\newblock Constraint satisfaction problem and universal algebra.
\newblock {\em {ACM} {SIGLOG} News}, 1(2):14--24, 2014.

\bibitem{Barto15:constraint}
Libor Barto.
\newblock The constraint satisfaction problem and universal algebra.
\newblock {\em Bull. Symb. Log.}, 21(3):319--337, 2015.

\bibitem{Barto17:polymorphisms}
Libor Barto, Andrei~A. Krokhin, and Ross Willard.
\newblock Polymorphisms, and how to use them.
\newblock In Andrei~A. Krokhin and Stanislav Zivn{\'{y}}, editors, {\em The
  Constraint Satisfaction Problem: Complexity and Approximability}, volume~7 of
  {\em Dagstuhl Follow-Ups}, pages 1--44. Schloss Dagstuhl - Leibniz-Zentrum
  f{\"{u}}r Informatik, 2017.

\bibitem{BeameIKPP94}
Paul Beame, Russell Impagliazzo, Jan Kraj{\'{\i}}cek, Toniann Pitassi, and
  Pavel Pudl{\'{a}}k.
\newblock Lower bound on hilbert's nullstellensatz and propositional proofs.
\newblock In {\em 35th Annual Symposium on Foundations of Computer Science,
  FOCS 1994, Santa Fe, New Mexico, USA, 20-22 November 1994}, pages 794--806.
  {IEEE} Computer Society, 1994.

\bibitem{becker93grobner}
Thomas Becker and Volker Weispfenning.
\newblock Gr{\"o}bner bases.
\newblock In {\em Gr{\"o}bner Bases}, pages 187--242. Springer, 1993.

\bibitem{Bharathi-Dual-Disc}
Arpitha~P. Bharathi and Monaldo Mastrolilli.
\newblock Ideal membership problem and a majority polymorphism over the ternary
  domain.
\newblock In Javier Esparza and Daniel Kr{\'{a}}l', editors, {\em 45th
  International Symposium on Mathematical Foundations of Computer Science,
  {MFCS} 2020, August 24-28, 2020, Prague, Czech Republic}, volume 170 of {\em
  LIPIcs}, pages 13:1--13:13. Schloss Dagstuhl - Leibniz-Zentrum f{\"{u}}r
  Informatik, 2020.

\bibitem{Bharathi-Minority}
Arpitha~P. Bharathi and Monaldo Mastrolilli.
\newblock Ideal membership problem for boolean minority.
\newblock {\em CoRR}, abs/2006.16422, 2020.

\bibitem{BUCHBERGER2006475}
Bruno Buchberger.
\newblock Bruno buchbergerÃ¢s phd thesis 1965: An algorithm for finding the
  basis elements of the residue class ring of a zero dimensional polynomial
  ideal.
\newblock {\em Journal of Symbolic Computation}, 41(3):475 -- 511, 2006.
\newblock Logic, Mathematics and Computer Science: Interactions in honor of
  Bruno Buchberger (60th birthday).

\bibitem{Bulatov17}
Andrei~A. Bulatov.
\newblock A dichotomy theorem for nonuniform csps.
\newblock In Chris Umans, editor, {\em 58th {IEEE} Annual Symposium on
  Foundations of Computer Science, {FOCS} 2017, Berkeley, CA, USA, October
  15-17, 2017}, pages 319--330. {IEEE} Computer Society, 2017.

\bibitem{Bulatov07:towards}
Andrei~A. Bulatov and V{\'{\i}}ctor Dalmau.
\newblock Towards a dichotomy theorem for the counting constraint satisfaction
  problem.
\newblock {\em Inf. Comput.}, 205(5):651--678, 2007.

\bibitem{BulatovJK05}
Andrei~A. Bulatov, Peter Jeavons, and Andrei~A. Krokhin.
\newblock Classifying the complexity of constraints using finite algebras.
\newblock {\em {SIAM} J. Comput.}, 34(3):720--742, 2005.

\bibitem{Burris81:universal}
S.~Burris and H.P. Sankappanavar.
\newblock {\em A course in universal algebra}, volume~78 of {\em Graduate Texts
  in Mathematics}.
\newblock Springer-Verlag, New York-Berlin, 1981.

\bibitem{BussP96}
Samuel~R. Buss and Toniann Pitassi.
\newblock Good degree bounds on nullstellensatz refutations of the induction
  principle.
\newblock In Steven Homer and Jin{-}Yi Cai, editors, {\em Proceedings of the
  11th Annual {IEEE} Conference on Computational Complexity, {CCC} 1996,
  Philadelphia, Pennsylvania, USA, May 24-27, 1996}, pages 233--242. {IEEE}
  Computer Society, 1996.

\bibitem{choi1995sums}
Man-Duen Choi, Tsit~Yuen Lam, and Bruce Reznick.
\newblock Sums of squares of real polynomials.
\newblock In {\em Proceedings of Symposia in Pure mathematics}, volume~58,
  pages 103--126. American Mathematical Society, 1995.

\bibitem{CleggEI96}
Matthew Clegg, Jeff Edmonds, and Russell Impagliazzo.
\newblock Using the groebner basis algorithm to find proofs of
  unsatisfiability.
\newblock In Gary~L. Miller, editor, {\em Proceedings of the 28th Annual {ACM}
  Symposium on the Theory of Computing, {STOC} 1996, Philadelphia,
  Pennsylvania, USA, May 22-24, 1996}, pages 174--183. {ACM}, 1996.

\bibitem{CooperCJ94}
Martin~C. Cooper, David~A. Cohen, and Peter Jeavons.
\newblock Characterising tractable constraints.
\newblock {\em Artif. Intell.}, 65(2):347--361, 1994.

\bibitem{Cox}
David Cox, John Little, and Donal OShea.
\newblock {\em Ideals, varieties, and algorithms: an introduction to
  computational algebraic geometry and commutative algebra}.
\newblock Springer Science \& Business Media, 2013.

\bibitem{DickensteinFGS91}
Alicia Dickenstein, Noa{\"{\i}} Fitchas, Marc Giusti, and Carmen Sessa.
\newblock The membership problem for unmixed polynomial ideals is solvable in
  single exponential time.
\newblock {\em Discret. Appl. Math.}, 33(1-3):73--94, 1991.

\bibitem{FGLM}
Jean{-}Charles Faug{\`{e}}re, Patrizia~M. Gianni, Daniel Lazard, and Teo Mora.
\newblock Efficient computation of zero-dimensional gr{\"{o}}bner bases by
  change of ordering.
\newblock {\em J. Symb. Comput.}, 16(4):329--344, 1993.

\bibitem{FV98}
Tom{\'{a}}s Feder and Moshe~Y. Vardi.
\newblock The computational structure of monotone monadic {SNP} and constraint
  satisfaction: {A} study through datalog and group theory.
\newblock {\em {SIAM} J. Comput.}, 28(1):57--104, 1998.

\bibitem{GoemansW95}
Michel~X. Goemans and David~P. Williamson.
\newblock Improved approximation algorithms for maximum cut and satisfiability
  problems using semidefinite programming.
\newblock {\em J. {ACM}}, 42(6):1115--1145, 1995.

\bibitem{GouveiaLPT12}
Jo{\~{a}}o Gouveia, Monique Laurent, Pablo~A. Parrilo, and Rekha~R. Thomas.
\newblock A new semidefinite programming hierarchy for cycles in binary
  matroids and cuts in graphs.
\newblock {\em Math. Program.}, 133(1-2):203--225, 2012.

\bibitem{GouveiaPT10-ThetaBody}
Jo{\~{a}}o Gouveia, Pablo~A. Parrilo, and Rekha~R. Thomas.
\newblock Theta bodies for polynomial ideals.
\newblock {\em {SIAM} J. Optim.}, 20(4):2097--2118, 2010.

\bibitem{Grigoriev98}
Dima Grigoriev.
\newblock Tseitin's tautologies and lower bounds for nullstellensatz proofs.
\newblock In {\em 39th Annual Symposium on Foundations of Computer Science,
  {FOCS} 1998, November 8-11, 1998, Palo Alto, California, {USA}}, pages
  648--652. {IEEE} Computer Society, 1998.

\bibitem{GrotschelLS81}
Martin Gr{\"{o}}tschel, L{\'{a}}szl{\'{o}} Lov{\'{a}}sz, and Alexander
  Schrijver.
\newblock The ellipsoid method and its consequences in combinatorial
  optimization.
\newblock {\em Comb.}, 1(2):169--197, 1981.

\bibitem{esa/HellMNR12}
Pavol Hell, Monaldo Mastrolilli, Mayssam~Mohammadi Nevisi, and Arash Rafiey.
\newblock Approximation of minimum cost homomorphisms.
\newblock In {\em 20th Annual European Symposium on Algorithms, {ESA} 2012,
  Ljubljana, Slovenia, September 10-12, 2012. Proceedings}, pages 587--598,
  2012.

\bibitem{hell2016bi}
Pavol Hell, Akbar Rafiey, and Arash Rafiey.
\newblock Bi-arc digraphs and conservative polymorphisms.
\newblock {\em arXiv preprint arXiv:1608.03368}, 2016.

\bibitem{hermann1926frage}
Grete Hermann.
\newblock Die frage der endlich vielen schritte in der theorie der
  polynomideale.
\newblock {\em Mathematische Annalen}, 95(1):736--788, 1926.

\bibitem{hilbert1888darstellung}
David Hilbert.
\newblock {\"U}ber die darstellung definiter formen als summe von
  formenquadraten.
\newblock {\em Mathematische Annalen}, 32(3):342--350, 1888.

\bibitem{hilbertBasisThm}
David Hilbert.
\newblock Ueber die theorie der algebraischen formen.
\newblock {\em Mathematische annalen}, 36(4):473--534, 1890.

\bibitem{hilbert1893vollen}
David Hilbert.
\newblock {\"U}ber die vollen invariantensysteme.
\newblock {\em Mathematische annalen}, 42(3):313--373, 1893.

\bibitem{JeavonsCG97}
Peter Jeavons, David~A. Cohen, and Marc Gyssens.
\newblock Closure properties of constraints.
\newblock {\em J. {ACM}}, 44(4):527--548, 1997.

\bibitem{JeffersonJGD13}
Christopher Jefferson, Peter Jeavons, Martin~J. Green, and M.~R.~C. van Dongen.
\newblock Representing and solving finite-domain constraint problems using
  systems of polynomials.
\newblock {\em Annals of Mathematics and Artificial Intelligence},
  67(3):359--382, Mar 2013.

\bibitem{KhannaSTW00}
Sanjeev Khanna, Madhu Sudan, Luca Trevisan, and David~P. Williamson.
\newblock The approximability of constraint satisfaction problems.
\newblock {\em {SIAM} J. Comput.}, 30(6):1863--1920, 2000.

\bibitem{Kolaitis98:conjunctive}
Phokion~G. Kolaitis and Moshe~Y. Vardi.
\newblock Conjunctive-query containment and constraint satisfaction.
\newblock In Alberto~O. Mendelzon and Jan Paredaens, editors, {\em Proceedings
  of the Seventeenth {ACM} {SIGACT-SIGMOD-SIGART} Symposium on Principles of
  Database Systems, June 1-3, 1998, Seattle, Washington, {USA}}, pages
  205--213. {ACM} Press, 1998.

\bibitem{KumarMTV11}
Amit Kumar, Rajsekar Manokaran, Madhur Tulsiani, and Nisheeth~K. Vishnoi.
\newblock On lp-based approximability for strict csps.
\newblock In {\em Proceedings of the 22nd Annual {ACM-SIAM} Symposium on
  Discrete Algorithms, {SODA} 2011, San Francisco, California, USA, January
  23-25, 2011}, pages 1560--1573, 2011.

\bibitem{Laurent07}
Monique Laurent.
\newblock Semidefinite representations for finite varieties.
\newblock {\em Math. Program.}, 109(1):1--26, 2007.

\bibitem{Lovasz79-Shannon}
L{\'{a}}szl{\'{o}} Lov{\'{a}}sz.
\newblock On the shannon capacity of a graph.
\newblock {\em {IEEE} Trans. Inf. Theory}, 25(1):1--7, 1979.

\bibitem{lovasz2003semidefinite}
L{\'a}szl{\'o} Lov{\'a}sz.
\newblock Semidefinite programs and combinatorial optimization.
\newblock In {\em Recent advances in algorithms and combinatorics}, pages
  137--194. Springer, 2003.

\bibitem{Mastrolilli19}
Monaldo Mastrolilli.
\newblock The complexity of the ideal membership problem for constrained
  problems over the boolean domain.
\newblock In {\em Proceedings of the 13th Annual {ACM-SIAM} Symposium on
  Discrete Algorithms, {SODA} 2019, San Diego, California, USA, January 6-9,
  2019}, pages 456--475, 2019.

\bibitem{Mayr89}
Ernst~W. Mayr.
\newblock Membership in plynomial ideals over {Q} is exponential space
  complete.
\newblock In Burkhard Monien and Robert Cori, editors, {\em 6th Annual
  Symposium on Theoretical Aspects of Computer Science, {STACS} 1989,
  Paderborn, FRG, February 16-18, 1989, Proceedings}, volume 349 of {\em
  Lecture Notes in Computer Science}, pages 400--406. Springer, 1989.

\bibitem{mayr1982complexity}
Ernst~W Mayr and Albert~R Meyer.
\newblock The complexity of the word problems for commutative semigroups and
  polynomial ideals.
\newblock {\em Advances in mathematics}, 46(3):305--329, 1982.

\bibitem{mckenzie2018algebras}
Ralph~N McKenzie, George~F McNulty, and Walter~F Taylor.
\newblock {\em Algebras, lattices, varieties}, volume 383.
\newblock American Mathematical Soc., 2018.

\bibitem{motzkin1967arithmetic}
T.S. Motzkin.
\newblock The arithmetic-geometric inequality, inequalities (proc. sympos.
  wright-patterson air force base, ohio, 1965), 1967.

\bibitem{ODonnell17}
Ryan O'Donnell.
\newblock {SOS} is not obviously automatizable, even approximately.
\newblock In Christos~H. Papadimitriou, editor, {\em 8th Innovations in
  Theoretical Computer Science Conference, {ITCS} 2017, January 9-11, 2017,
  Berkeley, CA, {USA}}, volume~67 of {\em LIPIcs}, pages 59:1--59:10. Schloss
  Dagstuhl - Leibniz-Zentrum f{\"{u}}r Informatik, 2017.

\bibitem{Papert64:congruence}
Dona Papert.
\newblock Congruence relations in semi-lattices.
\newblock {\em J. London Math. Soc.}, 39:723--729, 1964.

\bibitem{parrilo2002explicit}
Pablo~A Parrilo.
\newblock An explicit construction of distinguished representations of
  polynomials nonnegative over finite sets.
\newblock {\em IfA AUT02-02, ETH Z{\"u}rich}, 2002.

\bibitem{parrilo2005exploiting}
Pablo~A Parrilo.
\newblock Exploiting algebraic structure in sum of squares programs.
\newblock In {\em Positive polynomials in control}, pages 181--194. Springer,
  2005.

\bibitem{phillips2003interpolation}
George~M Phillips.
\newblock {\em Interpolation and approximation by polynomials}, volume~14.
\newblock Springer Science \& Business Media, 2003.

\bibitem{PorkolabK97}
Lorant Porkolab and Leonid Khachiyan.
\newblock On the complexity of semidefinite programs.
\newblock {\em J. Glob. Optim.}, 10(4):351--365, 1997.

\bibitem{post1941two}
EL~Post.
\newblock The two-valued iterative systems of mathematical logic. number 5 in
  annals of math.
\newblock {\em Studies. Princeton Univ. Press}, 1941.

\bibitem{powers1998algorithm}
Victoria Powers and Thorsten W{\"o}rmann.
\newblock An algorithm for sums of squares of real polynomials.
\newblock {\em Journal of pure and applied algebra}, 127(1):99--104, 1998.

\bibitem{icalp/RafieyRS19}
Akbar Rafiey, Arash Rafiey, and Thiago Santos.
\newblock Toward a dichotomy for approximation of h-coloring.
\newblock In {\em 46th International Colloquium on Automata, Languages, and
  Programming, {ICALP} 2019, July 9-12, 2019, Patras, Greece}, pages
  91:1--91:16, 2019.

\bibitem{RaghavendraW17}
Prasad Raghavendra and Benjamin Weitz.
\newblock On the bit complexity of sum-of-squares proofs.
\newblock In Ioannis Chatzigiannakis, Piotr Indyk, Fabian Kuhn, and Anca
  Muscholl, editors, {\em 44th International Colloquium on Automata, Languages,
  and Programming, {ICALP} 2017, July 10-14, 2017, Warsaw, Poland}, volume~80
  of {\em LIPIcs}, pages 80:1--80:13. Schloss Dagstuhl - Leibniz-Zentrum
  f{\"{u}}r Informatik, 2017.

\bibitem{Ramana97}
Motakuri~V. Ramana.
\newblock An exact duality theory for semidefinite programming and its
  complexity implications.
\newblock {\em Math. Program.}, 77:129--162, 1997.

\bibitem{Razborov98}
Alexander~A. Razborov.
\newblock Lower bounds for the polynomial calculus.
\newblock {\em Comput. Complex.}, 7(4):291--324, 1998.

\bibitem{richman1974constructive}
Fred Richman.
\newblock Constructive aspects of noetherian rings.
\newblock {\em Proceedings of the American Mathematical Society},
  44(2):436--441, 1974.

\bibitem{Schaefer78}
Thomas~J. Schaefer.
\newblock The complexity of satisfiability problems.
\newblock In Richard~J. Lipton, Walter~A. Burkhard, Walter~J. Savitch, Emily~P.
  Friedman, and Alfred~V. Aho, editors, {\em Proceedings of the 10th Annual
  {ACM} Symposium on Theory of Computing, {STOC} 1978, May 1-3, 1978, San
  Diego, California, {USA}}, pages 216--226. {ACM}, 1978.

\bibitem{seidenberg1974constructions}
Abraham Seidenberg.
\newblock Constructions in algebra.
\newblock {\em Transactions of the American Mathematical Society},
  197:273--313, 1974.

\bibitem{Srinivasan95}
Aravind Srinivasan.
\newblock Improved approximations of packing and covering problems.
\newblock In {\em Proceedings of the 27th Annual {ACM} Symposium on Theory of
  Computing, STOC 1995, 29 May-1 June 1995, Las Vegas, Nevada, {USA}}, pages
  268--276, 1995.

\bibitem{szendrei1986clones}
{\'A}gnes Szendrei.
\newblock Clones in universal algebra.
\newblock {\em Les presses de L'universite de Montreal}, 1986.

\bibitem{TarasovV08}
Sergey~P. Tarasov and Mikhail~N. Vyalyi.
\newblock Semidefinite programming and arithmetic circuit evaluation.
\newblock {\em Discret. Appl. Math.}, 156(11):2070--2078, 2008.

\bibitem{vandongenPhd}
Marc~R.C. van Dongen.
\newblock {\em Constraints, Varieties, and Algorithms}.
\newblock PhD thesis, Department of Computer Science, University College, Cork,
  Ireland, 2002.

\bibitem{weitz-PhD}
Benjamin Weitz.
\newblock {\em Polynomial proof systems, effective derivations, and their
  applications in the sum-of-squares hierarchy}.
\newblock PhD thesis, UC Berkeley, 2017.

\bibitem{wiegele2006nonlinear}
Angelika Wiegele.
\newblock {\em Nonlinear optimization techniques applied to combinatorial
  optimization problems}.
\newblock na, 2006.

\bibitem{Zhuk17}
Dmitriy Zhuk.
\newblock A proof of {CSP} dichotomy conjecture.
\newblock In Chris Umans, editor, {\em 58th {IEEE} Annual Symposium on
  Foundations of Computer Science, {FOCS} 2017, Berkeley, CA, USA, October
  15-17, 2017}, pages 331--342. {IEEE} Computer Society, 2017.

\end{thebibliography}

\appendix
\section{\GBs\ for linear equations mod $p$ via conversion technique}\label{sec:lin-mod-p}
In this section we focus on constraint languages that are expressible as a system of linear equations modulo a prime number. Let $\Gm$ be a constraint language over a set $D$ with $|D|=p$, and $p$ a prime number. Suppose $\Gm$ has an \emph{affine} polymorphism modulo $p$ (i.e.\ a ternary operation $\psi(x,y,z)=x\ominus y\oplus z$, where $\oplus,\ominus$ are addition and subtraction modulo $p$, or, equivalently, of the field $\GF(p)$). In this case every \CSP\ can be represented as a system of linear equations over $\GF(p)$. Without loss of generality, we may assume that the system of linear equations at hand is already in the \emph{reduced row echelon} form. Transforming system of linear equations mod $p$ in its reduced row echelon form to a system of polynomials in $\zR[X]$ that are a \GB\ is not immediate and requires substantial work. This is the case even if we restrict ourselves to lexicographic order. Let us elaborate on this by an example considering linear equations over $GF(2)$.

\begin{example}
    We assume a lexicographic order $\succ$ with $x_1 \succ_\lex \dots \succ_\lex x_n$. We also assume that the linear system has $r \le n$ equations and is already in its reduced row echelon form with $x_i$ as the leading monomial of the $i$-th equation. Let $Supp_i \subset [n]$ such that $\{x_j : j \in Supp_i\}$ is the set of variables appearing in the $i$-th equation of the linear system except for $x_i$. Let the $i$-th equation be $g_i = 0 (mod~2)$ where $g_i=x_i\oplus f_i$, with $i \in [r]$ and $f_i$ is the Boolean function $(\bigoplus_{j\in Supp_i} x_j)\oplus \alpha_i$ and $\alpha_i \in \{0,1\}$. 
    Define a polynomial $M(f_i)\in\zR[\vc x n]$ interpolating $f_i$, that is, such that, for every $\mb{a}\in \{0,1\}^n$, $f_i(\mb{a})=0$ if and only if $M(f_i)(\mb{a})=0$, and $f_i(\mb{a})=1$ if and only if $M(f_i)(\mb{a})=1$. Now, consider the following set of polynomials. 
    \[
        G=\{ x_1 - M(f_1),\dots, x_r-M(f_r) ,(x^2_{r+1}-x_{r+1}),\dots,(x^2_{n}-x_n)\}.
    \] 
    Set $G\subset \zR[\vc x n]$ is a \GB\ with respect to \lex order. This is because for every pair of polynomials in $G$ the reduced $S$-polynomial is zero as the leading monomials of any two polynomials in $G$ are relatively prime. By Buchberger's Criterion (see Theorem~\ref{th:crit}) it follows that $G$ is a \GB\ with respect to the \lex ordering.  
    
However, this construction may not be computationally efficient as the polynomials $M(f_i)$ may have exponentially many monomials. This case was overlooked in \cite{Mastrolilli19}. Bharathi and Mastrolilli~\cite{Bharathi-Minority} resolved this issue in an elegant way. Having $G'=\{ x_1 - f_1,\dots, x_r- f_r ,(x^2_{r+1}-x_{r+1}),\dots,(x^2_{n}-x_n)\}$, they convert $G'$ to set of polynomials in $\zR[\vc x n]$ which is a $d$-truncated \GB. Their conversion algorithm has time complexity $n^{O(d)}$ where $d=O(1)$. Their algorithm is a modification of the conversion algorithm by Faug{\`{e}}re, Gianni, Lazard and Mora \cite{FGLM}. See \cite{Bharathi-Minority} for more details.
\end{example}

We consider this problem for any fixed prime $p$ and prove that a $d$-truncated \GB\ can be computed in time $n^{O(d)}$. First, we give a very brief introduction to the FGLM conversion algorithm~\cite{FGLM}. Next, we present our conversion algorithm that, given a system of linear equations mod $p$, produces a $d$-truncated \GB\ in graded lexicographic order. The heart of our algorithm is finding linearly independent expressions mod $p$ that helps us carry the conversion. 

\subsection{\GB\ conversion}

We say a \GB~$G=\{\vc g t\}$ is \emph{reduced} if $\LC(g_i)=1$ for all $g_i\in G$, and if for all $g_i\in G$ no monomial of $g_i$ lies in $\Ideal{\LT(G\setminus \{g_i\})}$. We note that for an ideal and a given monomial ordering the reduced \GB\ of $\I$ is unique (see, e.g., \cite{Cox}, Theorem 5, p.93). Given the reduced \GB\ of a zero-dimensional ideal $\I \subset \Field[X]$ with respect to a monomial order $\succ_1$, where $\Field$ is a computable field, the FGLM algorithm computes a \GB\ of $\I$ with respect to another monomial order $\succ_2$. The complexity of the FGLM algorithm depends on the dimension of $\Field$-vector space $\Field[X]/\I$. More precisely, let $\mathscr{D}(I)$ denote the dimension of $\Field$-vector space $\Field[X]/\I$, then we have the following proposition.
 
\begin{proposition}[Proposition 4.1 in \cite{FGLM}]
    Let $\I$ be a zero-dimensional ideal and $G_1$ be the reduce \GB\ with respect to an ordering $\succ_1$. Given a different ordering $\succ_2$, there is an algorithm that constructs a \GB\ $G_2$ with respect to ordering $\succ_2$ in time $O(n\mathscr{D}(I)^3)$.
\end{proposition}

We cannot apply the FGLM algorithm directly as $\mathscr{D}(\I)$ could be exponentially large in our setting. Note that $\mathscr{D}(\I)$ is equal to the number of common zeros of the polynomials from $\Ideal{G_1}$, which in the case of linear equations is $\mathscr{D}(\I)=O(p^{n-r})$ where $r$ is the number of equations in the reduced row echelon form. Furthermore, as we discussed, we are not given the explicit reduced \GB\ $G_1$ with respect to \lex ordering (the \GB\ is presented to us as a system of linear equation mod $p$ rather than polynomials in $\zR[X]$). We shall present an algorithm that resolves these issues.  

Let $\cP$ be an instance of $\CSP(\Gm)$ that is expressed as a system of linear equations $\mathscr{S}$ over $\zZ_{p}$ with variables $x_1,\dots,x_n$. A system of linear equations over $\zZ_{p}$ can be solved by Gaussian elimination (this immediately tells us if $1 \in \I(\cP)$ or not, and we proceed only if $1 \not\in \I(\cP)$). We assume a lexicographic order $\succ_\lex$ with $x_1 \succ_\lex \dots \succ_\lex x_n$. We also assume that the linear system has $r \le n$ equations and it is already in its reduced row echelon form with $x_i$ as the leading monomial of the $i$-th equation. Let $Supp_i \subset [n]$ such that $\{x_j : j \in Supp_i\}$ be the set of variables appearing in the $i$-th equation of the linear system except for $x_i$.
 
 Fix a prime $p$ and $\oplus$, $\ominus$, $\odot$ denote addition, subtraction, and multiplication modulo $p$, respectively. We will call a linear polynomial over $\zZ_p$ a \emph{$p$-expression}. Let the $i$-th equation be $g_i = 0\pmod p$ where $g_i := x_i\ominus f_i$, with $i \in [r]$ and $f_i$ is the $p$-expression $(\bigoplus_{j\in Supp_i} \alpha_jx_j)\oplus \alpha_i$ and $\al_j,\alpha_i \in \zZ_p$. We will assume that each variable $x_i$
is associated with its $p$-expression $f_i$ which comes from the mod $p$ equations. This is clear for $i \le r$; for $i > r$ the $p$-expression $f_i = x_i$ itself. Hence, we can write down the reduced \GB\ in the \lex order in an implicit form as follows.
\begin{align}
\label{eq:G1-appendix}
    G_1=\{ x_1 - f_1,\dots, x_r-f_r ,\prod_{i\in\zZ_p}(x_{r+1}-i),\dots,\prod_{i\in\zZ_p}(x_{n}-i)\}
\end{align}

\begin{algorithm}[t]
  \caption{Conversion algorithm}\label{alg:FGLM}
  \begin{algorithmic}[1]
    \Require{$G_1$ as in \eqref{eq:G1-appendix} that corresponds to $\I(\cP)$, degree $d$.}
    \State Let $Q$ be the list of all monomials of degree at most $d$ arranged in increasing order with respect to \grlex.
    \State $G_2=\emptyset, B(G_2)=\{1\}$ (we assume $1\not\in\I(\cP)$). Let $b_i$ (arranged in increasing \grlex order) be the elements of $B(G_2)$.
    \For{$q \in Q$}
    \If{$q$ is divisible by some \LM~in $G_2$}
        \State Discard it and go to Step 3,
    \EndIf
    \If{$\reduce q {G_1} = \sum\limits_{j}k_j \reduce {b_j} {G_1}$}
        \Comment{where $k_j\in \zR, b_j\in B(G_2)$}
        \State $G_2 = G_2 \cup \{q-\sum\limits_{j}k_j  b_j\}$
    \Else 
        \State $B(G_2) = B(G_2) \cup \{q\}$ 
    \EndIf
    \EndFor
    \State \textbf{return} $G_2$
\end{algorithmic}
\end{algorithm}

Given $G_1$, our conversion algorithm, Algorithm~\ref{alg:FGLM}, constructs a $d$-truncated \GB\ over $\Real[x_1,\dots, x_n]$ with respect to the \grlex order. At the beginning of the algorithm, there will be two sets: $G_2$, which is initially empty but will become the new \GB\ with respect to the \grlex order, and $B(G_2)$, which initially contains $1$ and will grow to be the \grlex monomial basis of the quotient ring $\zR[\vc x n]/\I(\cP)$ as a $\zR$-vector space. In fact, $B(G_2)$ contains the reduced monomials (of degree at most $d$) with respect to $G_2$. Every $f \in \zR[\vc x n]$ is congruent modulo $\I(\cP)$ to a unique polynomial $r$ which is a $\zR$-linear combination of the monomials in the \emph{complement} of $\Ideal{\LT(\I(\cP))}$. Furthermore, the elements of $\{\bx^\al\mid \bx^\al \not\in \Ideal{\LT(\I(\cP))}\}$ are "linearly independent modulo $\I(\cP)$" (see Proposition~\ref{prop:quotient-basis}). This suggests the following. In Algorithm~\ref{alg:FGLM}, $Q$ is the list of all monomials of degree at most $d$ arranged in increasing order with respect to \grlex ordering. The algorithm iterates over monomials in $Q$ in increasing \grlex order and at each iteration decides exactly one of the followings given the current sets $G_2$ and $B(G_2)$. 
\begin{enumerate}
\item 
$q$ should be discarded (if $q$ is divisible by some \LM~in $G_2$), or
\item 
a polynomial with $q$ as its leading monomial should be added to $G_2$ (if $\reduce q {G_1} = \sum\limits_{j}k_j \reduce {b_j} {G_1}$; $b_j\in B(G_2)$), or
\item 
$q$ should be added to $B(G_2)$.
\end{enumerate}

 The trickiest part is to decide if the current monomial $\reduce q {G_1}$ is a $\zR$-linear combination of $\reduce {b_j} {G_1}$ with $b_j$ being the current elements in $B(G_2)$. Provided this can be done correctly and in polynomial time, the correctness of Algorithm~\ref{alg:FGLM} follows by the analyses in \cite{FGLM} and it runs in polynomial time as there are at most $O(n^d)$ monomials in $Q$. The rest of this section is devoted to provide a polynomial time procedure that correctly decides if $\reduce q {G_1} = \sum\limits_{j}k_j \reduce {b_j} {G_1}$ holds for the current monomial $q$ and the current $b_j$s in $B(G_2)$.

\subsection{Expansion in a basis of $p$-expressions}
 For a monomial $q$, the normal form of $q$ by $G_1$, $\reduce q {G_1}$, is the remainder on division of $q$ by $G_1$ in the \lex order. $\reduce q {G_1}$ is unique and it does not matter how the elements of $G_1$ are listed when using the division algorithm. Here, it suffices for us to write $\reduce q {G_1}$ in terms of product of $p$-expressions.
 We start with a simple observation. Recall that $r$ denotes the number of linear equations in the reduced row echelon form.

\begin{observation}\label{reduced-to-G1}
    Let $q = x_1^{\al_1}\cdots x_n^{\al_n}$ be a monomial such that for all $r< i$ we have $\al_i\leq p-1$. Then, $ q\vert_{G_1}= f_1^{\al_1}\cdots f_n^{\al_n}$ where each $f_{i}$ is the $p$-expression associated to $x_i$. 
\end{observation}

A keystone of our conversion algorithm is a relation between a product of $p$-expressions and a sum of $p$-expressions. Intuitively speaking, we will prove that a product of $p$-expressions can be written as a $\zR$-linear combination of (linearly independent) $p$-expressions. Indeed, we provide a set of $p$-expressions and prove the $p$-expressions in this set are linearly independent and span the space of functions from $\zZ^d_p$ to $\zC$. Consider the set $\cV_{n,p}$ of functions from $\zZ^n_p$ to $\zC$ as a $p^n$-dimensional vector space, whose components are values of the function at the corresponding point of $\zZ^n_p$. Let
\[
    F_n=\left\{\bigoplus_{i=1}^n\al_ix_i\oplus x_{n+1}\oplus \beta\mid \al_i\in\{0\zd p-1\}, 
    i\in[n], \beta\in\{0\zd p-2\}\right\}
\]
be a collection of linear functions over $\zZ_p$, and let 
\[
    \cF_n=F_n\cup\dots\cup F_0\cup\{1\}. 
\]

\begin{theorem}\label{the:lin-independent-MainBody}
For any $n$, the collection $\cF_n$ of $p$-expressions is linearly independent as a set of vectors from $\cV_{n+1,p}$ and forms a basis of $\cV_{n+1,p}$.
\end{theorem}

As a first application of \Cref{the:lin-independent-MainBody} we show that any $p$-expression can be written as a $\zR$-linear combination of $p$-expressions in our basis $\cF_n$.

\begin{lemma}
\label{lem:p-exp-sum-to-basis}
     Any $p$-expression $\al_1x_1\oplus \al_2x_2\oplus \cdots \oplus \al_nx_n\oplus \beta$ with $\al_n\neq 0$ can be represented by a $\zR$-linear combination of the $p$-expression basis from $\cF_{n-1}$.
\end{lemma}

\begin{proof}
    By \Cref{the:lin-independent-MainBody}, for any $x$ and $y$, the expression $y\oplus \al x\oplus \beta$ can be written as a $\zR$-linear combination of $p$-expressions from $\cF_1$. Such a $\zR$-linear combination can be found in constant time $p^{O(1)}$ as the number of functions in $\cF_1$ is $p^2$. We continue by assuming such a $\zR$-linear combination of any $y\oplus \al x\oplus \beta$ is provided to us.

    Now consider $\al_1x_1\oplus \al_2x_2\oplus \cdots \oplus \al_nx_n\oplus \beta$. Introduce a new variable $y$ and set $y=\al_1x_1\oplus \al_2x_2\oplus \cdots \oplus \al_{n-1}x_{n-1}$. Hence, $\al_1x_1\oplus \al_2x_2\oplus \cdots \oplus \al_nx_n\oplus \beta=y\oplus \al_nx_n\oplus \beta$. 
    By the above discussion, $y\oplus \al_nx_n\oplus \beta$ can be written as a $\zR$-linear combination of $p$-expressions from $\cF_1$. Therefore, there exist $c_{\gm},c_{\al\gm},\kappa\in \zR$ so that
    \begin{align}
        y\oplus \al_nx_n\oplus \beta = \sum\limits_{\gm=0}^{p-2} c_{\gm}(y\oplus \gm) + \sum\limits_{\al\in[p-1],\gm\in[p-2]} c_{\al\gm}(\al y \oplus x_n\oplus \gm)+\kappa
    \end{align}
    Note that $p$-expressions $y\oplus \gm$ are in $F_0$, $p$-expressions $\al y \oplus x_n\oplus \gm$ are in $F_1$, and $\kappa$ is a constant. Substituting back for $y$, we observe that the second sum on the right hand side is already a $\zR$-linear combination of $p$-expressions from $F_{n-1}$. Consider the first sum.
    \begin{align}
    \label{eq:p-2}
         \sum\limits_{\gm=0}^{p-2} c_{\gm}(y\oplus \gm) = \sum\limits_{\gm=0}^{p-2} c_{\gm}(\al_1x_1\oplus \al_2x_2\oplus \cdots \oplus \al_{n-1}x_{n-1}\oplus \gm)
    \end{align}
    Now set $y=\al_1x_1\oplus \al_2x_2\oplus \cdots \oplus \al_{n-2}x_{n-2}$. This gives
    \begin{align}
        \sum\limits_{\gm=0}^{p-2} c_{\gm}(\al_1x_1\oplus \al_2x_2\oplus \cdots \oplus \al_{n-1}x_{n-1}\oplus \gm)= \sum\limits_{\gm=0}^{p-2} c_{\gm}(y \oplus \al_{n-1}x_{n-1}\oplus \gm)
    \end{align}
    Note that each term $y\oplus\al_{n-1}x_{n-1}\oplus \gm$ is expressible as a $\zR$-linear combination of $p$-expressions from $\cF_1$. This leads us to the following. 
    
    \begin{align}
    \nonumber
        \sum\limits_{\gm=0}^{p-2} c_{\gm}(y \oplus \al_{n-1}x_{n-1}\oplus \gm) 
        &=
        \sum\limits_{\gm=0}^{p-2} c_{\gm}\left(
        \sum\limits_{\dl=0}^{p-2} c_{\dl}(y\oplus \dl) + \sum\limits_{\al\in[p-1],\dl\in[p-2]} c_{\al\dl}(\al y \oplus x_{n-1}\oplus \dl)+\kappa'
        \right)\\
    \nonumber    
        &=\sum\limits_{\gm=0}^{p-2} c_{\gm}\left(
        \sum\limits_{\dl=0}^{p-2} c_{\dl}(y\oplus \dl)\right) + \sum\limits_{\gm=0}^{p-2} c_{\gm}\left( \sum\limits_{\substack{\al\in[p-1],\\
        \dl\in[p-2]}} c_{\al\dl}(\al y \oplus x_{n-1}\oplus \dl)\right)+\sum\limits_{\gm=0}^{p-2} c_{\gm}\kappa'\\
    \label{eq:p-2-3}    
        &=\sum\limits_{\gm=0}^{p-2} c'_{\gm}(y\oplus \gm)+\sum\limits_{\substack{\al\in[p-1],\\ \gm\in[p-2]}} c'_{\al\gm}(\al y \oplus x_{n-1}\oplus \gm)+\kappa''
    \end{align}
    Note that $p$-expressions $y\oplus \gm$ are in $F_0$, $p$-expressions $\al y \oplus x_{n-1}\oplus \gm$ are in $F_1$, and $\kappa''$ is a constant. Substituting back for $y$, we observe that the second sum on the right hand side is already a $\zR$-linear combination of $p$-expressions from $F_{n-2}$. Hence, it suffices to continue with the first term of the sum \eqref{eq:p-2-3} which we can handle similar to the above procedure.
    
    All in all, it requires to repeat the above procedure $n$ times where at the $i$-th iteration we deal with a sum of $p-2$ $p$-expressions of form $ \bigoplus\limits_{i=1}^{n-i}\al_ix_i\oplus \beta$. This results in $O(np^{O(1)})$ running time.
\end{proof}

Another application of \Cref{the:lin-independent-MainBody} is transforming a multiplication of $p$-expressions to an equivalent $\zR$-linear combination of the basis in $\cF_n$. Suppose  $x_1,\dots,x_d$ are (not necessary distinct) variables that take values $0,\dots,p-1$. Let $x_1\cdot x_2\cdots x_d$ be their multiplication. In general, we are interested in a multiplication of $p$-expressions however, let us first discuss the simpler case of $x_1\cdot x_2\cdots x_d$. Unfortunately, the trick we used in the proof of \Cref{lem:p-exp-sum-to-basis} is no longer effective here. However, assuming $d$ is a constant makes the situation easier. $x_1\cdot x_2\cdots x_d$ is a $p^d$-dimensional vector. By \Cref{the:lin-independent-MainBody}, the set $\cF_{d-1}$ of $p$-expressions spans the set $\cV_{d,p}$ of functions from $\zZ^d_p$ to $\zC$ as a $p^d$-dimensional vector space. Hence, in constant time (depending on $p$ and $d$), we can have a $\zR$-linear combination of the basis in $\cF_{d-1}$ that represents $x_1\cdot x_2\cdots x_d$. We continue by assuming such a $\zR$-linear combination of any $x_1\cdot x_2\cdots x_d$ is provided to us. The next lemma states that we can have a $\zR$-linear combination of the basis in $\cF_{n}$ for any $h_1\cdot h_2\cdots h_d$ where each $h_i$ is a $p$-expression over variables $x_1,\dots,x_n$.

\begin{lemma}
\label{lem:p-exp-multiplication-to-basis}
    Let $h_1, h_2,\dots, h_d$ be (not necessary distinct) $p$-expressions over variables $x_1,\dots,x_n$. The product $\mathscr{M}=h_1\cdot h_2\cdots h_d$ viewed as a function from $\zZ_p^n$ to $\Complex$ can be represented as a $\zR$-linear combination of the basis in $\cF_{n-1}$.
\end{lemma}

\begin{proof}
    Let us treat $h_i$s as indeterminates. Define 
    \[
        H_t=\left\{\bigoplus_{i=1}^t\al_ih_i\oplus h_{t+1}\oplus \beta\mid \al_i\in\{0\zd p-1\}, 
        i\in[t], \beta\in\{0\zd p-2\}\right\}
    \]
    to be a collection of linear functions over $\zZ_p$, and let 
    \[
    \mc{H}_{d-1}=H_{d-1}\cup\dots\cup H_0\cup\{1\}. 
    \]
     By \Cref{the:lin-independent-MainBody} and the above discussion, $\mathscr{M}$ can be written as a $\zR$-linear combination of functions in $\mc{H}_{d-1}$. Therefore, there are coefficients $c_{\al_1\dots\al_t\beta}\in\zR$ and constant $\kappa\in\zR$ so that  
     \begin{align}
     \label{eq:h_i-basis}
        \mathscr{M}=\sum\limits_{t=0}^{d-1}\sum\limits_{\substack{\al_i\in[p-1],\\\beta\in[p-2]}} c_{\al_1\dots\al_t\beta}\left(\bigoplus_{i=1}^t\al_ih_i\oplus h_{t+1}\oplus \beta\right) +\kappa
     \end{align}
     Recall that each $h_i$ is a $p$-expressions over variables $x_1,\dots,x_n$. By substituting back for each $h_i$ and rearranging terms, each $p$-expression $\bigoplus_{i=1}^t\al_ih_i\oplus h_{t+1}\oplus \beta$ is equivalent to $\al'_1x_1\oplus \al'_2x_2\oplus \cdots \oplus \al'_nx_n\oplus \beta'$ for some $\al'_1,\dots,\al'_n,\beta'\in [p-1]$. By \Cref{lem:p-exp-sum-to-basis}, such expression can be written as a $\zR$-linear combination of basis in $\cF_n$. 
     
     Note that number of $p$-expressions in \eqref{eq:h_i-basis} is at most $p^{d+1}$. By \Cref{lem:p-exp-sum-to-basis}, each $p$-expression can be written as a $\zR$-linear combination of basis in $\cF_n$ in time $O(np^{O(1)})$. Therefore, in time $O(np^{O(d)})$ we can write $\mathscr{M}=h_1\cdot h_2\cdots h_d$ as a $\zR$-linear combination of basis in $\cF_n$ which is polynomial in $n$ for fixed $p$ and $d$.
\end{proof}

\subsection{The correctness of the conversion algorithm}

Now we have enough ingredients to prove Algorithm~\ref{alg:FGLM} runs in polynomial time and correctly decides if $\reduce q {G_1} = \sum\limits_{j}k_j \reduce {b_j} {G_1}$ for every $q$.  In the following theorem, suppose $q$ is the current monomial considered by the algorithm. Furthermore, suppose sets $G_2$ and $B(G_2)$ have been constructed correctly so far.

\begin{theorem}\label{the:correctness-linear}
    Let $q = x_1^{\al_1}\cdots x_n^{\al_n}$ be a monomial of degree at most $d$. Suppose $q$ is not divisible by any leading monomial of polynomials in the current set $G_2$. Then, there exists a polynomial time algorithm that can decide whether $\reduce q {G_1} = \sum\limits_{j}k_j \reduce {b_j} {G_1}$ where $b_j$ are in the current set $B(G_2)$ in Algorithm~\ref{alg:FGLM}.
\end{theorem}

\begin{proof}
First, we discuss the case where for some $r< i$ we have $p-1<\al_i$. Set $q=q'\cdot x_i^{\al_i}$ where $q'=x_1^{\al_1}\cdots x_{i-1}^{\al_{i-1}}\cdot x_{i+1}^{\al_{i+1}}\cdots x_{n}^{\al_n}$. Then $\reduce q {G_1}=\reduce {q'} {G_1}\cdot \reduce {x_i^{\al_i}} {G_1}$. Note that $\reduce {x_i^{\al_i}} {G_1}$ is a linear combination of $x_i^{p-1},x_i^{p-2},\dots,x_i$. This is because the domain polynomial $\prod_{a\in\zZ_p}(x_{i}-a)$ is in $G_1$. Therefore,
\begin{align}
    \nonumber
    \reduce q {G_1}
    &=\reduce {q'} {G_1}\cdot \reduce {x_i^{\al_i}} {G_1}\\\nonumber
    &=\reduce {q'} {G_1}\cdot (c_{p-1}x_i^{p-1}+\dots+c_1x_i)\\
    &=\reduce {(q'\cdot c_{p-1}x_i^{p-1})} {G_1}+\dots+\reduce {(q'\cdot c_{1}x_i)} {G_1}\label{eq:q'-reduce-G1}
\end{align}
All of $q'\cdot x_i^{j}$ in \eqref{eq:q'-reduce-G1} have degree less than $q$, and hence they have been considered by Algorithm~\ref{alg:FGLM} before reaching $q$. Now, none of $q'\cdot x_i^{j}$ can be a multiple of the leading monomial of a polynomial in $G_2$ as otherwise $q$ divides the leading monomial of a polynomial in $G_2$. This implies that all $q'\cdot x_i^{j}$ with $c_j\neq 0$ in \eqref{eq:q'-reduce-G1} are in $B(G_2)$ and we have $\reduce q {G_1}=\sum\limits_{j}k_j \reduce {b_j} {G_1}$.

We continue by assuming for all $r< i$ we have $\al_i\leq p-1$. Note that if $\reduce q {G_1} = \sum\limits_{j}k_j \reduce {b_j} {G_1}$, then $\reduce q {G_1} - \sum\limits_{j}k_j \reduce {b_j} {G_1} = 0$ and hence $q - \sum\limits_{j}k_j  {b_j} \in \I(\cP)$. We proceed by checking, in a systematic way, if there exist coefficients $k_j$ such that $\reduce q {G_1} - \sum\limits_{j}k_j  \reduce {b_j} {G_1}=0$ holds. We will construct a system of linear equations over $\zR$ for coefficients $k_j$ so that this system has a solution if and only if $\reduce q {G_1} = \sum\limits_{j}k_j  \reduce {b_j} {G_1}$. 

By \Cref{reduced-to-G1}, we have $q\vert_{G_1}= f_1^{\al_1}\cdots f_n^{\al_n}$ where each $f_{i}$ is the  $p$-expression associated to $x_i$. Similarly, for each $b_j$ we have $\reduce {b_j} {G_1}= \mathscr{M}_j$ where $\mathscr{M}_j= h_{j1} \cdot h_{j2} \cdots h_{jd}$ is a multiplication of at most $d$ (not necessary distinct) $p$-expressions. 
    \begin{align}
    \label{linear-in-k}
       f_1^{\al_1}\cdots f_n^{\al_n} = \sum\limits_{j}k_j \reduce {b_j} {G_1} = \sum\limits_{j}k_j \mathscr{M}_j
    \end{align}

Recall that degree of $q$ is at most $d$ and, by \Cref{lem:p-exp-multiplication-to-basis}, $q$ can be written as a $\zR$-linear combination of the basis in $\cF_{n-1}$, say $\mathscr{L}_q$. Similarly, by \Cref{lem:p-exp-multiplication-to-basis}, each product $\mathscr{M}_j=h_{j1}\cdot h_{j2}\cdots h_{jd}$ can be written as a $\zR$-linear combination of the basis in $\cF_{n-1}$ in polynomial time. Therefore, \eqref{linear-in-k} is equivalent to 
     \begin{align}
        \label{linear-in-k3}
       f_1^{\al_1}\cdots f_n^{\al_n}=\mathscr{L}_q = \sum\limits_{j}k_j \reduce {b_j} {G_1} = \sum\limits_{j}k_j \mathscr{M}_j=\sum\limits_{j}k_j \mathscr{L}_j  
    \end{align}
     where each $\mathscr{L}_j$ is a $\zR$-linear combination equivalent to $\mathscr{M}_j$ via $p$-expression basis in $\cF_n$. Rearranging terms in $\sum\limits_{j}k_j \mathscr{L}_j$ and $\mathscr{L}_q$ yields
     \begin{align}
    \label{linear-in-k4}
       0 = \sum\limits_{j}k_j \mathscr{L}_j - \mathscr{L}_q=\sum\limits_{j}k'_j \mathscr{L}'_j 
    \end{align}
    where each $k'_j$ is linear combination of $k_j$s. Since $\cF_{n-1}$ is linearly independent and the left hand side of \eqref{linear-in-k4} is a constant we deduce that all $k'_j$ should be zero. Hence, we are left with (possibly more than one) linear equations with respect to $k_j$ over $\zR$. Note that at this point there is not any term with a $p$-expression. If such a system has a (unique) solution for $k_j$ then we conclude $\reduce q {G_1} = \sum\limits_{j}k_j \reduce {b_j} {G_1}$, else the equality does not hold. 
    
    Note that, by \Cref{lem:p-exp-multiplication-to-basis}, time complexity of finding a $\zR$-linear combination of basis in $\cF_{n-1}$ for a multiplication of $d$ $p$-expressions is $O(np^{O(d)})$. Moreover, we use \Cref{lem:p-exp-multiplication-to-basis} at most $O(n^{O(d)})$ many times. Hence, the whole process requires $O(n^{O(d)})$ time complexity.
\end{proof}

\begin{theorem}
\label{thm:lin-tractability}
 Let $\Gm$ be a constraint language where each relation in $\Gm$ is expressed as a system of linear equations modulo a prime number $p$. Then, the $\IMP_d(\Gm)$ can be solved in polynomial time for fixed $d$ and $p$.
\end{theorem}

\section{Proof of Theorem~\ref{the:lin-independent-MainBody}}

In this section we give a proof of Theorem~\ref{the:lin-independent-MainBody}. Also, it turns out in the case $p=3$ there is another basis of  $\cV_{n,p}$ that has a particularly clear structure. We give a proof of this result in Section~\ref{sec:3-independence}.

\subsection{Linear independence}

Recall that we consider the set $\cV_{k,p}$ of functions from $\zZ^k_p$ to $\zC$ as a $p^k$-dimensional vector space, whose components are values of the function at the corresponding point of $\zZ^k_p$. Let
\[
    F_k=\left\{\bigoplus_{i=1}^k\al_ix_i\oplus x_{k+1}\oplus \beta\mid \al_i\in\{0\zd p-1\}, 
    i\in[n], \beta\in\{0\zd p-2\}\right\}
\]
be a collection of linear functions over $\zZ_p$, and let 
\[
    \cF_k=F_k\cup\dots\cup F_0\cup\{1\}. 
\]

The result we are proving in this section is

\begin{theorem}\label{the:lin-independent}
For any $k$, the collection $\cF_k$ of $p$-expressions is linearly independent as a set of vectors from $\cV_{k+1,p}$ and forms a basis of $\cV_{k+1,p}$.
\end{theorem}

\subsubsection{Monomials and projections}

As a vector space the set $\cV_{k,p}$ has several natural bases. Let $\Mon(k,p)$ denote the set of all monomials $u=x_1^{d_1} \dots x_k^{d_k}$ where $d_i\in\{0\zd p-1\}$; and let $\cont(u)$ denote the set $\{i\mid d_i\ne0\}$. For $\oa\in\zZ_p^k$ we denote by $r_\oa$ the function given by $r_\oa(\oa)=1$ and $r_\oa(\ov x)=0$ when $\ov x\ne\oa$. We start with a simple observation.

\begin{lemma}\label{lem:2-bases}
The sets $\Mon(k,p)$ and $R(k,p)=\{r_\oa\mid \oa\in\zZ_p^k\}$ are bases of $\cV_{k,p}$.
\end{lemma}

As is easily seen, both sets contain $p^k$ elements, $R(k,p)$ obviously spans $\cV_{k,p}$. That $\Mon(k,p)$ also spans $\cV_{k,p}$ follows from polynomial interpolation properties.

Our goal here is to find yet another basis of $\cV_{k,p}$ suitable for our needs, which is the set $\cF_k$ defined above.

In this section we view functions from $\cF_k$ as elements of $\cV_{k+1,p}$. In particular, we will use coordinates of such functions in the bases $\Mon(k+1,p)$ and $R(k+1,p)$. The latter is of course just the collection of values of a function in points from $\zZ_p^{k+1}$, while the former is the polynomial interpolation of a function, which is unique when we restrict ourselves to polynomials of degree at most $p-1$ in each variable. 

The number of functions in $\cF_k$ is
\[
1+\sum_{\ell=0}^k|F_\ell|=1+\sum_{\ell=0}^kp^\ell(p-1)
=1+(p-1)\sum_{\ell=0}^kp^\ell=1+(p-1)\frac{p^{k+1}-1}{p-1}=p^{k+1}.
\]
Thus, we only need to prove that $\cF_k$ spans $\cV_{k+1,p}$.

We will need a finer partition of sets $F_k$: for $S\sse[k]$ let $F_k^S$
denote the set of functions $\bigoplus_{i=1}^k\al_ix_i\oplus x_{k+1}\oplus b$
such that $\al_i\ne0$ if and only if $i\in S$.

We prove by induction on $|S|$, $S\sse[k+1]$, that any monomial $u$ with $\cont(u)=S$ is in the span of 
\[
\cF_k^S=\bigcup_{\ell\le k, T\sse S\cap[\ell]}F_\ell^T.
\]
Since up to renaming the variables $\cF^S_k$ can be viewed as $\cF_{|S|}$, it suffices to prove the result for $S=[k]$, and our inductive process is actually on $k$. 

For $f\in F_k$ let $f'$ denote the sum of all monomials $u$ of $f$ (with the same coefficients) for which $\cont[u]=[k+1]$. In other words, $f'$ can be viewed as a projection $f$ onto the subspace $\cV_1$ spanned by $\Mon^*(k+1,p)=\{u\in\Mon(k+1,p)\mid \cont(u)=[k+1]\}$ parallel to the subspace $\cV_2$ spanned by $\Mon^\dagger(k+1,p)=\Mon(k+1,p)-\Mon^*(k+1,p)$. Let $F'_k=\{f'\mid f\in F_k\}$. Note that $\cV_1$ is also the subspace of $\cV_{n,p}$ spanned by the set $\{r_\oa\mid \oa\in(\zZ_p^*)^{k+1}\}$, and the dimensionality of $\cV_1$ is $(p-1)^{k+1}=|F_k|=|F'_k|$. Since by the induction hypothesis $f-f'$ is in the span of $\cF_k$, it suffices to prove that vectors in $F'_k$ are linearly independent, and therefore generate $\cV_1$. This will be proved in the rest of this section.

\begin{proposition}\label{pro:F-k-independence}
The set $F'_k$ is linearly independent.
\end{proposition}

We prove Proposition~\ref{pro:F-k-independence}  by constructing a matrix containing the values of functions from $F'_k$ and find its rank by finding all its eigenvalues. We do it in three steps. First, let $F^\dagger_k$ denote the superset of $F_k$ that apart from functions from $F_k$ also contains functions of the form $\bigoplus_{i=1}^k\al_i x_i\oplus x_{k+1}\oplus(p-1)$ for $\al_i\in\zZ^*_p$, $i\in[k]$. Then, let  $N_k$ be the $p(p-1)^k\tm p(p-1)^k$-dimensional matrix whose rows are labeled with $(\vc x{k+1})\in(\zZ^*_p)^k\tm\zZ_p$ representing values of the arguments of functions from $F^\dagger_k$,  and the columns are labeled with $f\in F^\dagger_k$. The entry of $N_k$ in row $(\vc x{k+1})$ and column $f$ is $f(\vc x{k+1})$. In the next section we find the eigenvectors and eigenvalues of $N_k$. In the second step we use the properties of $N_k$ to study the matrix $N''_k$ obtained from $N_k$ by replacing every entry of the form $f(\vc x{k+1})$ with the value $f''(\vc x{k+1})$, where $f''$ is the sum of all the monomials $u$ from $f$ with $\vc xk\in\cont(u)$. We again find the eigenvectors and eigenvalues of $N''_k$. Finally, we transform $N''_k$ to obtain a new matrix $N'_k$ in such a way that the entry of $N'_k$ in the row $(\vc x{k+1})$ with $x_{k+1}\ne0$ and column $f$ equals $f'(\vc x{k+1})$. We then finally prove that all the rows of $N'_k$ labeled  $(\vc x{k+1})$, $x_{k+1}\ne0$, are linearly independent.

\subsubsection{Kronecker sum and the eigenvalues of $N_k$}

We first introduce some useful notation.

Let $A,B$ be $q\tm r$ and $s\tm t$ matrices with entries in $\zZ_p$. The 
\emph{Kronecker sum} of $A$ and $B$ denoted $A\boxplus B$ is the 
$qs\tm rt$ matrix whose entry in row $is+i'$ and column $jt+j'$
equals $A(i,j)\oplus B(i',j')$. In other words, $A\boxplus B$ is defined 
the same way as Kronecker product, except using addition modulo $p$
rather than multiplication. We also use $A^{\boxplus k}$ to denote
$A\boxplus\dots\boxplus A$.

We consider two matrices, matrix $B_p$ essentially consists of values of 
unary functions $\al x$ on $[p-1]$, except that we rearrange its rows and 
columns as follows. Let $a$ be a primitive residue modulo $p$, that is, 
a generator of $\zZ_p^*$. Then
\[
B_p=\left(\begin{array}{cccccc}
1&a&a^2&a^3&\dots&a^{p-2}\\
a^{p-2}&1&a&a^2&\dots&a^{p-3}\\
a^{p-3}&a^{p-2}&1&a&\dots&a^{p-4}\\
\vdots&\vdots&\vdots&\vdots&&\vdots\\
a&a^2&a^3&a^4&\dots&1
\end{array}\right),
\]
where $a^i$ denotes exponentiation modulo $p$. Matrix $C_p$ is again 
the operation table of addition modulo $p$ with rearranged rows
\[
C_p=\left(\begin{array}{ccccc}
0&1&2&\dots&p-1\\
p-1&0&1&\dots&p-2\\
p-2&p-1&0&\dots&p-3\\
\vdots&\vdots&\vdots&&\vdots\\
1&2&3&\ldots&0
\end{array}\right).
\]
As is easily seen, the values of functions $\bigoplus_{i=1}^k\al_ix_i$,
$\vc\al k\in\zZ^*_p$ on $(\zZ^*_p)^k$ can be viewed as 
$B^{\boxplus k}$; and those of $\bigoplus_{i=1}^k\al_ix_i\oplus x_{k+1}\oplus b$,
$\vc\al k\in\zZ^*_p$, $b\in\zZ_p$, on $\vc xk\in\zZ^*_p$, $x_{k+1}\in\zZ_p$
can be represented as $N_k=C_p\boxplus B_p^{\boxplus k}$. 
Next we find the eigenvectors and eigenvalues of $N_k$.

Recall that a square matrix of the form 
\begin{equation}\label{equ:circulant}
A=\left(\begin{array}{ccccc}
a_1&a_2&a_3&\dots&a_n\\
a_n&a_1&a_2&\dots&a_{n-1}\\
a_{n-1}&a_n&a_1&\dots&a_{n-2}\\
\vdots&\vdots&\vdots&&\vdots\\
a_2&a_3&a_4&\dots&a_1
\end{array}\right)
\end{equation}
is called \emph{circulant}. The eigenvectors and eigenvalues of circulant
matrices are well known. We give a brief proof of the following fact for the
sake of completeness.

\begin{lemma}\label{lem:circulant}
The eigenvectors of the matrix $A$ in (\ref{equ:circulant}) have the 
form $\vec v_\xi=(1,\xi,\xi^2\zd\xi^{n-1})$ for $n$th roots of unity $\xi$. 
The eigenvalue corresponding to $\vec v_\xi$ is 
$\mu_\xi=a_1+a_2\xi+a_3\xi^2+\dots+a_n\xi^{n-1}$.
\end{lemma}

\begin{proof}
We compute $A\cdot\vec v_\xi$
\begin{eqnarray*}
A\cdot\vec v_\xi &=& \left(\begin{array}{ccccc}
a_1&a_2&a_3&\dots&a_n\\
a_n&a_1&a_2&\dots&a_{n-1}\\
a_{n-1}&a_n&a_1&\dots&a_{n-2}\\
\vdots&\vdots&\vdots&&\vdots\\
a_2&a_3&a_4&\dots&a_1
\end{array}\right)\cdot
\left(\begin{array}{c} 
1\\ \xi\\ \xi^2\\ \vdots\\ \xi^{n-1}
\end{array}\right)\\
&=& \left(\begin{array}{c} 
a_1+a_2\xi+a_3\xi^2+\dots+a_n\xi^{n-1}\\
a_n+a_1\xi+a_2\xi^2+\dots+a_{n-1}\xi^{n-1}\\
a_{n-1}+a_n\xi+a_1\xi^2+\dots+a_{n-2}\xi^{n-1}\\
\vdots\\
a_2+a_3\xi+a_4\xi^2+\dots+a_1\xi^{n-1}
\end{array}\right)\\
&=& (a_1+a_2\xi+a_3\xi^2+\dots+a_n\xi^{n-1})
\left(\begin{array}{c} 
1\\ \xi\\ \xi^2\\ \vdots\\ \xi^{n-1}
\end{array}\right),
\end{eqnarray*}
as required.
\end{proof}

A generalization of circulant matrices is obtained by replacing entries
of a circulant matrix by matrices. More precisely, a matrix of the 
form
\begin{equation}\label{equ:block-circulant}
A=\left(\begin{array}{ccccc}
A_1&A_2&A_3&\dots&A_n\\
A_n&A_1&A_2&\dots&A_{n-1}\\
A_{n-1}&A_n&A_1&\dots&A_{n-2}\\
\vdots&\vdots&\vdots&&\vdots\\
A_2&A_3&A_4&\dots&A_1
\end{array}\right),
\end{equation}
where $\vc An$ are square matrices of the same size is said to be 
\emph{block-circulant}. Note that $N_k$ can be viewed as a 
block-circulant matrix:
\[
N_k=C_p\boxplus B_p^{\boxplus k}=
\left(\begin{array}{cccc}
B_p^{\boxplus k}&B_p^{\boxplus k}\oplus1&
\dots&B_p^{\boxplus k}\oplus(p-1)\\
B_p^{\boxplus k}\oplus(p-1)&B_p^{\boxplus k}&
\dots&B_p^{\boxplus k}\oplus(p-2)\\
\vdots&\vdots&&\vdots\\
B_p^{\boxplus k}\oplus2&B_p^{\boxplus k}\oplus3&
\dots&B_p^{\boxplus k}
\end{array}\right).
\]
Similarly, $B_p^{\boxplus\ell}$ can also be viewed as a block-circulant matrix:
\[
B_p^{\boxplus(\ell-1)}\boxplus B_p=
\left(\begin{array}{cccc}
B_p^{\boxplus(\ell-1)}\oplus1&B_p^{\boxplus(\ell-1)}\oplus a&
\dots&B_p^{\boxplus(\ell-1)}\oplus a^{p-1}\\
B_p^{\boxplus(\ell-1)}\oplus a^{p-1}&B_p^{\boxplus(\ell-1)}\oplus 1&
\dots&B_p^{\boxplus(\ell-1)}\oplus a^{p-2}\\
\vdots&\vdots&&\vdots\\
B_p^{\boxplus(\ell-1)}\oplus a&B_p^{\boxplus(\ell-1)}\oplus a^2&
\dots&B_p^{\boxplus(\ell-1)}\oplus1
\end{array}\right).
\]

In some cases the eigenvectors and eigenvalues of block-circulant matrices
can also be found.

\begin{lemma}\label{lem:block-circulant}
Let $A$ be a block-circulant matrix with blocks $\vc An$ as in 
(\ref{equ:block-circulant}) such that $\vc An$ have the same eigenvectors.
Then every eigenvector $\vec w$ of $A$ has the form 
$\vec w_{v,\xi}=(1,\xi\zd \xi^{n-1})\otimes\vec v$, where $\xi$ is an 
$n$th root of unity and $\vec v$ is an eigenvector of $\vc An$. Conversely, 
for every eigenvector $\vec v$ of $\vc An$ and an $n$th root of unity $\xi$,
$\vec w_{v,\xi}$ is an eigenvector of $A$. The eigenvalue of $A$ associated
with $\vec w_{v,\xi}$ is $\mu_{1,\vec v}+\mu_{2,\vec v}\xi+
\mu_{3,\vec v}\xi^2+\dots+\mu_{n,\vec v}\xi^{n-1}$, where 
$\mu_{i,\vec v}$ is the eigenvalue of $A_i$ associated with $\vec v$.
\end{lemma}

\begin{proof}
As in the proof of Lemma~\ref{lem:circulant} we compute 
$A\cdot\vec w_{v,\xi}$:
\begin{eqnarray*}
A\cdot\vec w_{v,\xi} &=& \left(\begin{array}{ccccc}
A_1&A_2&A_3&\dots&A_n\\
A_n&A_1&A_2&\dots&A_{n-1}\\
A_{n-1}&A_n&A_1&\dots&A_{n-2}\\
\vdots&\vdots&\vdots&&\vdots\\
A_2&A_3&A_4&\dots&A_1
\end{array}\right)\cdot
\left(\begin{array}{c} 
\vec v\\ \xi\vec v\\ \xi^2\vec v\\ \vdots\\ \xi^{n-1}\vec v
\end{array}\right)\\
&=& \left(\begin{array}{c} 
A_1\vec v+\xi A_2\vec v+\xi^2A_3\vec v+\dots+\xi^{n-1}A_n\vec v\\
A_n\vec v+\xi A_1\vec v+\xi^2 A_2\vec v+\dots+\xi^{n-1}A_{n-1}\vec v\\
A_{n-1}\vec v+\xi A_n\vec v+\xi^2A_1\vec v+\dots+\xi^{n-1}A_{n-2}\vec v\\
\vdots\\
A_2\vec v+\xi A_3\vec v+\xi^2A_4\vec v+\dots+\xi^{n-1}A_1\vec v
\end{array}\right)\\
&=& (\mu_{1,\vec v}+\xi\mu_{2,\vec v}+\xi^2\mu_{3,\vec v}+\dots+
\xi^{n-1}\mu_{n,\vec v})
\left(\begin{array}{c} 
\vec v\\ \xi\vec v\\ \xi^2\vec v\\ \vdots\\ \xi^{n-1}\vec v
\end{array}\right),
\end{eqnarray*}
as required.
\end{proof}

We now apply these techniques to $N_k$. Let $\vec v(\eta)=(1,\eta,\eta^2\zd\eta^{p-1})$ for a $(p-1)$th root of unity, and let $\otimes$ denote Kronecker product.

\begin{lemma}\label{lem:B-eigenvectors}
\begin{itemize}
\item[(1)] 
For any $k$ and any $i,j\in\zZ_p$ the matrices $B_p^{\boxplus k}\oplus i$ 
and $B_p^{\boxplus k}\oplus j$ have the same eigenvectors, and every eigenvector has the form 
\[
\vec v(\vc\eta k)=\vec v(\eta_1)\otimes\dots\otimes\vec v(\eta_k)
\]
for some $(p-1)$th roots of unity.
\item[(2)] 
For any $k$ and any $j\in\zZ_p$ the eigenvalues of the matrix 
$B_p^{\boxplus k}\oplus j$ are
\[
\ld(\vc\eta k;j)=\sum_{\vc ik=0}^{p-2}(a^{i_1}
\oplus\dots\oplus a^{i_k}\oplus j)\eta_1^{i_1}\dots\eta_k^{i_k}
\]
for $(p-1)$th roots of unity $\vc\eta k$.
\item[(3)]
Let $\vc is\in[k]$ be such that if $\eta_i\ne1$ then $i=i_r$, $r\in[s]$, and $s\ne0$. Then
\[
\ld(\vc\eta k;j)=(-1)^{k-s}\ld(\eta_{i_1}\zd\eta_{i_k};j).
\]
\end{itemize}
\end{lemma}

\begin{proof}
We proceed by induction on $k$ to prove all three claims simultaneously. If $k=1$ 
then as $B_p\oplus j$ is a circulant
matrix whose first row is $(1\oplus j,a\oplus j,a^2\oplus j\zd a^{p-2}\oplus j)$,
by Lemma~\ref{lem:circulant} its eigenvalues have the form
\[
\ld(\eta_1)=(1\oplus j)+(a\oplus j)\eta_1+(a^2\oplus j)\eta_1^2+\dots+ 
(a^{p-2}\oplus j)\eta^{p-2}
\]
for a $(p-1)$th root of unity $\eta_1$, and the corresponding eigenvector 
is $\vec v(\eta_1)=(1,\eta_1,\eta_1^2\zd\eta_1^{p-1})$ regardless of $j$.

Now, suppose that the lemma is true for $k-1$. Also, suppose that 
every eigenvector of $B_p^{\boxplus (k-1)}\oplus j$ has the form
$\vec v(\eta_2\zd\eta_k)$ for $(p-1)$th roots of unity $\eta_2\zd\eta_k$. Then $B_p^{\boxplus k}\oplus j$ is a block-circulant matrix with the first row 
$(B_p^{\boxplus(k-1)}\oplus1\oplus j,B_p^{\boxplus(k-1)}\oplus a\oplus j,
B_p^{\boxplus(k-1)}\oplus a^2\oplus j\zd 
B_p^{\boxplus(k-1)}\oplus a^{p-1}\oplus j)$. 
As by the induction hypothesis the blocks in this row have the same 
eigenvectors, by Lemma~\ref{lem:block-circulant} the eigenvalues of 
$B_p^{\boxplus k}\oplus j$ have the form 
\[
\mu_{0,\vec v}+\mu_{1,\vec v}\eta_1+
\mu_{1,\vec v}\eta_1^2+\dots+\mu_{p-2,\vec v}\eta_1^{p-2},
\]
where $\vec v$ is an eigenvector of $B^{\boxplus(k-1)}$ and $\mu_{i,\vec v}$
is the eigenvalue of $B^{\boxplus(k-1)}\oplus a^i\oplus j$ associated with
$\vec v$. Thus, plugging in the inductive hypothesis we obtain the result.

To prove item (3) we need to inspect the case when $\eta_k=1$. In this case 
\begin{eqnarray*}
\ld(\vc\eta k) &=& \sum_{\vc ik=0}^{p-2}(a^{i_1}
\oplus\dots\oplus a^{i_k}\oplus j)\eta_1^{i_1}\dots\eta_k^{i_k}\\
&=&\sum_{\vc ik=0}^{p-2}(a^{i_1}
\oplus\dots\oplus a^{i_k}\oplus j)\eta_1^{i_1}\dots\eta_{k-1}^{i_{k-1}}\cdot1\\
&=& \sum_{\vc i{k-1}=0}^{p-2}\eta_1^{i_1}\dots\eta_{k-1}^{i_{k-1}}\left(\sum_{i_k=0}^{p-2}(a^{i_1}\oplus\dots\oplus a^{i_k}\oplus j)\right)\\
&=&\sum_{\vc i{k-1}=0}^{p-2}\eta_1^{i_1}\dots\eta_{k-1}^{i_{k-1}}\left(\frac{p(p-1)}2-(a^{i_1}\oplus\dots\oplus a^{i_{k-1}}\oplus j)\right)\\
&=&\frac{p(p-1)}2\sum_{\vc i{k-1}=0}^{p-2}\eta_1^{i_1}\dots\eta_{k-1}^{i_{k-1}}-\ld(\vc\eta{k-1})\\
&=& -\ld(\vc\eta{k-1}).
\end{eqnarray*}
The last equality is due to fact that 
\[
\sum_{\vc i{k-1}=0}^{p-2}\eta_1^{i_1}\dots\eta_{k-1}^{i_{k-1}}=
\sum_{i_1=0}^{p-2}\eta_1^{i_1}\cdot\dots\cdot\sum_{i_{k-1}=0}^{p-2}\eta_{k-1}^{i_{k-1}}=0.
\]
By the induction hypothesis the result follows.
\end{proof}

Since the matrix $N_k$ can be represented as $C_p\boxplus B_p^{\boxplus k}$,
its eigenvalues can be found by Lemma~\ref{lem:block-circulant}.

\begin{lemma}\label{lem:N-eigenvalues}
The eigenvalues of $N_k$ can be represented in one of the following forms.
\begin{itemize}
\item
for a $p$th root of unity $\xi$ and $(p-1)$th roots of unity $\vc\eta k$
\[
\mu(\vc\eta k;\xi)=\sum_{j=0}^{p-1}\xi^j \sum_{\vc ik=0}^{p-2}(a^{i_1}
\oplus\dots\oplus a^{i_k}\oplus j)\eta_1^{i_1}\dots\eta_k^{i_k}.
\]
\item
for a $p$th root of unity $\xi$ and $(p-1)$th roots of unity $\vc\eta k$
\[
\mu(\vc\eta k;\xi)=P(\xi)\cdot Q(\eta_1,\xi)\cdot\dots\cdot
Q(\eta_k,\xi),
\]
where 
$P(\xi)=\frac p{\xi-1}$ unless $\xi=1$, in which case $P(1)=\frac{p(p-1)}2$,
and 
\[
Q(\eta,\xi)=\sum_{i=0}^{p-2}\eta^i\xi^{a^i}.
\]
\end{itemize}
\end{lemma}

\begin{proof}
(1) By Lemma~\ref{lem:block-circulant} the eigenvalues of $N_k$ have the
form $\ld_{0,\vec v}+\ld_{1,\vec v}\xi+
\ld_{2,\vec v}\xi^2+\dots+\ld_{p-1,\vec v}\xi^{p-1}$,
where $\vec v$ is an eigenvector of $B_p^{\boxplus k}\oplus i$ for all 
$i\in\zZ_p$, and $\ld_{i,\vec v}$ is the eigenvalue of $B_p^{\boxplus k}\oplus i$
associated with $\vec v$, and $\xi$ is a $p$th root of unity. By 
Lemma~\ref{lem:B-eigenvectors} we obtain item (1) of the lemma.

(2) Consider the values $a^{i_1}\oplus\dots\oplus a^{i_k}\oplus j$
in the formula from part (1). For $j=0\zd p-1$ they constitute the set 
$\zZ_p$ regardless of $\vc ik$, and the sequence, when $j$ grows from 
0 to $p-1$, is a sequence of consequent residues modulo $p$. Therefore 
\[
\sum_{j=0}^{p-1} (a^{i_1}\oplus\dots\oplus a^{i_k}\oplus j)\xi^j 
= \xi^{a^{i_1}\oplus\dots\oplus a^{i_k}} 
\sum_{j=0}^{p-1}j\xi^j,
\]
and let $P(\xi)=\sum_{j=0}^{p-1}j\xi^j$. Therefore by part (1)
\begin{eqnarray*}
\mu(\vc\eta k;\xi) &=& \sum_{j=0}^{p-1}\xi^j \sum_{\vc ik=0}^{p-2}
(a^{i_1}\oplus\dots\oplus a^{i_k}\oplus j)\eta_1^{i_1}\dots
\eta_k^{i_k}\\
&=& \sum_{\vc ik=0}^{p-2} \xi^{a^{i_1}\oplus\dots
\oplus a^{i_k}} P(\xi) \eta_1^{i_1}\dots \eta_k^{i_k}\\
&=& P(\xi)\sum_{\vc ik=0}^{p-2} (\eta^{i_1}\xi^{a^{i_1}})
\cdot\dots\cdot(\eta_k^{i_k}\xi^{a^{i_k}})\\
&=& P(\xi)\left(\sum_{i_1=0}^{p-2} \eta^{i_1}\xi^{a^{i_1}}\right)
\cdot\dots\cdot\left(\sum_{i_k=0}^{p-2}\eta_k^{i_k}
\xi^{a^{i_k}}\right)\\
&=& P(\xi)\cdot Q(\eta_1,\xi)\cdot\dots\cdot
Q(\eta_k,\xi).
\end{eqnarray*}

Finally, we show that $P(\xi)$ has the required form. Since $\sum_{j=0}^{p-1}\xi^j=0$, we have $P(\xi)=P'(\xi)$, where $P'(x)=\sum_{j=0}^{p-1}(j+1) x^j$. Then
\begin{eqnarray*}
\sum_{j=0}^{p-1}(j+1) x^j 
&=& \frac d{dx}\left(\sum_{j=0}^{p-1} x^{j+1}\right)\\
&=& \frac d{dx}\left(\frac{x^{p+1}-x}{x-1}\right)\\
&=& \frac{((p+1)x^p-1)(x-1)-(x^{p+1}-x)}{(x-1)^2}\\
&=& \frac{px^{p+1}-(p+1)x^p +1}{(x-1)^2}.
\end{eqnarray*}
Since $\xi$ is a $p$th root of unity, if $\xi\ne1$ we have
\[
P(\xi)=\frac{p\xi-(p+1) +1}{(\xi-1)^2}=\frac p{\xi-1}.
\]
Finally, $P(1)=\frac{p(p-1)}2$, as is easily seen.
\end{proof}

Clearly, the co-rank of $N_k$ equals the multiplicity of the eigenvalue 0. Thus,
we need to find the number of combinations of $\xi,\vc\eta k$ such that
$\mu(\vc\eta k;\xi)=0$. For some of them it is easy.

\begin{lemma}\label{lem:N-zeroes}
If $\eta_i\ne1$ for some $i\in[k]$ then $\mu(\vc\eta k;1)=0$.
\end{lemma}

\begin{proof}
We use Lemma~\ref{lem:N-eigenvalues}. Let $\xi=1$, and, say, $\eta_1\ne1$.
Then 
\[
Q(\eta_1,1) = \sum_{i=0}^{p-2}\eta_1^i = \frac{\eta_1^{p-1}-1}{\eta_1-1}
= 0,
\]
as $\eta_1$ is a $(p-1)$th root of unity and $\eta_1\ne1$.
\end{proof}

\begin{lemma}\label{lem:polynomial}
Let $\xi\ne 1$ be a $p$th root of unity and $\eta$ a $(p-1)$th root of unity. 
Then $Q(\eta,\xi)\ne0$.
\end{lemma}

\begin{proof}
Let $\chi$ be a primitive $p(p-1)$th root of unity. Then $\eta,\xi$ can be represented as $\eta=\chi^{up}$, $\xi=\chi^{v(p-1)}$ and $Q(\eta,\xi)$ can be rewritten as 
\[
Q^*(\chi)=\sum_{j=1}^{p-1}\chi^{jup+a^jv(p-1)}.
\]
Note that all the arithmetic operations in the exponent including $a^j$ can be treated as regular ones rather than modular, as $\chi^b=\chi^c$ whenever $b\equiv c\pmod{p(p-1)}$. Therefore if there are $\eta,\xi$, $\xi\ne1$ such that $Q(\eta,\xi)=0$, then there exists a primitive $p(p-1)$th root of unity $\chi$ that is also a root of the polynomial
\[
Q^*(x)=\sum_{j=1}^{p-1}x^{jup+a^jv(p-1)}.
\]
This means that $Q^*(x)$ is divisible by $p(p-1)$ cyclotomic polynomial $C_{p(p-1)}$. The degree of $C_{p(p-1)}$ equals $\varphi(p(p-1))$, where $\varphi$ is Euler's totient function. In particular, the degree of $C_{p(p-1)}$ is divisible by $p-1$, and so is the degree of $Q^*$. Since $a$ and $p-1$ are relatively prime with $p$, it is only possible iv $u$ is divisible by $p-1$, that is, $\eta=1$, in which case, as is easily seen, $Q(1,\xi)=-1$ if $\xi\ne1$ and $Q(1,1)=p-1$. 
\end{proof}

The next proposition follows from Lemma~\ref{lem:polynomial} and the 
observation that $P(\xi)\ne0$ whenever $\xi$ is a $p$th root of unity.

\begin{proposition}\label{pro:N-rank}
The rank of $N_k$ is $(p-1)^k+1$.
\end{proposition}

\subsubsection{Changed matrices}

In this subsection we make the second step in our proof.

\begin{lemma}\label{lem:concrete-projections}
Let $f=\bigoplus_{i=1}^k\al_ix_i\oplus x_{k+1}\oplus b$ then 
\begin{equation}
f''=\sum_{S\sse[k]}(-1)^{k-|S|}\left(\bigoplus_{i\in S}\al_ix_i\oplus x_{k+1}\oplus b
\right). \label{equ:f-prime}
\end{equation}
\end{lemma}

\begin{proof}
We need to show that $f''(\vc x{k+1})=0$ whenever $x_i=0$ for some $i\in[k]$. In order to do that observe that the terms in (\ref{equ:f-prime}) can be paired up so that every $S$ containing $i$ is paired with $S-\{i\}$. Then $\bigoplus_{i\in S}\al_ix_i\oplus x_{k+1}\oplus b$ and $\bigoplus_{i\in S-\{i\}}\al_ix_i\oplus x_{k+1}\oplus b$ appear in (\ref{equ:f-prime}) with opposite signs, and, as $x_i=0$ are equal.
\end{proof}

Let $N''_k$ denote the matrix constructed the same way as $N_k$ only with $f''$, $f\in F^\dagger$, in place of $f$. More precisely, $N''_k$ is the $p(p-1)^k\tm p(p-1)^k$-dimensional matrix whose rows are labeled with $(\vc x{k+1})\in(\zZ^*_p)^k\tm\zZ_p$ representing values of the arguments of functions from $F^\dagger_k$,  and the columns are labeled with $f\in F^\dagger_k$. The entry of $N''_k$ in row $(\vc x{k+1})$ and column $f$ is $f''(\vc x{k+1})$. 

Using Lemma~\ref{lem:concrete-projections} we represent $N''_k$ as a sum of matrices. Let $f=\bigoplus_{i=1}^k\al_ix_i\oplus x_{k+1}\oplus b\in F^\dagger_k$ and $S\sse[k]$. Then let $f_S$ denote the function $\bigoplus_{i\in S}\al_ix_i\oplus x_{k+1}\oplus b$. In other words, by Lemma~\ref{lem:concrete-projections}
\[
f=\sum_{S\sse[k]}(-1)^{k-|S|}f_S.
\]
By $N_k(S)$, $S\sse[k]$, we denote the matrix constructed in a similar way to $N_k$ and $N''_k$. Again, its rows are labeled with $(\vc x{k+1})\in(\zZ^*_p)^k\tm\zZ_p$,  and the columns are labeled with $f\in F^\dagger_k$. The entry of $N_k(S)$ in row $(\vc x{k+1})$ and column $f$ is $f_S(\vc x{k+1})$. It is now easy to see that 
\[
N''_k=\sum_{S\sse[k]}(-1)^{k-|S|}N_k(S).
\]

In order to determine the structure of $N_k(S)$ we need one further observation. Let $\zzero_\ell$ denote the square $\ell$-dimensional matrix whose entries are all 0. Note that for a matrix $B$ and $\zzero_\ell$ 
\[
B\boxplus\zzero_\ell=
\left(\begin{array}{ccc} B&\dots&B\\ \vdots&&\vdots\\ B&\dots&B\end{array}\right).
\]

\begin{lemma}\label{lem:repeated-eigen}
Let $B$ an $n$-dimensional diagonalizable matrix. Then 
the eigenvectors of $B\boxplus\zzero_\ell$ are of the form $(\beta_1\vec v\zd\beta_\ell \vec v)$ where $v$ is an eigenvector of $B$ and either $\beta_1=\dots=\beta_\ell$ or $\beta_1+\dots+\beta_\ell=0$. The corresponding eigenvalue in the former case is $\ell\ld$, where $\ld$ is the eigenvalue of $B$ associated with $\vec v$, and 0 in the latter case.
\end{lemma}

The following lemma establishes the structure of $N_k(S)$, its eigenvalues and eigenvectors.

\begin{lemma}\label{lem:N-k-S-structure}
\begin{itemize}
\item[(a)]
$N_k(S)=C_p\boxplus D_1\boxplus\dots\boxplus D_k$, where 
\[
D_i=\left\{ \begin{array}{ll}
B_p, &  \text{if $i\in S$},\\ 
\zzero_{p-1}, & \text{otherwise}.
\end{array}\right.
\]
\item[(b)] 
Every eigenvector of $N_k$ is also an eigenvector of $N_k(S)$.
\item[(c)]
The eigenvalue $\mu(\vc\eta;\xi;S)$ of $N_k(S)$ associated with eigenvector $\vec v(\vc\eta k;\xi)$ equals
\[
\mu(\vc\eta;\xi;S)=\left\{\begin{array}{ll}
0 & \text{if $\eta_i\ne1$ for some $i\in[k]-S$},\\
\mu(1\zd1;1),& \text{if $\xi=\eta_1=\dots=\eta_k=1$},\\
(1-p)^{|[k]-S|}\mu(\vc\eta k;\xi), & \text{otherwise}. 
\end{array}\right.
\]
\end{itemize}
\end{lemma}

\begin{proof}
We will construct the matrix $N_k(S)$ inductively and prove the three claims of the lemma as we go. Let $N_k(S,\ell)$, $\ell\le k$, denote the $(p-1)^\ell\tm(p-1)^\ell$-matrix whose rows are labeled with $(\vc x\ell)\in(\zZ^*_p)^\ell$, columns are labeled with functions $f=\bigoplus_{i=1}^\ell\al_ix_i$. The entry of $N_k(S,\ell)$ in row $(\vc x\ell)$ and column $f$ is $f_S(\vc x\ell)$. We show that
\begin{itemize}
\item[(a')]
$N_k(S,\ell)=D_1\boxplus\dots\boxplus D_\ell$, where the $D_i$'s are defined as in the lemma.
\item[(b')] 
Every vector of the form $\vec v(\vc\eta\ell)=\vec v(\eta_1)\otimes\dots\otimes\vec v(\eta_\ell)$, where $\eta_i$ is a $(p-1)$th root of unity is an eigenvector of $N_k(S,\ell)\oplus j$ for $j\in\zZ_p$.
\item[(c')]
The eigenvalue $\mu(\vc\eta\ell,S,j)$ of $N_k(S,\ell)\oplus j$ associated with eigenvector $\vec v(\vc\eta\ell)$ equals 
\begin{eqnarray*}
\lefteqn{\mu(\vc\eta\ell,S,j)}\\
&=& \left\{\begin{array}{ll}
(p-1)^\ell\cdot j & \text{if $[\ell]\cap S=\emptyset$ and $\eta_1=\dots=\eta_\ell$},\\
0, & \text{if $\eta_i\ne1$ for some $i\in[\ell]-S$},\\
(p-1)^{|[\ell]-S|}\ld(\eta_{i_1}\zd\eta_{i_t};j), & \text{otherwise (see Lemma~\ref{lem:B-eigenvectors}),}\\
& \text{where $\{\vc jt\}=[\ell]\cap S$}.
\end{array}\right.
\end{eqnarray*}
\end{itemize}

If $\ell=1$ then either $N_k(S,1)=B_p$ if $1\in S$, or $N_k(S,1)=\zzero_{p-1}$ if $1\not\in S$. In the former case we have the result by Lemma~\ref{lem:B-eigenvectors}, and in the latter case by Lemma~\ref{lem:repeated-eigen} every vector of the form $\vec v(\eta)$, $\eta$ is a $(p-1)$th root of unity, is an eigenvector with eigenvalue $(p-1)j$ if $\eta=1$ and 0 otherwise. 

Suppose the statement is true for some $\ell$. If $[\ell+1]\cap S=\emptyset$, the claim is straightforward, as $f_S(\vc x{\ell+1})=0$ for all $\vc x{\ell+1}\in\zZ^*_p$, and the result follows by Lemma~\ref{lem:repeated-eigen}. 

Next, suppose that $[\ell]\cap S\ne\emptyset$, but $\ell+1\not\in S$. In this case the entry of $N_k(S,\ell+1)$ indexed with row $(\vc x{\ell+1})$ and column $f=\bigoplus_{i=1}^{\ell+1}\al_ix_i$ is $f_S(\vc x{\ell+1})=f^*_S(\vc x\ell)=\bigoplus_{i\in S\cap[\ell]}\al_ix_i$, where $f^*=\bigoplus_{i=1}^\ell\al_ix_i$.This implies that $N_k(S,\ell+1)\oplus j=(N_k(S,\ell)\oplus j)\boxplus\zzero_{p-1}$. By Lemma~\ref{lem:repeated-eigen} the eigenvectors of $N_k(S,\ell+1)\oplus j$ are of the two types: $\vec v'=(\vec v\zd\vec v)$ or $(\beta_1\vec v\zd\beta_{p-1}\vec v)$ with $\beta_1+\dots+\beta_{p-1}=0$, where $\vec v$ is an eigenvector of $N_k(S,\ell)\oplus j$. In the first case $\vec v'=\vec v\otimes\vec v(\eta)$ for $\eta=1$ and by the induction hypothesis has the required form. The corresponding  eigenvalue of $N_k(S,\ell+1)\oplus j$ equals $(p-1)\ld$, where $\ld$ is the eigenvalue of $N_k(S,\ell)\oplus j$ associated with $\vec v$, and so also has the required form. In the latter case $v\otimes(1,\eta\zd\eta^{p-2})$, $\eta\ne1$ and $\vec v$ is an eigenvector of $N_k(S,\ell)\oplus j$, satisfies the condition $1+\eta+\dots+\eta^{p-2}=0$ and has eigenvalue 0. By the induction hypothesis and Lemma~\ref{lem:B-eigenvectors}(3) we get the result.

Finally, let $\ell+1\in S$. In this case for any $\vc x{\ell+1}\in\zZ^*_p$ and $f=\bigoplus_{i=1}^{\ell+1}\al_ix_i$ the entry of $N_k(S,\ell+1)$ equals
\[
f_S(\vc x{\ell+1})=\bigoplus_{i=1}^{\ell+1}\al_ix_i=f^*_S(\vc x\ell)\oplus\al_{\ell+1}x_{\ell+1},
\] 
which implies $N_k(S,\ell+1)=N_k(S,\ell)\boxplus B_p$ proving (a'). Therefore $N_k(S,\ell+1)\oplus j$ is a block-circulant matrix and we can apply Lemma~\ref{lem:block-circulant} to show that eigenvectors of $N_k(S,\ell+1)$ are of the form 
\[
\vec v\otimes(1,\eta\zd\eta^{p-2})=\vec v\otimes\vec v(\eta),
\]
where $\eta$ is a $(p-1)$th root of unity and $\vec v$ is any eigenvector of $N_k(S,\ell)\oplus j$. The eigenvalue of such a vector can be found using the inductive hypothesis and the last part of the proof of Lemma~\ref{lem:B-eigenvectors} as follows. We have
\[
\mu(\vc\eta{\ell+1},S,j)=\mu_{0,\vec v}+\mu_{1,\vec v}\eta_{\ell+1}+
\mu_{1,\vec v}\eta_{\ell+1}^2+\dots+\mu_{p-2,\vec v}\eta_{\ell+1}^{p-2},
\]
where $\vec v$ is an eigenvector of $N_k(S,\ell)$ and $\mu_{i,\vec v}$
is the eigenvalue of $N_k(S,\ell)\oplus a^i\oplus j$ associated with
$\vec v$. If there is $i\in S\cap[\ell]$ such that $\eta_i\ne1$, then $\mu(\vc\eta{\ell+1},S,j)=0$. If $S\cap[\ell]=\emptyset$ and $\eta_1=\dots=\eta_\ell=1$ then 
\[
\mu(\vc\eta{\ell+1},S,j)=\sum_{i=0}^{p-2}(p-1)^\ell(a^i\oplus j)\eta_{\ell+1}^i=(p-1)^\ell\ld(\eta_{\ell+1};j).
\]
If $S\cap[\ell]=\{\vc it\}\ne\emptyset$, then by the induction hypothesis
\begin{eqnarray*}
\mu(\vc\eta{\ell+1},S,j) &=& 
\sum_{i=0}^{p-2}(p-1)^{|[\ell]-S|}\ld(\eta_{i_1}\zd\eta_{i_t};a^i\oplus j)\eta_{\ell+1}^i\\
&=& (p-1)^{|[\ell]-S|}\sum_{i=0}^{p-2}\sum_{\vc jt=0}^{p-2}(a^{j_1}\oplus\dots \oplus a^{j_t}\oplus a^i\oplus j)\eta_{i_1}\dots\eta_{i_t}\eta_{\ell+1}^i\\
&=& (p-1)^{|[\ell]-S|}\ld(\eta_{i_1}\zd\eta_{i_t},\eta_{\ell+1};j).
\end{eqnarray*}

We now consider the last step in constructing $N_k(S)$, from $N_k(S,k)$ to $N_k(S)$. As is easily seen, $N_k(S)=C_p\boxplus N_k(S,k)$, implying item (a) of the lemma, and by Lemma~\ref{lem:block-circulant} every vector of the form $(1,\xi\zd\xi^{p-1})\otimes\vec v$, where $\vec v$ is an eigenvector of $N_k(S,k)$ and $\xi$ is a $p$th root of unity is an eigenvector of $N_k(S)$. By the induction hypothesis this implies item (b) of the lemma. Finally, again by Lemma~\ref{lem:block-circulant} and the induction hypothesis the eigenvalue associated with the vector $(1,\xi\zd\xi^{p-1})\otimes\vec v(\vc\eta k)$ equals
\begin{eqnarray*}
\lefteqn{\mu(\vc\eta k;\xi;S)}\\
 &=& \mu(\vc\eta k,S,0)+\mu(\vc\eta k,S,1)\xi+\dots+\mu(\vc\eta k,S,p-1)\xi^{p-1}.
\end{eqnarray*}
If $S=\emptyset$, $\xi\ne1$, and $\eta_1=\dots=\eta_k=1$ then 
\begin{eqnarray*}
\mu(1\zd1;\xi;\emptyset) &=& \sum_{j=0}^{p-1}(p-1)^k j\xi^j=(p-1)^k P(\xi)\\
&=& (p-1)^k(-1)^k\mu(1\zd1;\xi)\\
&=& (1-p)^k\mu(1\zd1;\xi),
\end{eqnarray*}
as $Q(1,\xi)=-1$, as is easily seen. Also,
\[
\mu(1\zd1;1;\emptyset)=(p-1)^k P(1)=\mu(1\zd1;1).
\]
If $\eta_i\ne1$ for some $i\in[k]-S$ then $\mu(\vc\eta k,S,j)=0$, and so $\mu(\vc\eta k;\xi;S)=0$. Otherwise if $S=\{\vc is\}$,
\begin{eqnarray*}
\mu(\vc\eta k;\xi;S) 
&=& (p-1)^{|[k]-S|}\sum_{j=0}^{p-1}\ld(\eta_{i_1}\zd\eta_{i_s};j)\xi^j\\
&=& (p-1)^{|[k]-S|}\mu(\eta_{i_1}\zd\eta_{i_s};\xi)
\end{eqnarray*}
Finally, by Lemma~\ref{lem:B-eigenvectors} $\mu(\eta_{i_1}\zd\eta_{i_s};\xi)=(-1)^{|[k]-S|}\mu(\vc\eta k;\xi)$, if $\xi\ne1$, $\mu(\eta_{i_1}\zd\eta_{i_s};1)=0$ if $\eta_{i_j}\ne1$ for some $j$, and $\mu(1\zd1;1)=(p-1)^{|S|}P(1)$, and the result follows.
\end{proof}

Now, we are ready to find the eigenvalues of $N''_k$.

\begin{lemma}\label{lem:full-eigenvalue}
Let $\vec w=\vec v(\vc\eta k,\xi)$ and $T\sse[k]$ be such that $i\in T$ iff $\eta_i\ne1$.
Then
\[
\mu''(\vc\eta k;\xi)=\left\{\begin{array}{ll}
p^{k-|T|}\mu(\vc\eta k;\xi), & \text{if $\xi\ne1$},\\
0, & \text{if $\xi=1$ and $\eta_i\ne1$ for some $i\in[k]$}.
\end{array}\right.
\]
\end{lemma} 

\begin{proof}
Assume first that $\xi\ne1$. By Lemmas~\ref{lem:concrete-projections} and~\ref{lem:N-k-S-structure} we have
\begin{eqnarray*}
\mu''(\vc\eta k;\xi) &=& \sum_{S\sse[k]}(-1)^{|[k]-S|}\mu(\vc\eta k;\xi;S)\\
&=& \sum_{[k]\supseteq S\supseteq T}(-1)^{|[k]-S|}\mu(\vc\eta k;\xi;S)\\
&=& \sum_{\ell=0}^{k-|T|}(-1)^\ell{{k-|T|}\choose\ell}(1-p)^\ell\mu(\vc\eta k;\xi)\\
&=& \mu(\vc\eta k;\xi)\sum_{\ell=0}^{k-|T|}{{k-|T|}\choose\ell}(p-1)^\ell\\
&=& \mu(\vc\eta k;\xi) p^{k-|T|},
\end{eqnarray*}
as required.

Now let $\xi=1$. If $T\ne\emptyset$, then $\mu(\vc\eta;\xi;S)=0$ for any $S\sse[k]$. Otherwise, we have 
\begin{eqnarray*}
\mu'(\vc\eta k;\xi) &=& \sum_{T\sse[k]}(-1)^{k-|T|}\mu(1\zd1;1;T)\\
&=& \sum_{T\sse[k]}(-1)^{k-|T|}\mu(1\zd1;1)\\
&=& \mu(1\zd1;1)\sum_{\ell=0}^{k-1}(-1)^\ell{k\choose\ell}\\
&=& 0.
\end{eqnarray*}
\end{proof}

Since $f'$ consists of all the monomials $u$ of $f$ with $\vc x{k+1}\in\cont(U)$ and $f''$ consists of those with $\vc xk\in\cont(u)$, it is easy to see that
\begin{equation}\label{equ:f-prime-prime}
f'(\vc xk,x_{k=1})=f''(\vc xk,x_{k+1})-f''(\vc xk,0).
\end{equation}
Let $N'''_k$ denote the matrix obtained from $N''_k$ by subtracting the row labeled $(\vc xk,0)$ from every row labeled by $(\vc xk,x_{k+1})$, $x_{k+1}\in\zZ^*_p$. By (\ref{equ:f-prime-prime}) the rows of $N'''_k$ labeled $(\vc xk,x_{k+1})$, $x_{k+1}\in\zZ^*_p$ contain the values of $f'(\vc xk,x_{k+1})$. Let $N'_k$ be the submatrix of $N'''_k$ containing only such rows. We need to prove that the columns of $N'_k$ labeled with $f\in F_k$ are linearly independent. We do it by first proving that the rank of $N'_k$ equals $(p-1)^{k+1}$ and then demonstrating that the column labeled $f^{(p-1)}=\bigoplus_{i=1}^k\al_i x_i\oplus x_{k+1}\oplus(p-1)$ is a linear combination of columns labeled $f^{(b)}=\bigoplus_{i=1}^k\al_i x_i\oplus x_{k+1}\oplus b$ for $b\in\{0\zd p-2\}$.

\begin{lemma}\label{lem:linear-combinations}
Let $f=\bigoplus_{i=1}^k\al_i x_i\oplus x_{k+1}\oplus b$. 
\begin{itemize}
\item[(a)]
\[
\Sigma'f=\sum_{x_{k+1}=0}^{p-1}f''(\vc xk,x_{k+1})=0.
\]
\item[(b)]
Let $f^{(a)}$ denote the function $f^{(a)}=\bigoplus_{i=1}^k\al_i x_i\oplus x_{k+1}\oplus a$ (and so $f=f^{(b)}$). Then
\[
\Sigma f=\sum_{b\in\zZ_p}f''^{(b)}(\vc x{k+1})=0.
\]
\end{itemize}
\end{lemma}

\begin{proof}
(a) We have
\begin{eqnarray*}
\Sigma'f &=& \sum_{x_{k+1}=0}^{p-1}f''(\vc xk,x_{k+1})\\
&=& \sum_{x_{k+1}=0}^{p-1}\sum_{S\sse[k]}(-1)^{k-|S|}f_S(\vc x{k+1})\\
&=& \sum_{S\sse[k]}(-1)^{k-|S|}\sum_{x_{k+1}=0}^{p-1}f_S(\vc x{k+1}).
\end{eqnarray*}
Let $S\sse[k]$, $\vc xk\in\zZ^*_p$, and let us denote $A=f_S(\vc xk,0)=\bigoplus_{i\in S}^k\al_i x_i\oplus b$. Then 
\begin{eqnarray*}
\sum_{x_{k+1}=0}^{p-1}f_S(\vc x{k+1}) &=& \sum_{a=0}^{p-1}(A\oplus a)\\
&=& \frac{p(p+1)}2.
\end{eqnarray*}
Now, 
\[
\Sigma' f = \frac{p(p+1)}2\sum_{S\sse[k]}(-1)^{k-|S|}=0.
\]
\smallskip

(b) We have 
\begin{eqnarray*}
\Sigma f &=& \sum_{b\in\zZ_p}f''^{(b)}(\vc x{k+1})\\
&=& \sum_{b\in\zZ_p}\sum_{S\sse[k]}(-1)^{k-|S|}f^{(b)}_S(\vc x{k+1})\\
&=& \sum_{S\sse[k]}(-1)^{k-|S|}\sum_{b\in\zZ_p}f^{(b)}_S(\vc x{k+1}).
\end{eqnarray*}
Let $S\sse[k]$, $\vc x{k+1}\in\zZ^*_p$, and let us denote $A=\bigoplus_{i\in S}\al_i x_i\oplus x_{k+1}$. Then
\[
\sum_{b\in\zZ_p}f^{(b)}_S(\vc x{k+1})=\sum_{b\in\zZ_p}(A\oplus b)=\sum_{b\in\zZ_p}b=\frac{p(p+1)}2.
\]
Now, 
\[
\Sigma f = \frac{p(p+1)}2\sum_{S\sse[k]}(-1)^{k-|S|}=0.
\]
\end{proof}

We are now in a position to complete the proof of Proposition~\ref{pro:F-k-independence}. Let $\vec a(\vc xk,x_{k+1})$ denote the row of $N''_k$ labeled with $(\vc xk,x_{k+1})$. Then by Lemma~\ref{lem:linear-combinations}(a)
\[
\vec a(\vc xk,0)=-\sum_{b\in\zZ^*_p}\vec a(\vc xk,b),
\]
and the row of $N'''_k$ labeled with $(\vc xk,x_{k+1})$ is
\[
\vec b(\vc xk,x_{k+1})=\vec a(\vc xk,x_{k+1})+\sum_{b\in\zZ^*_p}\vec a(\vc xk,b).
\]
As is easily seen the row $\vec a(\vc xk,0)$ is still a linear combination of $\vec b(\vc xk,x_{k+1})$, $x_{k+1}\in\zZ^*_p$, implying that the rows of $N'''_k$ labeled $(\vc x{k+1})\in(\zZ_p^*)^{p+1}$ are linearly independent and $N'_k$ has rank $(p-1)^{k+1}$. Finally, by Lemma~\ref{lem:linear-combinations}(a) the columns of $N'_k$ labeled with $f\in F_k$ are also linear independent.

\subsection{Linear equations mod 3}
\label{sec:3-independence}
In this section we consider the case where $p=3$ and provide linearly independent $p$-expressions that span the space of functions from $\zZ_3^n$ to $\zC$. The $p$-expressions we consider here are different from the ones considered in \Cref{the:lin-independent-MainBody} and we prove they are linearly independent using somewhat a simpler approach. 

In this subsection set $p=3$ and $\oplus$, $\odot$ denote addition and multiplication 
modulo $3$, respectively. Let $x_1,\dots,x_n$ be variables that take values from the ternary domain $\{0,1,2\}$.  Here we prove that all the linear expressions of the form 
\[(a_1\odot x_1)\oplus(a_2\odot x_2)\oplus\dots\oplus(a_n\odot x_n)\] 
with $a_i\in \{0,1,2\}$ are linearly independent, except the zero expression. For instance, in the case where $n=1$, the following matrix has rank $2$ meaning that $x_1$ and $2\odot x_1$ are linearly independent.
\[
A=
\begin{blockarray}{cccc}
{\scriptstyle 0} & {\scriptstyle x_1} & {\scriptstyle 2\odot x_1} \\
\begin{block}{(ccc)c}
 0 & 0 & 0 & {\scriptstyle x_1=0} \\
 0 & 1 & 2 & {\scriptstyle x_1=1} \\
 0 & 2 & 1 & {\scriptstyle x_1=2} \\
\end{block}
\end{blockarray}
 \]

Now define sequence of matrices as follows. Set ${C}_1={A}$ and recursively define ${C}_n$ to be the following $3^n\times 3^n$ matrix
\[
C_n=
\left(
\begin{array}{ccc}
C_{n-1} & C_{n-1} & C_{n-1}\\
C_{n-1} & C_{n-1}\oplus \mb{1} & C_{n-1}\oplus \mb{2}\\
C_{n-1} & C_{n-1}\oplus \mb{2} & C_{n-1}\oplus \mb{1}
\end{array}
\right)
\]
\begin{observation}\label{obs:rank+1}
    For any real numbers $a,b\neq 0$ and any integer $n$ we have $\rank(aC_n+\mb{b})=\rank(C_n)+1$. 
\end{observation}
\begin{proof}
    Note that the first row and the first column of $C_n$ contain only zeros. That is 
    \[
        C_n=
        \left(
        \begin{array}{c|ccc}
        0 & 0 & \cdots & 0\\
        \hline
        0 & & & \\
        \vdots & & {B} & \\
        0 &  & &
        \end{array}
        \right)
    \]
    where ${B}$ is a $3^n-1\times 3^n-1$ matrix and has the same rank as $C_n$. Now, $\rank(aC_n+B)=\rank(C_n+\mb{\frac{b}{a}})$.
    \begin{align*}
        C_n+\mb{\frac{b}{a}}=
        \left(
        \begin{array}{c|ccc}
        \frac{b}{a} & \frac{b}{a} & \cdots & \frac{b}{a}\\
        \hline
        \frac{b}{a} & & & \\
        \vdots & & B+\mb{\frac{b}{a}} & \\
        \frac{b}{a} &  & &
        \end{array}
        \right)
        \to 
        \left(
        \begin{array}{c|ccc}
        \frac{b}{a} & \frac{b}{a} & \cdots & \frac{b}{a}\\
        \hline
        0 & & & \\
        \vdots & & B+\mb{0} & \\
        0 &  & &
        \end{array}
        \right)
        \to 
        \left(
        \begin{array}{c|ccc}
        \frac{b}{a} & 0 & \cdots & 0\\
        \hline
        0 & & & \\
        \vdots & & B & \\
        0 &  & &
        \end{array}
        \right)
    \end{align*}
    Hence, $\rank(aC_n+\mb{b})=\rank(C_n+\mb{\frac{b}{a}})=\rank(C_n)+1$.
\end{proof}
\begin{lemma}
    For any integer $n$, $C_n$ has rank $3^n-1$ i.e., all the linear expressions of the form $(a_1\odot x_1)\oplus(a_2\odot x_2)\oplus \dots\oplus(a_n\odot x_n)$ 
with $a_i\in \{0,1,2\}$ are linearly independent, except the zero expression.
\end{lemma}
\begin{proof}
    The proof is by induction. Clearly, for $n=1$, the matrix $C_1={A}$ has rank 2. Suppose $C_i$ has rank $3^i-1$ for all $1\leq i\leq n$. For a matrix ${M}$ with $0,1,2$ entries, let \begin{align*}
        p_1({M})&=\frac{3}{2}{M}\circ {M}+\frac{5}{2}{M}+\mb{1}\\
        p_2({M})&=-\frac{3}{2}{M}\circ {M}-\frac{7}{2} +\mb{2}
    \end{align*}
    where $\circ$ denotes the \emph{Hadamard} product or the \emph{element-wise} product of two matrices. Observe that ${M}\oplus \mb{1}=p_1({M})$ and ${M}\oplus \mb{2}=p_2({M})$. Hence, we can write $C_{n+1}$ as follow
    \[
    C_{n+1}=
    \left(
    \begin{array}{ccc}
        C_{n} & C_{n} & C_{n}\\
        C_{n} & p_1(C_{n}) & p_2(C_{n})\\
        C_{n} & p_2(C_{n}) & p_1(C_{n})
    \end{array}
    \right)
    \]
    
    Next, we perform a series of row and column operations to transform $C_{n}$ into a block-diagonal matrix.
    \begin{align*}
        C_{n+1}&=
    \left(
    \begin{array}{ccc}
        C_{n} & C_{n} & C_{n}\\
        C_{n} & p_1(C_{n}) & p_2(C_{n})\\
        C_{n} & p_2(C_{n}) & p_1(C_{n})
    \end{array}
    \right)
    \to 
    \left(
    \begin{array}{ccc}
        C_{n} & C_{n} & C_{n}\\
        \mb{0} & p_1(C_{n})-C_{n} & p_2(C_{n})-C_{n}\\
        \mb{0} & p_2(C_{n})-C_{n} & p_1(C_{n})-C_{n}
    \end{array}
    \right)
     \\
     &\to
    \left(
    \begin{array}{ccc}
        C_{n} & \mb{0} & \mb{0}\\
        \mb{0} & p_1(C_{n})-C_{n} & p_2(C_{n})-C_{n}\\
        \mb{0} & p_2(C_{n})-C_{n} & p_1(C_{n})-C_{n}
    \end{array}
    \right)\\
    &\to 
    \left(
    \begin{array}{ccc}
        C_{n} & \mb{0} & \mb{0}\\
        \mb{0} & p_1(C_{n})-C_{n} & p_2(C_{n})-C_{n}\\
        \mb{0} & p_1(C_{n})+p_2(C_{n})-2C_{n} & p_1(C_{n})+p_2(C_{n})-2C_{n}
    \end{array}
    \right)
    \\
    &\to 
    \left(
    \begin{array}{ccc}
        C_{n} & \mb{0} & \mb{0}\\
        \mb{0} & p_1(C_{n})-C_{n} & p_2(C_{n})-C_{n}\\
        \mb{0} & -3C_n+\mb{3} & -3C_n+\mb{3}
    \end{array}
    \right)\\
    &\to 
    \left(
    \begin{array}{ccc}
        C_{n} & \mb{0} & \mb{0}\\
        \mb{0} & p_1(C_{n})-C_{n} & p_1(C_{n})+p_2(C_{n})-2C_{n}\\
        \mb{0} & -3C_n+\mb{3} & -6C_n+\mb{6}
    \end{array}
    \right)
    \\
    &\to 
    \left(
    \begin{array}{ccc}
        C_{n} & \mb{0} & \mb{0}\\
        \mb{0} & p_1(C_{n})-C_{n} & -3C_n+\mb{3}\\
        \mb{0} & -3C_n+\mb{3} & -6C_n+\mb{6}
    \end{array}
    \right)
    \to 
    \left(
    \begin{array}{ccc}
        C_{n} & \mb{0} & \mb{0}\\
        \mb{0} & p_1(C_{n})-C_{n}-\frac{1}{2}(-3C_n+\mb{3}) & -3C_n+\mb{3}\\
        \mb{0} & \mb{0} & -6C_n+\mb{6}
    \end{array}
    \right)
    \\
    &\to 
    \left(
    \begin{array}{ccc}
        C_{n} & \mb{0} & \mb{0}\\
        \mb{0} & p_1(C_{n})-C_{n}-\frac{1}{2}(-3C_n+\mb{3}) & \mb{0}\\
        \mb{0} & \mb{0} & -6C_n+\mb{6}
    \end{array}
    \right)\\
    &\to 
    \left(
    \begin{array}{ccc}
        C_{n} & \mb{0} & \mb{0}\\
        \mb{0} & -\frac{3}{2}C_n\circ C_n+3C_n -\mb{\frac{1}{2}}& \mb{0}\\
        \mb{0} & \mb{0} & 6C_n-\mb{6}
    \end{array}
    \right)
    \\
    &\to 
    \left(
    \begin{array}{ccc}
        C_{n} & \mb{0} & \mb{0}\\
        \mb{0} & -3C_n\circ C_n+6C_n -\mb{1}& \mb{0}\\
        \mb{0} & \mb{0} & C_n-\mb{1}
    \end{array}
    \right)
    \end{align*}
    Hence, rank of $C_{n+1}$ is $\rank(C_n)+\rank(C_n-\mb{1})+\rank(-3C_n\circ C_n+6C_n -\mb{1})$. Moreover, Observation~\ref{obs:rank+1} yields
    \begin{align*}
        \rank(C_{n+1})&= \rank(C_n)+\rank(C_n-\mb{1})+\rank(-3C_n\circ C_n+6C_n -\mb{1})\\
        &=3^n-1+3^n+\rank(-3C_n\circ C_n+6C_n -\mb{1})
    \end{align*}
    In what follows we prove that $\rank(-3C_n\circ C_n+6C_n -\mb{1})=3^n$. Let us define the following two matrices associated to a matrix $M$ with $\{0,1,2\}$ entries.
    \begin{align*}
    M^{\dagger}[i,j]=
    \begin{cases} 
      0 & \text{if } M[i,j]=0\\
      1 & \text{if } M[i,j]=1 \\
      0 & \text{if } M[i,j]=2 
   \end{cases}
   \quad \text{and} \quad 
  M^{\dagger\dagger}[i,j]=
    \begin{cases} 
      0 & \text{if } M[i,j]=0\\
      0 & \text{if } M[i,j]=1 \\
      2 & \text{if } M[i,j]=2 
   \end{cases}
\end{align*}
Note that $C_n=C_n^{\dagger}+C_n^{\dagger\dagger}$. For instance in the case $C_1={A}$ the two matrices ${A}^{\dagger}$ and ${A}^{\dagger\dagger}$ are ${A}^{\dagger}=\begin{pmatrix}
0 & 0 & 0\\
0 & 1 & 0\\
0 & 0 & 1
\end{pmatrix}$ and ${A}^{\dagger\dagger}=\begin{pmatrix}
0 & 0 & 0\\
0 & 0 & 2\\
0 & 2 & 0
\end{pmatrix}$. Here, we simplify the expression $-3C_n\circ C_n+6C_n -\mb{1}$ and write it in terms of $C_n^{\dagger}$ and $C_n^{\dagger\dagger}$.
\begin{align*}
       -3C_n\circ C_n+6C_n -\mb{1} 
       &= -3(C_n^{\dagger}+C_n^{\dagger\dagger})\circ(C_n^{\dagger}+C_n^{\dagger\dagger})+6(C_n^{\dagger}+C_n^{\dagger\dagger})-\mb{1}\\
       &= -3(C_n^{\dagger}\circ C_n^{\dagger}+C_n^{\dagger\dagger}\circ C_n^{\dagger\dagger})+6C_n^{\dagger}+6C_n^{\dagger\dagger}-\mb{1}\\
       &= -3(C_n^{\dagger}+2C_n^{\dagger\dagger})+6C_n^{\dagger}+6C_n^{\dagger\dagger}-\mb{1}\\
       &= -3C_n^{\dagger}-6C_n^{\dagger\dagger}+6C_n^{\dagger}+6C_n^{\dagger\dagger}-\mb{1}\\
       &= 3C_n^{\dagger}-\mb{1}
\end{align*}
\begin{claim}\label{claim:rank-3A-1}
    For every positive integer $n$, the matrix $3C_n^{\dagger}-\mb{1}$ has full rank. This implies that $-3C_n\circ C_n+6C_n -\mb{1}$ has full rank. 
\end{claim}
\begin{proof}
    For the base case $n=1$, the matrix $3 {A}^{\dagger}-\mb{1}=\begin{pmatrix}
    -1 & -1 & -1\\
    -1 & 2 & -1\\
    -1 & -1 & 2
    \end{pmatrix}$ has full rank i.e., $\rank(3 {A}^{\dagger}-\mb{1})=3$. For our induction hypothesis suppose the claim is correct for every $n$. Next we show $3C_{n+1}^{\dagger}-\mb{1}$ has rank $3^{n+1}$. 
    
    \begin{align*}
        3C_{n+1}^{\dagger} &=
    3\left(
    \begin{array}{ccc}
        C_{n}^{\dagger} & C_{n}^{\dagger} & C_{n}^{\dagger}\\
        C_{n}^{\dagger} & (C_{n}\oplus \mb{1})^{\dagger} & (C_{n}\oplus \mb{2})^{\dagger}\\
        C_{n}^{\dagger} & (C_{n}\oplus \mb{2})^{\dagger} & (C_{n}\oplus \mb{1})^{\dagger}
    \end{array}
    \right)-\mb{1}\\
    &= 
    \left(
    \begin{array}{ccc}
        3C_{n}^{\dagger}-\mb{1} & 3C_{n}^{\dagger}-\mb{1} & 3C_{n}^{\dagger}-\mb{1}\\
        3C_{n}^{\dagger}-\mb{1} & 3(C_{n}\oplus \mb{1})^{\dagger}-\mb{1} & 3(C_{n}\oplus \mb{2})^{\dagger}-\mb{1}\\
        3C_{n}^{\dagger}-\mb{1} & 3(C_{n}\oplus \mb{2})^{\dagger}-\mb{1} & 3(C_{n}\oplus \mb{1})^{\dagger}-\mb{1}
    \end{array}
    \right)
     \\
     &\to
    \left(
    \begin{array}{ccc}
        3C_{n}^{\dagger}-\mb{1} & 3C_{n}^{\dagger}-\mb{1} & 3C_{n}^{\dagger}-\mb{1}\\
        \mb{0} & 3(C_{n}\oplus \mb{1})^{\dagger}-3C_{n}^{\dagger} & 3(C_{n}\oplus \mb{2})^{\dagger}-3C_{n}^{\dagger}\\
        \mb{0} & 3(C_{n}\oplus \mb{2})^{\dagger}-3C_{n}^{\dagger} & 3(C_{n}\oplus \mb{1})^{\dagger}-3C_{n}^{\dagger}
    \end{array}
    \right)\\
    &\to 
    \left(
    \begin{array}{ccc}
        3C_{n}^{\dagger}-\mb{1} & \mb{0} & \mb{0}\\
        \mb{0} & 3(C_{n}\oplus \mb{1})^{\dagger}-3C_{n}^{\dagger} & 3(C_{n}\oplus \mb{2})^{\dagger}-3C_{n}^{\dagger}\\
        \mb{0} & 3(C_{n}\oplus \mb{2})^{\dagger}-3C_{n}^{\dagger} & 3(C_{n}\oplus \mb{1})^{\dagger}-3C_{n}^{\dagger}
    \end{array}
    \right)
    \\
    &\to 
    \left(
    \begin{array}{ccc}
        3C_{n}^{\dagger}-\mb{1} & \mb{0} & \mb{0}\\
        \mb{0} & (C_{n}\oplus \mb{1})^{\dagger}-C_{n}^{\dagger} & (C_{n}\oplus \mb{2})^{\dagger}-C_{n}^{\dagger}\\
        \mb{0} & (C_{n}\oplus \mb{2})^{\dagger}-C_{n}^{\dagger} & (C_{n}\oplus \mb{1})^{\dagger}-C_{n}^{\dagger}
    \end{array}
    \right)
    \end{align*}
    For a matrix ${M}$ with $0,1,2$ entries, define $p_1'({M})=\frac{1}{2}{M}\circ {M}-\frac{3}{2} {M}+\mb{1}$ and $p_2'({M})=\frac{1}{2} {M}\circ {M}-\frac{1}{2} {M}$ where $\circ$ denotes the Hadamard product or the element-wise product of two matrices. Observe that $({M}\oplus \mb{1})^{\dagger}=p_1'({M})$ and $({M}\oplus \mb{2})^{\dagger}=p_2'({M})$. 
    \begin{align*}
        p_1'({M})&=\frac{1}{2}({M}^{\dagger}+{M}^{\dagger\dagger})\circ({M}^{\dagger}+{M}^{\dagger\dagger})+\frac{3}{2}({M}^{\dagger}+{M}^{\dagger\dagger})+\mb{1} =-{M}^{\dagger}-\frac{1}{2}{M}^{\dagger\dagger}+\mb{1}\\
        p_2'({M})&=\frac{1}{2}({M}^{\dagger}+{M}^{\dagger\dagger})\circ({M}^{\dagger}+{M}^{\dagger\dagger})-\frac{1}{2}({M}^{\dagger}+{M}^{\dagger\dagger}) =\frac{1}{2}{M}^{\dagger\dagger}
    \end{align*}
    We continue by performing row and column operation to to transform $3C_{n+1}-\mb{1}$ into a block-diagonal matrix.
    \begin{align*}
       & \left(
    \begin{array}{ccc}
        3C_{n}^{\dagger}-\mb{1} & \mb{0} & \mb{0}\\
        \mb{0} & (C_{n}\oplus \mb{1})^{\dagger}-C_{n}^{\dagger} & (C_{n}\oplus \mb{2})^{\dagger}-C_{n}^{\dagger}\\
        \mb{0} & (C_{n}\oplus \mb{2})^{\dagger}-C_{n}^{\dagger} & (C_{n}\oplus \mb{1})^{\dagger}-C_{n}^{\dagger}
    \end{array}
    \right)\\
    &=
    \left(
    \begin{array}{ccc}
        3C_{n}^{\dagger}-\mb{1} & \mb{0} & \mb{0}\\
        \mb{0} & -2C_n^{\dagger}-\frac{1}{2}C_n^{\dagger\dagger}+\mb{1} & -C_n^{\dagger}+\frac{1}{2}C_n^{\dagger\dagger}\\
        \mb{0} & -C_n^{\dagger}+\frac{1}{2}C_n^{\dagger\dagger} & -2C_n^{\dagger}-\frac{1}{2}C_n^{\dagger\dagger}+\mb{1}
    \end{array}
    \right)\\
    &\to
    \left(
    \begin{array}{ccc}
        3C_{n}^{\dagger}-\mb{1} & \mb{0} & \mb{0}\\
        \mb{0} & -2C_n^{\dagger}-\frac{1}{2}C_n^{\dagger\dagger}+\mb{1} & -C_n^{\dagger}+\frac{1}{2}C_n^{\dagger\dagger}\\
        \mb{0} & -3C_{n}^{\dagger}+\mb{1} & -3C_{n}^{\dagger}+\mb{1}
    \end{array}
    \right)\\
    &\to
    \left(
    \begin{array}{ccc}
        3C_{n}^{\dagger}-\mb{1} & \mb{0} & \mb{0}\\
        \mb{0} & -C_n^{\dagger}-C_n^{\dagger\dagger}+\mb{1} & -C_n^{\dagger}+\frac{1}{2}C_n^{\dagger\dagger}\\
        \mb{0} & \mb{0} & -3C_{n}^{\dagger}+\mb{1}
    \end{array}
    \right)
    \to
    \left(
    \begin{array}{ccc}
        3C_{n}^{\dagger}-\mb{1} & \mb{0} & \mb{0}\\
        \mb{0} & -C_n^{\dagger}-C_n^{\dagger\dagger}+\mb{1} & -\frac{3}{2}C_n^{\dagger}+\mb{\frac{1}{2}}\\
        \mb{0} & \mb{0} & -3C_{n}^{\dagger}+\mb{1}
    \end{array}
    \right)\\
    &\to
    \left(
    \begin{array}{ccc}
        3C_{n}^{\dagger}-\mb{1} & \mb{0} & \mb{0}\\
        \mb{0} & -C_n^{\dagger}-C_n^{\dagger\dagger}+\mb{1} & \mb{0}\\
        \mb{0} & \mb{0} & -3C_{n}^{\dagger}+\mb{1}
    \end{array}
    \right)
    \to
    \left(
    \begin{array}{ccc}
        3C_{n}^{\dagger}-\mb{1} & \mb{0} & \mb{0}\\
        \mb{0} & C_n^{\dagger}+C_n^{\dagger\dagger}-\mb{1} & \mb{0}\\
        \mb{0} & \mb{0} & 3C_{n}^{\dagger}-\mb{1}
    \end{array}
    \right)
    \\
    &=
    \left(
    \begin{array}{ccc}
        3C_{n}^{\dagger}-\mb{1} & \mb{0} & \mb{0}\\
        \mb{0} & C_n-\mb{1} & \mb{0}\\
        \mb{0} & \mb{0} & 3C_{n}^{\dagger}-\mb{1}
    \end{array}
    \right)
    \end{align*}
    By the induction hypothesis, $\rank(3C_{n}^{\dagger}-\mb{1})=3^n$. Moreover, by the induction hypothesis and Observation~\ref{obs:rank+1} $\rank(C_n-\mb{1})=3^n$. As a result, $\rank(3C_{n+1}^{\dagger}-\mb{1})=3^n+3^n+3^n=3^{n+1}$ and $-3C_n\circ C_n+6C_n -\mb{1}$ has full rank.
\end{proof}
    Recall that $\rank(C_{n+1})=\rank(C_n)+\rank(C_n-\mb{1})+\rank(-3C_n\circ C_n+6C_n -\mb{1})$. By the induction hypothesis and Observation~\ref{obs:rank+1}, we have $\rank(C_n-\mb{1})=\rank(C_n)+1=3^n$. Moreover, Claim~\ref{claim:rank-3A-1} yields $\rank(-3C_n\circ C_n+6C_n -\mb{1})=3^n$. Hence,
    \[
        \rank(C_{n+1})=3^n-1+3^n+1+3^n+1=3^{n+1}-1.
    \]
\end{proof}

\end{document}